\documentclass[a4paper,11pt]{article} 
  
\usepackage[english]{babel}
\usepackage[a4paper]{geometry}
\usepackage{amsmath}
\usepackage{amsthm}
\usepackage{amssymb}

\def\bR{\mathbb{R}}

\def\bN{\mathbb{N}}

\def\bZ{\mathbb{Z}}

\def\cQ{\mathcal{Q}}
\def\cD{\mathcal{D}}
\def\cA{\mathcal{A}}
\def\cM{\mathcal{M}}

\def\cV{\mathcal{V}}
\def\cO{\mathcal{O}}
\def\cF{\mathcal{F}}
\def\cG{\mathcal{G}}
\def\cL{\mathcal{L}}
\def\cN{\mathcal{N}}
\def\cE{\mathcal{E}}
\def\cK{\mathcal{K}}
\def\cH{\mathcal{H}}
\def\eps{\varepsilon}
\def\ph{\varphi}

\def\wt{\widetilde}

\def\indic{\hbox{\raise-2pt \hbox{\indbf 1}}}

\let\==\equiv

\let\0=\noindent

\def\*{{\hfill\break\null\hfill\break}}

\def\tende#1{\,\vtop{\ialign{##\crcr\rightarrowfill\crcr
             \noalign{\kern-1pt\nointerlineskip}
             \hskip3.pt${\scriptstyle #1}$\hskip3.pt\crcr}}\,}
\def\otto{\,{\kern-1.truept\leftarrow\kern-5.truept\to\kern-1.truept}\,}

\newtheorem{theorem}{Theorem}[section]  
\newtheorem{cor}[theorem]{Corollary}
\newtheorem{prop}[theorem]{Proposition}
\newtheorem{lemma}[theorem]{Lemma}

\numberwithin{equation}{section}

\begin{document}

\title{The Excitation Spectrum of Bose Gases \\ Interacting Through Singular Potentials}

\author{Chiara Boccato$^1$, Christian Brennecke$^1$, Serena Cenatiempo$^2$, Benjamin Schlein$^1$ \\
\\
Institute of Mathematics, University of Zurich\\
Winterthurerstrasse 190, 8057 Zurich, 
Switzerland$^1$ \\
\\
Gran Sasso Science Institute, Viale Francesco Crispi 7 \\ 
67100 L'Aquila, Italy$^2$}

\maketitle

\begin{abstract}
We consider systems of $N$ bosons in a box with volume one, interacting through a repulsive two-body potential of the form $\kappa N^{3\beta-1} V(N^\beta x)$. For all $0 < \beta < 1$, and for sufficiently small coupling constant $\kappa > 0$, we establish the validity of Bogoliubov theory, identifying the ground state energy and the low-lying excitation spectrum up to errors that vanish in the limit of large $N$. 
\end{abstract}

\section{Introduction and main result}

We consider systems of $N$ bosons in the three dimensional box $\Lambda = [-1/2 ; 1/2]^{\times 3}$  with periodic boundary conditions. The Hamilton operator is given by 
\begin{equation}\label{eq:Ham0} H_N^\beta = \sum_{j=1}^N -\Delta_{x_j} + \frac{\kappa}{N} \sum_{i<j}^N N^{3\beta} V (N^{\beta} (x_i -x_j)) \end{equation}
for some $\beta \in [0;1]$ and acts on the Hilbert space $L^2_s (\Lambda^N)$, the subspace of $L^2 (\Lambda^N)$ consisting of all functions that are symmetric with respect to any permutation of the $N$ particles. In (\ref{eq:Ham0}), $V$ is a non-negative, compactly supported and spherically symmetric two-body potential. Later, we will require that $V \in L^3 (\bR^3)$ and that the coupling constant $\kappa > 0$ is sufficiently small. For $\beta = 0$, (\ref{eq:Ham0}) is the Hamilton operator of a Bose gas in the so-called mean-field limit. If $\beta = 1$, on the other end, we recover the Gross-Pitaevskii regime. 

In this paper, we are interested in low-energy states of (\ref{eq:Ham0}), i.e. the ground state and eigenstates with small excitation energy. The properties of low-energy states of dilute Bose gases have already been studied in the pioneering work of Bogoliubov, see \cite{Bog}. Bogoliubov wrote the Hamilton operator of a dilute Bose gas in Fock space, using standard creation and annihilation operators. Since low-energy states exhibit Bose-Einstein condensation, he proposed to replace creation and annihilation operators associated to the zero momentum mode by commuting numbers and then to neglect all cubic and quartic contributions to the resulting Hamilton operator. With this approximation, Bogoliubov obtained a quadratic Hamilton operator, which he was able to diagonalize finding precise expressions for the ground state energy and for excited eigenvalues. 

Mathematically, the validity of Bogoliubov theory has only been established for a limited number of systems, so far. In \cite{LSo}, Lieb and Solovej used a rigorous version of Bogoliubov's method to compute the ground state energy of bosonic jellium. Similar ideas were used by Giuliani and Seiringer in \cite{GiS} to show the Lee-Huang-Yang formula for the ground state energy of a Bose gas in a simultaneous limit of weak coupling and high density. In \cite{Sei}, Seiringer proved the validity of Bogoliubov theory for the Hamilton operator (\ref{eq:Ham0}) in the case $\beta =0$ (the mean field limit). More precisely, assuming $V$ to be positive definite, he proved that the ground state energy of (\ref{eq:Ham0}) for $\beta =0$   satisfies 
\begin{equation}\label{eq:GS0} E_N^{\beta=0} = \frac{(N-1)}{2} \kappa\widehat{V} (0) - \frac{1}{2} \sum_{p \in 2\pi \bZ^3 \backslash \{ 0 \}} \left[ p^2 + \kappa \widehat{V} (p) - \sqrt{|p|^4 + 2p^2 \kappa \widehat{V} (p)}\right] + o (1) \end{equation}
as $N\to\infty$. Furthermore, he showed that the spectrum of $H_N^{\beta=0} - E_N^{\beta=0}$ below a fixed threshold $\zeta >0$ is given, up to errors vanishing in the limit $N \to \infty$, by finite sums of the form 
\begin{equation}\label{eq:excit-zero} \sum_{p \in 2\pi \bZ^3 \backslash \{ 0 \}} n_p  \sqrt{|p|^4 + 2p^2 \kappa \widehat{V} (p)} \end{equation}
where $n_p \in \bN$ for all $p \in \Lambda^*_+ = 2\pi \bZ^3 \backslash \{ 0 \}$ (in fact, in \cite{Sei} the threshold $\zeta$ may grow with $N$; (\ref{eq:excit-zero}) remains true as long as $\zeta \ll N^{1/3}$).

The results of \cite{Sei} were extended by Grech and Seiringer to mean-field systems of bosons trapped by external potentials in  \cite{GS}. Results similar to (\ref{eq:GS0}), (\ref{eq:excit-zero}) for Bose gases in the mean-field limit were also obtained, with different approaches, by Lewin, Nam, Serfaty and Solovej in \cite{LNSS} and by Derezinski and Napiorkowski in \cite{DN}. Furthermore, in \cite{Piz1,Piz2,Piz3}, Pizzo obtained, for $\beta=0$ and imposing an ultraviolet cutoff, a convergent series expansion for the ground state of (\ref{eq:Ham0}) in powers of $N^{-1}$.

The goal of our paper is to extend the results (\ref{eq:GS0}), (\ref{eq:excit-zero}) and to prove the validity of Bogoliubov's prediction for the ground state energy and the excitation spectrum of (\ref{eq:Ham0}) in scaling limits with $0 < \beta < 1$, where the range of the interaction potential shrinks to zero, as $N \to \infty$. This is the content of our main theorem.
\begin{theorem}\label{thm:main}
Let $V \in L^3 (\bR^3)$ be non-negative, spherically symmetric, compactly supported and assume that the coupling constant $\kappa > 0$ is small enough. Fix $0 < \beta < 1$ and let $m_\beta \in \bN$ be the largest integer with $m_\beta \leq 1/(1-\beta) + \min (1/2 , \beta/ (1-\beta))$. Then, in the limit $N \to \infty$, the ground state energy $E_N^\beta$ of the Hamilton operator $H_N^\beta$ defined in (\ref{eq:Ham0}) is given by  
\begin{equation}
\begin{split}\label{1.groundstate}
E^\beta_{N} = \; &4\pi (N-1) a_N^\beta \\ &-\frac{1}{2}\sum_{p\in\Lambda^*_+} \left[ p^2+ \kappa \widehat{V} (0) - \sqrt{|p|^4 + 2 p^2 \kappa \widehat{V} (0)} - \frac{\kappa^2 \widehat{V}^2 (0)}{2p^2}\right] + \cO (N^{-\alpha})
    \end{split}
    \end{equation}      
for all $0< \alpha < \beta$ such that $\alpha \leq (1-\beta)/2$. Here we set $\Lambda_+^* = 2\pi \bZ^3 \backslash \{0 \}$ and we defined 
    \begin{equation} \label{bN}
    \begin{split}
    8\pi a_N^\beta = \; & \kappa \widehat{V} (0)-\frac{1}{2N}\sum_{p\in\Lambda^*_+}  \frac{\kappa^2 \widehat{V}^2 (p/N^\beta)}{p^2} \\
    &+\sum_{k=2}^{m_\beta}\frac{(-1)^{k}}{(2N)^{k}}\sum_{p\in\Lambda^*_+}\frac{ \kappa \widehat{V} (p/N^\beta)}{p^2} \\ &\quad \times \sum_{q_1, q_2, \dots, q_{k-1}\in\Lambda^*_+ }\frac{ \kappa \widehat{V} ((p-q_1)/N^\beta)}{q_1^2}\left(\prod_{i=1}^{k-2}\frac{ \kappa \widehat{V} ((q_i-q_{i+1})/N^\beta)}{q^2_{i+1}}\right)  \kappa\widehat{V} (q_{k-1}/N^\beta)
    \end{split}
    \end{equation}
Moreover, the spectrum of $H_N^\beta-E^\beta_{N}$ below an energy $\zeta$ consists of eigenvalues given, in the limit $N \to \infty$, by 
    \begin{equation}
    \begin{split}\label{1.excitationSpectrum}
    \sum_{p\in\Lambda^*_+} n_p \sqrt{|p|^4+2p^2  \kappa\widehat{V} (0)}+ \cO (N^{-\alpha} (1+ \zeta^3))
    \end{split}
    \end{equation}
for all $0< \alpha < \beta$ such that $\alpha \leq (1-\beta)/2$. Here $n_p \in \bN$ for all $p\in\Lambda^*_+$ and $n_p \not = 0$ for finitely many $p\in \Lambda^*_+$ only.  
\end{theorem}
{\it Remarks:}
\begin{itemize}
\item[1)] The sum over $p \in \Lambda^*_+$ appearing on the r.h.s. of (\ref{1.groundstate}) converges (a careful analysis shows that the expression in the parenthesis decays as $|p|^{-4}$ for large $|p|$). It gives therefore a contribution of order one to the ground state energy $E_N^\beta$.  
\item[2)] The r.h.s. of (\ref{bN}) is $N$ times a Born series expansion for the scattering length of the potential $\kappa N^{3\beta -1} V (N^\beta .)$. A simple computation shows that the $k$-th term in the sum on the r.h.s. of (\ref{bN}) (including the term on the first line, which is associated with $k=1$) is of the order $N^{k (\beta - 1)}$, for all $k \in \bN$. Hence, it gives a contribution to the ground state energy (\ref{1.groundstate}) of the order $N^{k\beta - (k-1)}$ which is negligible (vanishes, as $N \to \infty$) if $\beta < (k-1)/k$ or, equivalently, if $k > 1/(1-\beta)$. The truncation of the Born series at $k = m_\beta \simeq 1/(1-\beta) + \min (1/2 , \beta/(1-\beta))$ is chosen to make sure that the error we do in neglecting terms with $k > m_\beta$ is of the order $\mathcal{O} (N^{-\alpha})$, for all $\alpha \leq \min (\beta , (1-\beta)/2)$ (and therefore it is negligible, compared to other errors arising in the analysis). Notice that in the Gross-Pitaevskii regime $\beta=1$, which is not covered by Theorem \ref{thm:main}, the situation is different. In this case, all terms in the Born series are of order one; their sum reconstructs the scattering length $a_0$ of the unscaled potential $V$. In this sense, (\ref{1.groundstate}) is consistent with the results of Lieb and Yngvason in \cite{LY} and of Lieb, Seiringer and Yngvason in \cite{LSY} which imply that $E_N^{\beta=1} = 4 \pi a_0 N + o (N)$. 
\item[3)] As explained in \cite{Bog,LSSY,GiS}, the validity of Bogoliubov's approximation is restricted 
to regimes where the ratio $\mu = a / R$ between the scattering length $a$ of the interaction and its range $R$ is such that 
\begin{equation}\label{eq:cond-Bog} 1 \gg \mu \gg \sqrt{\rho a^3} \gg \mu^2 
\end{equation}
For the trapped gas described by (\ref{eq:Ham0}), we have $a \simeq N^{-1}$, $R \simeq N^{-\beta}$ and $\rho = N$. Hence, (\ref{eq:cond-Bog}) is only satisfied if $\beta < 1/2$. In other words, for \mbox{$1/2 \leq \beta < 1$}, Theorem \ref{thm:main} establishes the validity of the predictions of Bogoliubov's theory in regimes where Bogoliubov approximation fails (in fact, Bogoliubov theory is expected to hold in any dilute limit, with $\rho a^3 \ll 1$).  
\item[4)] It is worth noticing that the expression (\ref{1.excitationSpectrum}) for the excitation spectrum of (\ref{eq:Ham0}) has important consequences from the point of view of physics. It shows that the dispersion of bosons described by (\ref{eq:Ham0}) is linear for small momenta, in sharp contrast with the quadratic dispersion of free particles. This observation was used by Bogoliubov in \cite{Bog} to explain the emergence of superfluidity, via the so-called Landau criterion \cite{Lan}. 
\item[5)] Theorem \ref{thm:main} describes low-lying eigenvalues of the Hamiltonian (\ref{eq:Ham0}). As a corollary of (\ref{1.groundstate}) and (\ref{1.excitationSpectrum}), it is also possible to describe eigenvectors associated to low-lying eigenvalues of (\ref{eq:Ham0}). Referring to arguments from \cite{GS}, we provide a norm approximation of these eigenvectors in a remark at the end of Section \ref{sec:proof}, after the proof of Theorem \ref{thm:main}. 
\end{itemize}

To show Theorem \ref{thm:main}, we follow the strategy introduced in \cite{BBCS} to show complete Bose-Einstein condensation for low-energy states in the Gross-Pitaevskii regime (and also used in \cite{BS} to study the time-evolution of condensates in this scaling limit). We start with an idea from \cite{LNSS}. Every $\psi \in L^2_s (\Lambda^N)$ can be represented uniquely as 
\[ \psi = \sum_{n=0}^N \psi^{(n)} \otimes_s \ph_0^{\otimes (N-n)} \] for a sequence $\psi_n \in L^2_\perp (\Lambda)^{\otimes_s n}$, where $L^2_\perp (\Lambda)$ is the orthogonal complement of $\ph_0 (x) = 1$ in $L^2 (\Lambda)$ and $L^2_\perp (\Lambda)^{\otimes_s n}$ is the symmetric tensor product of $n$ copies of $L^2_\perp (\Lambda)$. This remark allows us to define a unitary map \begin{equation}\label{eq:Uph-def} 
U : L^2_s (\Lambda^N) \to\cF_+^{\leq N} = \bigoplus_{n=0}^N L^2_\perp (\Lambda)^{\otimes_s n}  \end{equation} through $U \psi = \{ \psi^{(0)} ,\dots , \psi^{(N)} \}$ ($\cF_+^{\leq N}$ is the Fock space over $L^2_\perp (\Lambda)$, truncated to exclude sectors with more than $N$ particles). The map $U$ factors out the condensate (particles described by the wave function $\ph_0$) and let us focus on its orthogonal excitations. Using $U$, we can define a first excitation Hamiltonian $\cL_N^\beta = U H_N^\beta U^* : \cF_+^{\leq N} \to \cF_+^{\leq N}$.  
Conjugating with $U$, we effectively extract, from the interaction term in the original Hamiltonian (\ref{eq:Ham0}), contributions to $\cL^\beta_N$ that are constant (c-numbers) and quadratic in creation and annihilation operators. This is very much in the spirit of the Bogoliubov approximation, where creation and annihilation operators involving the condensate mode 
$\ph_0$ are replaced by commuting numbers. 

In the mean-field regime (i.e. for $\beta = 0$), conjugation with $U$ is sufficient to extract all contributions of the many-body interaction whose expectation in low energy states survives, as $N \to \infty$. In other words, in the mean-field regime, the excitation Hamiltonian $\cL^{\beta =0}_N$ can be approximated by the sum of a constant and of a term  quadratic in creation and annihilation operators; the expected value of all other terms vanishes, as $N \to \infty$, when we consider low-energy states. 

For $\beta > 0$, the situation is different; some of the quartic terms in $\cL^\beta_N$ that were negligible for $\beta = 0$ are now important, in the limit $N \to \infty$. To better understand this point, let us observe that 
\[ \langle \Omega, \cL^{\beta}_N \Omega \rangle = \langle U^* \Omega, H_N^\beta U^* \Omega\rangle = \langle \ph_0^{\otimes N} , H_N^\beta \ph_0^{\otimes N} \rangle = \frac{(N-1) \widehat{V} (0)}{2} \]
According to (\ref{1.groundstate}), the difference between $\langle \Omega, \cL^{\beta}_N \Omega \rangle$ and the real ground state energy of (\ref{eq:Ham0}) is of the order $N^{\beta}$ and diverges, as 
$N \to \infty$. To make up for this error, we have to take into account correlations among particles. In \cite{BDS}, this goal was achieved by conjugating the excitation Hamiltonian  with a unitary Bogoliubov transformation of the form
\begin{equation}\label{eq:wtT} \wt{T} = \exp \left[ \frac{1}{2} \sum_{p \in 2\pi\bZ^3 , p \not = 0} \left( \eta_p a^*_p a^*_{-p} - \bar{\eta}_p a_p a_{-p} \right) \right] 
\end{equation}
for appropriate coefficients $\eta_p = \eta_{-p}$ (the context of \cite{BDS} was slightly different; it focused on the time-evolution for approximately coherent initial data on the Fock space). In (\ref{eq:wtT}), the operators $a_p^*$ and $a_p$ create and, respectively, annihilate a particle with momentum $p \in 2\pi \bZ^3$ (see Section \ref{sec:fock} below for precise definitions). While Bogoliubov transformations of the form (\ref{eq:wtT}) work well 
on the Fock space, they do not map the space $\cF_+^{\leq N}$ into itself (because they do not preserve the constraint $\cN \leq N$). 

To obviate this problem, we follow the strategy used in \cite{BS,BBCS}. We introduce modified creation and annihilation operators $b^*_p, b_p$ for all $p \in 2\pi \bZ^3$. The creation operator $b_p^*$ creates a particle with momentum $p$ but, at the same time, it removes a particle with momentum zero from the condensate. Similarly, $b_p$ annihilates a particle with momentum $p$ but, simultaneously, it creates a particles with momentum $0$ in the condensate. Hence, $b^*_p$ and $b_p$ create and annihilate excitations but they do not change the total number of particles in the system. As a consequence, when transformed with $U$, they map the Hilbert space $\cF^{\leq N}_+$ into itself. Using these operators, we can therefore define generalized Bogoliubov transformations of the form 
\begin{equation}\label{eq:T} T = \exp \left[ \frac{1}{2} \sum_{p \in 2\pi \bZ^3} \left( \eta_p b^*_p b^*_{-p} - \bar{\eta}_p b_p b_{-p} \right) \right] \end{equation}
with appropriate coefficients $\eta_p = \eta_{-p}$. In contrast with (\ref{eq:wtT}), generalized Bogoliubov transformations leave the space $\cF_+^{\leq N}$ invariant. This allows us to define a modified excitation Hamiltonian having the form $\cG_N^\beta = T^* U H_N U^* T : \cF_+^{\leq N} \to \cF_+^{\leq N}$. 

With the right definition of the coefficients $\eta_p$, 
we can show that the modified excitation Hamiltonian $\cG_N^\beta$ can be approximated by the sum of a constant and of a term quadratic in creation and annihilation operators. To be more precise, we first prove, as we recently did in \cite{BBCS} for the case $\beta =1$, that, for sufficiently small $\kappa > 0$, $\cG_N^\beta$ satisfies the lower bound
\begin{equation}\label{eq:cond0} \cG_N^\beta - E_N^\beta \geq \frac{1}{2} \cN_+ - C \end{equation}
where $\cN_+$ denotes the number of particles operator on $\cF_+^{\leq N}$. As we will show in Prop.~4.1,  Eq. (\ref{eq:cond0}) easily  implies that states $\psi_N \in L^2_s (\Lambda^N)$ with bounded excitation energy can be written as $\psi_N = U T \xi_N$, for an excitation vector $\xi_N \in \cF_+^{\leq N}$ satisfying 
\begin{equation}\label{eq:cond1} \langle \xi_N, \cN_+  \xi_N \rangle \leq C \end{equation}
with a constant $C > 0$ independent of $N$. In other words, (\ref{eq:cond0}) shows that low-energy states 
exhibit complete Bose-Einstein condensation in a very strong sense: the number of excitations, that is the number of particles that are not in the condensate, remains bounded, uniformly in $N$. Notice that Bose-Einstein condensation in the ground state of the Gross-Pitaevskii Hamiltonian (i.e. (\ref{eq:Ham0}) for $\beta =1$) has been known since the work \cite{LS} of Lieb and Seiringer; the novelty of (\ref{eq:cond1}) is the fact that it gives a bound, in fact an optimal bound, on the number of excitations. 

Combining (\ref{eq:cond0}) with the commutator estimate
\begin{equation}\label{eq:comm0} \pm \big[ \cG_N^\beta, i\cN_+ \big] \leq C (\cH_N^\beta +1) \end{equation}
where $\cH_N^\beta$ denotes the Hamiltonian (\ref{eq:Ham0}), rewritten as an operator on the Fock space, we show then that the excitation vector $\xi_N \in \cF^{\leq N}_+$ associated with a low-energy state also satisfy the stronger bound
\begin{equation}\label{eq:KN} \langle \xi_N, (\cN_+ +1) ( \cH_N^\beta +1) \xi_N \rangle \leq C 
\end{equation}
for a constant $C > 0$ independent of $N$. Eq.  (\ref{eq:KN}) does not only provide control on the expectation of the number of excitations, but also on their energy. It is worth noticing that, like for (\ref{eq:cond1}), the improved estimate (\ref{eq:KN}) does not require the assumption $\beta < 1$; it also holds true for the Gross-Pitaevskii Hamiltonian obtained with $\beta =1$.

Equipped with the bound (\ref{eq:KN}) we go back to the excitation Hamiltonian $\cG_N^\beta$ and we show that, in low-energy states, the expectation of all terms that are not constant or quadratic in creation and annihilation operators vanish, in the limit of large $N$. More precisely, we prove that
\begin{equation} \cG_N^\beta = C^\beta_N + \cQ_N^\beta + \cE_N^\beta \end{equation}
where $C_N^\beta$ is a constant (to leading order, the ground state energy of $H_N^\beta$), $\cQ_N^\beta$ is  quadratic in creation and annihilation operators, and $\cE_N^\beta$ is such that
\begin{equation}\label{eq:impr} \pm \cE_N^\beta \leq C N^{-\alpha} (\cN_+ +1) (\cH_N^\beta +1) \end{equation}
for all $0< \alpha < \beta$ such that $\alpha \leq (1-\beta)/2$. Combining (\ref{eq:impr}) with the bound (\ref{eq:cond1}), it follows that, on low-energy states, $\cG_N^\beta$ is dominated by its quadratic part. As a consequence, to conclude the proof of Theorem \ref{thm:main}, we only have to conjugate $\cG_N^\beta$ with a second generalized Bogoliubov transformation, diagonalizing the quadratic operator $\cQ_N^\beta$.

It is in the proof of (\ref{eq:impr}) that the assumption 
$\beta < 1$ plays a crucial role. For $\beta =1$, in the Gross-Pitaevskii regime, the error term $\cE_N^\beta$ is not small; in this case, (\ref{eq:impr}) only holds with $\alpha =0$. In other words, the excitation Hamiltonian $\cG_N^{\beta=1}$ contains cubic and quartic contributions that remain of order one in the limit of large $N$. This observation is not surprising. Already in \cite{ESY} and more recently in \cite{NRS1,NRS2}, it has been shown that quasi-free states can only approximate the ground state of a dilute Bose gas up to an error of order one. For this reason, when $\beta =1$ we cannot expect to extract all relevant terms from the Hamiltonian (\ref{eq:Ham0}) by applying Bogoliubov transformations. To prove Theorem \ref{thm:main} in the Gross-Pitaevskii regime, the Hamilton operator $H_N^{\beta=1}$ must instead be conjugated with more complicated maps. A first partial result in this direction is the upper bound for the ground state energy obtained by Yau and Yin in \cite{YY}. 

The paper is organized as follows. In Section \ref{sec:fock} we introduce the formalism of second quantization, defining generalized Bogoliubov transformations and studying their properties. The main results of this section are Lemma \ref{lm:indu} and Lemma \ref{lm:conv-series} (taken from \cite{BS}) where we show how to expand the action of generalized Bogoliubov transformations in absolutely convergent infinite series. In Section \ref{sec:ex}, we define the excitation Hamiltonian $\cG_N^\beta$ and we state its most important properties (namely the bounds (\ref{eq:cond0}), (\ref{eq:comm0}) and (\ref{eq:impr})) in Theorem \ref{thm:gene}, whose long and technical proof is deferred to Section \ref{sec:prop}. In Section \ref{sec:cond}, we show how (\ref{eq:cond0}) and (\ref{eq:comm0}) can be used to show the bounds (\ref{eq:cond1}) and, more importantly, (\ref{eq:KN}) for the excitation vectors of low-energy states. In Section \ref{sec:diag}, we show how to diagonalize the quadratic part of $\cG_N^\beta$ using a second generalized Bogoliubov transformation. Using the results of Section \ref{sec:cond} and Section \ref{sec:diag}, we prove our main result, Theorem \ref{thm:main}, in Section \ref{sec:proof}.

\medskip

{\it Acknowledgement.} B.S. gratefully acknowledge support from the NCCR SwissMAP and from the Swiss National Foundation of Science through the SNF Grant ``Effective equations from quantum dynamics'' and the SNF Grant ``Dynamical and energetic properties of Bose-Einstein condensates''.

\section{Fock space}
\label{sec:fock}

Let 
\[ \cF = \bigoplus_{n \geq 0} L^2_s (\Lambda^{n}) = \bigoplus_{n \geq 0} L^2 (\Lambda)^{\otimes_s n} \]
be the bosonic Fock space over $L^2 (\Lambda)$, where $L^2_s (\Lambda^{n})$ is the subspace of $L^2 (\Lambda^n)$ consisting of wave functions that are symmetric w.r.t. permutations. We use the notation 
$\Omega = \{ 1, 0, \dots \} \in \cF$ for the vacuum vector in $\cF$. 

For $g \in L^2 (\Lambda)$, we define the creation operator $a^* (g)$ and the annihilation operator $a(g)$ by 
\[ \begin{split} 
(a^* (g) \Psi)^{(n)} (x_1, \dots , x_n) &= \frac{1}{\sqrt{n}} \sum_{j=1}^n g (x_j) \Psi^{(n-1)} (x_1, \dots , x_{j-1}, x_{j+1} , \dots , x_n) 
\\
(a (g) \Psi)^{(n)} (x_1, \dots , x_n) &= \sqrt{n+1} \int_\Lambda  \bar{g} (x) \Psi^{(n+1)} (x,x_1, \dots , x_n) \, dx \end{split} \]
Notice that $a^* (g)$ is the adjoint of $a(g)$ and that  creation and annihilation operators satisfy canonical commutation relations
\begin{equation}\label{eq:ccr} [a (g), a^* (h) ] = \langle g,h \rangle , \quad [ a(g), a(h)] = [a^* (g), a^* (h) ] = 0 \end{equation}
for all $g,h \in L^2 (\Lambda)$ (here $\langle g,h \rangle$ is the usual inner product on $L^2 (\Lambda)$). 

Since we consider a translation invariant system, it will be convenient to work in momentum space. {F}rom now on, let $\Lambda^* = 2\pi \bZ^3$. For $p \in \Lambda^*$, we define the normalized wave function $\ph_p  (x) = e^{-ip\cdot x}$ in $L^2 (\Lambda)$. We set
\begin{equation}\label{eq:ap} a^*_p = a^* (\ph_p), \quad \text{and } \quad  a_p = a (\ph_p) \end{equation} 
In other words, $a^*_p$ and $a_p$ create, respectively annihilate, a particle with momentum $p$. 

In some occasions, it will also be important to switch to position space (where it is easier to use the positivity of the potential $V(x)$). For this reason, we introduce operator valued distributions $\check{a}_x, \check{a}_x^*$ defined so that
\begin{equation}\label{eq:axf} a(f) = \int \bar{f} (x) \,  \check{a}_x \, dx , \quad a^* (f) = \int f(x) \, \check{a}_x^* \, dx  \end{equation}

On $\cF$, we also introduce the number of particles operator, defined by $(\cN\Psi)^{(n)} = n \Psi^{(n)}$. Notice that 
\[ \cN = \sum_{p \in \Lambda^*} a_p^* a_p  = \int \check{a}^*_x \check{a}_x \, dx \, . \]
It is important to observe that creation and annihilation operators are bounded by the square root of the number of particles operator, i.e.   
\begin{equation}\label{eq:abd} \| a (f) \Psi \| \leq \| f \| \| \cN^{1/2} \Psi \|, \quad \| a^* (f) \Psi \| \leq \| f \| \| (\cN+1)^{1/2} \Psi \| 
\end{equation}
for all $f \in L^2 (\Lambda)$. 

We will often deal with quadratic (and translation invariant) expressions in creation and annihilation operators. For $f \in \ell^2 (\Lambda^*)$, we define
\begin{equation}\label{eq:AAdef} A_{\sharp_1, \sharp_2} (f) = \sum_{p \in \Lambda^*} f_p \, a_{\alpha_1 p}^{\sharp_1} a_{\alpha_2 p}^{\sharp_2} \end{equation}
where $\sharp_1, \sharp_2 \in \{ \cdot, * \}$, and where we use the notation $a^\sharp = a$, if $\sharp = \cdot$, and $a^\sharp = a^*$ if $\sharp = *$. Also, $\alpha_j \in \{ \pm 1 \}$ is chosen so that $\alpha_1 = 1$, if $\sharp_1 = *$, $\alpha_1 = -1$ if $\sharp_1 = \cdot$, $\alpha_2 = 1$ if $\sharp_2 = \cdot$ and $\alpha_2 = -1$ if $\sharp_2 = *$. Notice that, in position space  
\[ A_{\sharp_1, \sharp_2} (f) = \int dx dy \, \check{f} (x-y) \, \check{a}_x^{\sharp_1} \, \check{a}_y^{\sharp_2} \]
with the inverse Fourier transform \[ \check{f} (x) = \sum_{p \in \Lambda^*} f_p e^{ip\cdot x} \, . \]  
In the next simple lemma, taken from \cite{BS}, we show how to bound quadratic operators of the form (\ref{eq:AAdef}). 
\begin{lemma} \label{lm:Abds} Let $f \in \ell^2 (\Lambda^*)$ and, if $\sharp_1 = \cdot$ and $\sharp_2 = *$ assume additionally that $j \in \ell^1 (\Lambda^*)$. Then we have, for any $\Psi \in \cF$, 
\[ \begin{split} 
\| A_{\sharp_1, \sharp_2} (f) \Psi \| &\leq \sqrt{2} \| (\cN+1) \Psi \| \cdot \left\{ \begin{array}{ll}  \| f \|_2 + \| f \|_1 \, \qquad &\text{if } \sharp_1 = \cdot, \sharp_2 = * \\  \| f \|_2  \qquad &\text{otherwise} \end{array} \right. \end{split} \]
\end{lemma} 

As already explained in the introduction, we will work on certain subspaces of $\cF$. Recall that $\ph_0 \in L^2 (\Lambda)$ is the constant wave function with $\ph_0 (x) = 1$ for all $x \in \Lambda$. We denote by $L^2_{\perp} (\Lambda)$ the orthogonal complement of the one dimensional space spanned by $\ph_0$ in $L^2 (\Lambda)$. We
define 
\[ \cF_{+} = \bigoplus_{n \geq 0} L^2_{\perp} (\Lambda)^{\otimes_s n} \, . \]
as the Fock space constructed over $L^2_\perp (\Lambda)$, i.e. the Fock space generated by creation and annihilation operators $a_p^*, a_p$, with $p \in \Lambda^*_+ := 2\pi \bZ^3 \backslash \{ 0 \}$. On $\cF_+$, we denote the number of particles operator by 
\[ \cN_+ = \sum_{p \in \Lambda^*_+} a_p^* a_p \]
We will also need a truncated version of $\cF_+$. For $N \in \bN$, we define 
\[ \cF_{+}^{\leq N} = \bigoplus_{n=0}^N L^2_{\perp} (\Lambda)^{\otimes_s n} \, . \]
On $\cF_+^{\leq N}$, we construct modified creation and annihilation operators. For $f \in L^2_\perp (\Lambda)$, we set 
\[ b (f) = \sqrt{\frac{N- \cN_+}{N}} \, a (f), \qquad \text{and } \quad  b^* (f) = a^* (f) \, \sqrt{\frac{N-\cN_+}{N}} \]
We have $b(f), b^* (f) : \cF_+^{\leq N} \to \cF_+^{\leq N}$. As we will discuss in the next section, the importance of these fields arises from the application of the map $U$, defined in (\ref{eq:Uph-def}), since, for instance,  
\begin{equation}\label{eq:UaU} U a^* (f) a(\ph_0) U^* = a^* (f) \sqrt{N-\cN_+}  = \sqrt{N} b^* (f) 
\end{equation}
Based on (\ref{eq:UaU}), we can interpret $b^* (f)$ as an operator exciting a particle from the condensate into its orthogonal complement. Compared with the standard fields $a^*,a$, the modified operators $b^*,b$ have an important advantage; they create (or annihilate) excitations but, at the same time, they annihilate (or create) a particle in the condensate, preserving thus the total number of particles.  

It is convenient to introduce modified creation and annihilation operators in momentum space, setting
\[ b_p = \sqrt{\frac{N-\cN_+}{N}} \, a_p, \qquad \text{and } \quad  b^*_p = a^*_p \, \sqrt{\frac{N-\cN_+}{N}} \]
for all $p \in \Lambda^*_+$ and operator valued distributions 
in position space 
\[ \check{b}_x = \sqrt{\frac{N-\cN_+}{N}} \, \check{a}_x, \qquad \text{and } \quad  \check{b}^*_x = \check{a}^*_x \, \sqrt{\frac{N-\cN_+}{N}} \]
for all $x \in \Lambda$. 

Modified creation and annihilation operators satisfy the commutation relations 
\begin{equation}\label{eq:comm-bp} \begin{split} [ b_p, b_q^* ] &= \left( 1 - \frac{\cN_+}{N} \right) \delta_{p,q} - \frac{1}{N} a_q^* a_p 
\\ [ b_p, b_q ] &= [b_p^* , b_q^*] = 0 
\end{split} \end{equation}
and, in position space, 
\begin{equation}\label{eq:comm-b}
\begin{split}  [ \check{b}_x, \check{b}_y^* ] &= \left( 1 - \frac{\cN_+}{N} \right) \delta (x-y) - \frac{1}{N} \check{a}_y^* \check{a}_x \\ 
[ \check{b}_x, \check{b}_y ] &= [ \check{b}_x^* , \check{b}_y^*] 
= 0 
\end{split} \end{equation}
Furthermore 
\begin{equation}\label{eq:comm-b2}
\begin{split}
[\check{b}_x, \check{a}_y^* \check{a}_z] &=\delta (x-y)\check{b}_z, \qquad 
[\check{b}_x^*, \check{a}_y^* \check{a}_z] = -\delta (x-z) \check{b}_y^*
\end{split} \end{equation}
It follows that $[ \check{b}_x, \cN_+ ] = \check{b}_x$, $[ \check{b}_x^* , \cN_+ ] = - \check{b}_x^*$ and, in momentum space, $[b_p , \cN_+] = b_p$, $[b_p^*, \cN_+] = - b_p^*$. With (\ref{eq:abd}), we obtain 
\begin{equation}\label{lm:bbds}
\begin{split} 
\| b(f) \xi \| &\leq \| f \| \left\| \cN_+^{1/2} \left( \frac{N+1-\cN_+}{N} \right)^{1/2} \xi \right\| \\ 
\| b^* (f) \xi \| &\leq \| f \| \left\| (\cN_+ +1)^{1/2} \left( \frac{N-\cN_+ }{N} \right)^{1/2} \xi \right\|
\end{split} \end{equation}
for all $f \in L^2_\perp (\Lambda)$ and $\xi \in \cF^{\leq N}_+$. Since $\cN_+  \leq N$ on $\cF_+^{\leq N}$, $b(f), b^* (f)$ are bounded operators with $\| b(f) \|, \| b^* (f) \| \leq (N+1)^{1/2} \| f \|$. 

We will consider quadratic expressions in the $b$-fields. As we did in (\ref{eq:AAdef}), we restrict our attention to translation invariant operators. For $f \in \ell^2 (\Lambda^*_+)$, we let
\[ B_{\sharp_1, \sharp_2} (f) = \sum_{p \in \Lambda^*} f_p b_{\alpha_1 p}^{\sharp_1} b_{\alpha_2 p}^{\sharp_2} \]
with $\alpha_1 = 1$ if $\sharp_1 = *$, $\alpha_1 = -1$ if $\sharp_1 = \cdot$, $\alpha_2 = 1$ if $\sharp_2 = \cdot$ and $\alpha_2 = -1$ if $\sharp_2 = *$. By construction, $B_{\sharp_1, \sharp_2} (f) : \cF_+^{\leq N} \to \cF_+^{\leq N}$. In position space, we find 
\[ \begin{split} 
B_{\sharp_1, \sharp_2} (f) &= \int \check{f} (x-y) \, b^{\sharp_1}_x b^{\sharp_2}_y \, dx dy  
\end{split}  \]
From Lemma \ref{lm:Abds}, we obtain the following bounds.
\begin{lemma}\label{lm:Bbds}
Let $f \in \ell^2 (\Lambda^*_+)$. If $\sharp_1 = \cdot$ and $\sharp_2 = *$, we assume additionally that $f \in \ell^1 (\Lambda^*_+)$. Then
\[ \begin{split} \frac{\| B_{\sharp_1,\sharp_2} (f) \xi \|}{\left\| (\cN_+ +1) \left(\frac{N+2-\cN_+ }{N} \right) \xi \right\|} &
\leq \sqrt{2} \cdot  \left\{ \begin{array}{ll}  \| f \|_2 + \| f \|_1 \qquad &\text{if } \sharp_1 = \cdot, \sharp_2 = * \\  \| f \|_2  \qquad &\text{otherwise} \end{array} \right. 
\end{split} \]
for all $\xi \in \cF^{\leq N}$. Since $\cN_+  \leq N$ on $\cF^{\leq N}_+$, the operator $B_{\sharp_1 , \sharp_2} (f)$ is bounded, with 
\[ \begin{split} \| B_{\sharp_1, \sharp_2} (f) \| &\leq \sqrt{2} N \left\{ \begin{array}{ll} \| f \|_2 + \| f \|_1 \quad &\text{if } \sharp_1 = \cdot, \sharp_2 = * \\  \| f \|_2  \qquad &\text{otherwise} \end{array} \right. \end{split} \]
\end{lemma}

We will need to consider products of several creation and annihilation operators. In particular, two types of monomials in creation and annihilation operators will play an important role in our analysis. For $f_1, \dots , f_n \in \ell_2 (\Lambda^*_+)$, $\sharp = (\sharp_1, \dots , \sharp_n), \flat = (\flat_0, \dots , \flat_{n-1}) \in \{ \cdot, * \}^n$, we set 
\begin{equation}\label{eq:Pi2}
\begin{split}  
\Pi^{(2)}_{\sharp, \flat} &(f_1, \dots , f_n) \\ &= \sum_{p_1, \dots , p_n \in \Lambda^*}  b^{\flat_0}_{\alpha_0 p_1} a_{\beta_1 p_1}^{\sharp_1} a_{\alpha_1 p_2}^{\flat_1} a_{\beta_2 p_2}^{\sharp_2} a_{\alpha_2 p_3}^{\flat_2} \dots  a_{\beta_{n-1} p_{n-1}}^{\sharp_{n-1}} a_{\alpha_{n-1} p_n}^{\flat_{n-1}} b^{\sharp_n}_{\beta_n p_n} \, \prod_{\ell=1}^n f_\ell (p_\ell)  \end{split} \end{equation}
where, for every $\ell=0,1, \dots , n$, we set $\alpha_\ell = 1$ if $\flat_\ell = *$, $\alpha_\ell =    -1$ if $\flat_\ell = \cdot$, $\beta_\ell = 1$ if $\sharp_\ell = \cdot$ and $\beta_\ell = -1$ if $\sharp_\ell = *$. In (\ref{eq:Pi2}), we impose the condition that for every $j=1,\dots, n-1$, we have either $\sharp_j = \cdot$ and $\flat_j = *$ or $\sharp_j = *$ and $\flat_j = \cdot$ (so that the product $a_{\beta_\ell p_\ell}^{\sharp_\ell} a_{\alpha_\ell p_{\ell+1}}^{\flat_\ell}$ always preserves the number of particles, for all $\ell =1, \dots , n-1$). With this assumption, we find that the operator $\Pi^{(2)}_{\sharp,\flat} (f_1, \dots , f_n)$ maps $\cF^{\leq N}_+$ into itself. If, for some $\ell=1, \dots , n$, $\flat_{\ell-1} = \cdot$ and $\sharp_\ell = *$ (i.e. if the product $a_{\alpha_{\ell-1} p_\ell}^{\flat_{\ell-1}} a_{\beta_\ell p_\ell}^{\sharp_\ell}$ for $\ell=2,\dots , n$, or the product $b_{\alpha_0 p_1}^{\flat_0} a_{\beta_1 p_1}^{\sharp_1}$ for $\ell=1$, is not normally ordered) we require additionally that $f_\ell  \in \ell^1 (\Lambda^*_+)$. In position space, the same operator can be written as 
\begin{equation}\label{eq:Pi2-pos} \Pi^{(2)}_{\sharp, \flat} (f_1, \dots , f_n) = \int   \check{b}^{\flat_0}_{x_1} \check{a}_{y_1}^{\sharp_1} \check{a}_{x_2}^{\flat_1} \check{a}_{y_2}^{\sharp_2} \check{a}_{x_3}^{\flat_2} \dots  \check{a}_{y_{n-1}}^{\sharp_{n-1}} \check{a}_{x_n}^{\flat_{n-1}} \check{b}^{\sharp_n}_{y_n} \, \prod_{\ell=1}^n \check{f}_\ell (x_\ell - y_\ell) \, dx_\ell dy_\ell \end{equation}
An operator of the form (\ref{eq:Pi2}), (\ref{eq:Pi2-pos}) with all the properties listed above, will be called a $\Pi^{(2)}$-operator of order $n$.

For $g, f_1, \dots , f_n \in \ell_2 (\Lambda^*_+)$, $\sharp = (\sharp_1, \dots , \sharp_n)\in \{ \cdot, * \}^n$, $\flat = (\flat_0, \dots , \flat_{n}) \in \{ \cdot, * \}^{n+1}$, we also define the operator 
\begin{equation}\label{eq:Pi1}
\begin{split} \Pi^{(1)}_{\sharp,\flat} &(f_1, \dots , f_n;g) \\ &= \sum_{p_1, \dots , p_n \in \Lambda^*}  b^{\flat_0}_{\alpha_0, p_1} a_{\beta_1 p_1}^{\sharp_1} a_{\alpha_1 p_2}^{\flat_1} a_{\beta_2 p_2}^{\sharp_2} a_{\alpha_2 p_3}^{\flat_2} \dots a_{\beta_{n-1} p_{n-1}}^{\sharp_{n-1}} a_{\alpha_{n-1} p_n}^{\flat_{n-1}} a^{\sharp_n}_{\beta_n p_n} a^{\flat n} (g) \, \prod_{\ell=1}^n f_\ell (p_\ell) \end{split} \end{equation}
where $\alpha_\ell$ and $\beta_\ell$ are defined as above. Also here, we impose the condition that, for all $\ell = 1, \dots , n$, either $\sharp_\ell = \cdot$ and $\flat_\ell = *$ or $\sharp_\ell = *$ and $\flat_\ell = \cdot$. This implies that $\Pi^{(1)}_{\sharp,\flat} (f_1, \dots , f_n;g)$ maps $\cF^{\leq N}_+$ back into $\cF_+^{\leq N}$. Additionally, we assume that $f_\ell \in \ell^1 (\Lambda^*_+)$, if $\flat_{\ell-1} = \cdot$ and $\sharp_\ell = *$ for some $\ell = 1,\dots , n$ (i.e. if the pair $a_{\alpha_{\ell-1} p_\ell}^{\flat_{\ell-1}} a^{\sharp_\ell}_{\beta_\ell p_\ell}$ is not normally ordered). In position space, the same operator can be written as
\begin{equation}\label{eq:Pi1-pos} \Pi^{(1)}_{\sharp,\flat} (f_1, \dots ,f_n;g) = \int \check{b}^{\flat_0}_{x_1} \check{a}_{y_1}^{\sharp_1} \check{a}_{x_2}^{\flat_1} \check{a}_{y_2}^{\sharp_2} \check{a}_{x_3}^{\flat_2} \dots  \check{a}_{y_{n-1}}^{\sharp_{n-1}} \check{a}_{x_n}^{\flat_{n-1}} \check{a}^{\sharp_n}_{y_n} \check{a}^{\flat n} (g) \, \prod_{\ell=1}^n \check{f}_\ell (x_\ell - y_\ell) \, dx_\ell dy_\ell \end{equation}
An operator of the form (\ref{eq:Pi1}), (\ref{eq:Pi1-pos}) will be called a $\Pi^{(1)}$-operator of order $n$. Operators of the form $b(f)$, $b^* (f)$, for a $f \in \ell^2 (\Lambda^*_+)$, will be called $\Pi^{(1)}$-operators of order zero. 

In the next lemma we show how to bound $\Pi^{(2)}$- and $\Pi^{(1)}$-operators. The simple proof, based on Lemma \ref{lm:Abds}, can be found in \cite{BS}. 
\begin{lemma}\label{lm:Pi-bds}
Let $n \in \bN$, $g, f_1, \dots , f_n \in \ell^2 (\Lambda^*_+)$. Assume that $\Pi^{(2)}_{\sharp,\flat} (f_1,\dots , f_n)$ and $\Pi^{(1)}_{\sharp,\flat} (f_1,\dots, f_n ; g)$ are defined as in (\ref{eq:Pi2}), (\ref{eq:Pi1}). Then  
\begin{equation}\label{eq:Pi-bds} \begin{split} 
\left\| \Pi^{(2)}_{\sharp,\flat} (f_1,\dots ,f_n) \xi \right\| &\leq 6^n \prod_{\ell=1}^n K_\ell^{\flat_{\ell-1}, \sharp_\ell} \left\| (\cN_+ +1)^n \left(1- \frac{\cN_+ -2}{N} \right) \xi \right\| \\
 \left\| \Pi^{(1)}_{\sharp,\flat} (f_1,\dots , f_n;g) \xi \right\| 
 &\leq 6^n \| g \| \prod_{\ell=1}^n K_\ell^{\flat_{\ell-1}, \sharp_\ell} \left\| (\cN_+ +1)^{n+1/2} \left(1- \frac{\cN_+ -2}{N} \right)^{1/2} \xi \right\| 
\end{split} \end{equation} 
where
\[ K_\ell^{\flat_{\ell-1}, \sharp_\ell} = \left\{ \begin{array}{ll} \|f_\ell \|_2 + \| f_\ell \|_1  \quad &\text{if } \flat_{\ell-1} = \cdot \text{ and } \sharp_\ell = * \\
\| f_\ell \|_2 \quad &\text{otherwise} \end{array} \right. \]
Since $\cN_+  \leq N$ on $\cF^{\leq N}_+$, it follows that  
\[ \begin{split} 
\left\|  \Pi^{(2)}_{\sharp,\flat} (f_1,\dots, f_n) \right\| &\leq (12 N)^n \prod_{\ell=1}^n K_\ell^{\flat_{\ell-1}, \sharp_\ell} \\ 
\left\|  \Pi^{(1)}_{\sharp,\flat} (f_1,\dots ,f_n;g) \right\| &\leq (12 N)^n \sqrt{N} \| g \| \prod_{\ell=1}^n K_\ell^{\flat_{\ell-1}, \sharp_\ell} 
\end{split} \]
\end{lemma}

Next, we introduce generalized Bogoliubov transformations and we discuss their properties. For $\eta \in \ell^2 (\Lambda^*_+)$ with $\eta_{-p} = \eta_{p}$ for all $p \in \Lambda^*_+$, we define 
\begin{equation}\label{eq:defB} 
B(\eta) = \frac{1}{2} \sum_{p\in \Lambda^*_+}  \left( \eta_p b_p^* b_{-p}^* - \bar{\eta}_p b_p b_{-p} \right) \,. \end{equation}
and we consider 
\begin{equation}\label{eq:eBeta} 
e^{B(\eta)} = \exp \left[ \frac{1}{2} \sum_{p \in \Lambda^*_+}   \left( \eta_p b_p^* b_{-p}^* - \bar{\eta}_p  b_p b_{-p} \right) \right] 
\end{equation}
Notice that $B(\eta), e^{B(\eta)} : \cF_{+}^{\leq N} \to \cF_{+}^{\leq N}$. We refer to unitary operators of the form (\ref{eq:eBeta}) as generalized Bogoliubov transformations. The name arises from the observation that, on states with $\cN_+  \ll N$, we have $b_p \simeq a_p$, $b^*_p \simeq a_p^*$ and therefore 
\begin{equation}\label{eq:wtB} B(\eta) \simeq \wt{B} (\eta) = \frac{1}{2} \sum_{p \in \Lambda^*_+} \left( \eta_p a_p^* a_{-p}^* - \bar{\eta}_p a_p a_{-p} \right)  \end{equation}
Since $\wt{B} (\eta)$ is quadratic in creation and annihilation operators, $\exp (\wt{B} (\eta))$ is a standard Bogoliubov transformation, whose action on creation and annihilation operators is explicitly 
given by 
\begin{equation}\label{eq:act-Bog} e^{-\wt{B} (\eta)} a_p e^{\wt{B} (\eta)} = \cosh (\eta_p) a_p + \sinh (\eta_p) a^*_{-p} \, . 
\end{equation}
As discussed in the introduction, (\ref{eq:wtB}) does not map $\cF^{\leq N}$ into itself. For this reason, in the following it will be convenient for us to work with generalized Bogoliubov transformations of the form (\ref{eq:eBeta}). The price we have to pay is the fact that there is no explicit expression like (\ref{eq:act-Bog}) for the action of (\ref{eq:eBeta}). We need other tools to control the action of 
generalized Bogoliubov transformations. 

A first important observation is the following lemma, whose proof can be found in \cite{BS} (a similar result was previously established in \cite{Sei}).
\begin{lemma}\label{lm:Ngrow}
Let $\eta \in \ell^2 (\Lambda^*)$ and $B(\eta)$ be defined as in (\ref{eq:defB}). Then, for every $n_1, n_2 \in \bZ$, there exists a constant $C > 0$ (depending on $\| \eta\|$)  such that, on $\cF_+^{\leq N}$, 
\[ e^{-B(\eta)} (\cN_+ +1)^{n_1} (N+1-\cN_+ )^{n_2} e^{B(\eta)} \leq C (\cN_+ +1)^{n_1} (N+1- \cN_+ )^{n_2} \,. \]
\end{lemma}

Controlling the change of the number of particles operator is not enough for our purposes. We will often need to express the action of generalized Bogoliubov transformations by means of a convergent series of nested commutators. We notice, first of all, that for any $p \in \Lambda^*_+$, 
\[\begin{split} e^{-B(\eta)} \, b_p \, e^{B(\eta)} &= b_p + \int_0^1 ds \, \frac{d}{ds}  e^{-sB(\eta)} b_p e^{sB(\eta)} \\ &= b_p - \int_0^1 ds \, e^{-sB(\eta)} [B(\eta), b_p] e^{s B(\eta)} \\ &= b_p - [B(\eta),b_p] + \int_0^1 ds_1 \int_0^{s_1} ds_2 \, e^{-s_2 B(\eta)} [B(\eta), [B(\eta),b_p] e^{s_2 B(\eta)} \end{split} \]
Iterating $m$ times, we find 
\begin{equation}\label{eq:BCH} \begin{split} 
e^{-B(\eta)} b_p e^{B(\eta)} = &\sum_{n=1}^{m-1} (-1)^n \frac{\text{ad}^{(n)}_{B(\eta)} (b_p)}{n!} \\ &+ \int_0^{1} ds_1 \int_0^{s_1} ds_2 \dots \int_0^{s_{m-1}} ds_m \, e^{-s_m B(\eta)} \text{ad}^{(m)}_{B(\eta)} (b_p) e^{s_m B(\eta)} \end{split} \end{equation}
where we recursively defined \[ \text{ad}_{B(\eta)}^{(0)} (A) = A \quad \text{and } \quad \text{ad}^{(n)}_{B(\eta)} (A) = [B(\eta), \text{ad}^{(n-1)}_{B(\eta)} (A) ]  \]
We will later show that, under appropriate assumptions on $\eta$, the norm of the error term on the r.h.s. of (\ref{eq:BCH}) vanishes, as $m \to \infty$. Hence, the action of $e^{B(\eta)}$ on $b_p, b_p^*$ can be described in terms of the nested commutators $\text{ad}^{(n)}_{B(\eta)} (b_p)$ and $\text{ad}^{(n)}_{B(\eta)} (b^*_p)$ for $n \in \bN$. In the next lemma, we give a detailed analysis of these operators; the proof can be found in \cite[Lemma 2.5]{BBCS}. 
\begin{lemma}\label{lm:indu}
Let $\eta \in \ell^2 (\Lambda^*_+)$ be such that $\eta_p = \eta_{-p}$ for all $p \in \ell^2 (\Lambda^*)$. To simplify the notation,  assume also $\eta$ to be real-valued (as it will be in applications). Let $B(\eta)$ be defined as in (\ref{eq:defB}), $n \in \bN$ and $p \in \Lambda^*$. Then the nested commutator $\text{ad}^{(n)}_{B(\eta)} (b_p)$ can be written as the sum of exactly $2^n n!$ terms, with the following properties. 
\begin{itemize}
\item[i)] Possibly up to a sign, each term has the form
\begin{equation}\label{eq:Lambdas} \Lambda_1 \Lambda_2 \dots \Lambda_i \, N^{-k} \Pi^{(1)}_{\sharp,\flat} (\eta^{j_1}, \dots , \eta^{j_k} ; \eta^{s}_p \ph_{\alpha p}) 
\end{equation}
for some $i,k,s \in \bN$, $j_1, \dots ,j_k \in \bN \backslash \{ 0 \}$, $\sharp \in \{ \cdot, * \}^k$, $ \flat \in \{ \cdot, * \}^{k+1}$ and $\alpha \in \{ \pm 1 \}$ chosen so that $\alpha = 1$ if $\flat_k = \cdot$ and $\alpha = -1$ if $\flat_k = *$ (recall here that $\ph_p (x) = e^{-ip \cdot x}$). In (\ref{eq:Lambdas}), each operator $\Lambda_w : \cF^{\leq N} \to \cF^{\leq N}$, $w=1, \dots , i$, is either a factor $(N-\cN_+ )/N$, a factor $(N-(\cN_+ -1))/N$ or an operator of the form
\begin{equation}\label{eq:Pi2-ind} N^{-h} \Pi^{(2)}_{\sharp',\flat'} (\eta^{z_1}, \eta^{z_2},\dots , \eta^{z_h}) \end{equation}
for some $h, z_1, \dots , z_h \in \bN \backslash \{ 0 \}$, $\sharp,\flat  \in \{ \cdot , *\}^h$. 
\item[ii)] If a term of the form (\ref{eq:Lambdas}) contains $m \in \bN$ factors $(N-\cN_+ )/N$ or $(N-(\cN_+ -1))/N$ and $j \in \bN$ factors of the form (\ref{eq:Pi2-ind}) with $\Pi^{(2)}$-operators of order $h_1, \dots , h_j \in \bN \backslash \{ 0 \}$, then 
we have
\begin{equation}\label{eq:totalb} m + (h_1 + 1)+ \dots + (h_j+1) + (k+1) = n+1 \end{equation}
\item[iii)] If a term of the form (\ref{eq:Lambdas}) contains (considering all $\Lambda$-operators and the $\Pi^{(1)}$-operator) the arguments $\eta^{i_1}, \dots , \eta^{i_m}$ and the factor $\eta^{s}_p$ for some $m, s \in \bN$, and $i_1, \dots , i_m \in \bN \backslash \{ 0 \}$, then \[ i_1 + \dots + i_m + s = n .\]
\item[iv)] There is exactly one term having of the form (\ref{eq:Lambdas}) with $k=0$ and such that all $\Lambda$-operators are factors of $(N-\cN_+ )/N$ or of $(N+1-\cN_+ )/N$. It is given by 
\begin{equation}\label{eq:iv1} \left(\frac{N-\cN_+ }{N} \right)^{n/2} \left(\frac{N+1-\cN_+ }{N} \right)^{n/2} \eta^{n}_p b_p 
\end{equation}
if $n$ is even, and by 
\begin{equation}\label{eq:iv2} - \left(\frac{N-\cN_+ }{N} \right)^{(n+1)/2} \left(\frac{N+1-\cN_+ }{N} \right)^{(n-1)/2} \eta^{n}_p b^*_{-p}  \end{equation}
if $n$ is odd.
\item[v)] If the $\Pi^{(1)}$-operator in (\ref{eq:Lambdas}) is of order $k \in \bN \backslash \{ 0 \}$, it has either the form  
\[ \sum_{p_1, \dots , p_k}  b^{\flat_0}_{\alpha_0 p_1} \prod_{i=1}^{k-1} a^{\sharp_i}_{\beta_i p_{i}} a^{\flat_i}_{\alpha_i p_{i+1}}  a^*_{-p_k} \eta^{2r}_p  a_p \prod_{i=1}^k \eta^{j_i}_{p_i}  \]
or the form 
\[\sum_{p_1, \dots , p_k} b^{\flat_0}_{\alpha_0 p_1} \prod_{i=1}^{k-1} a^{\sharp_i}_{\beta_i p_{i}} a^{\flat_i}_{\alpha_i p_{i+1}}  a_{p_k} \eta^{2r+1}_p a^*_p \prod_{i=1}^k \eta^{j_i}_{p_i}  \]
for some $r \in \bN$, $j_1, \dots , j_k \in \bN \backslash \{ 0 \}$. If it is of order $k=0$, then it is either given by $\eta^{2r}_p b_p$ or by $\eta^{2r+1}_p b_{-p}^*$, for some $r \in \bN$. 
\item[vi)] For every non-normally ordered term of the form 
\[ \begin{split} &\sum_{q \in \Lambda^*} \eta^{i}_q a_q a_q^* , \quad \sum_{q \in \Lambda^*} \, \eta^{i}_q b_q a_q^* \\  &\sum_{q \in \Lambda^*} \, \eta^{i}_q a_q b_q^*, \quad \text{or } \quad \sum_{q \in \Lambda^*} \, \eta^{i}_q b_q b_q^*  \end{split} \]
appearing either in the $\Lambda$-operators or in the $\Pi^{(1)}$-operator in (\ref{eq:Lambdas}), we have $i \geq 2$.
\end{itemize}
\end{lemma}
With Lemma \ref{lm:indu}, it follows that the series on the r.h.s. of (\ref{eq:BCH}) converges absolutely, if the $\ell^2$-norm $\| \eta \|$ is small enough. The proof of the next Lemma is a simple adaptation of the proof of \cite[Lemma 3.3]{BS} 
\begin{lemma}\label{lm:conv-series}
Let $\eta \in \ell^2 (\Lambda^*)$ be symmetric, with $\| \eta \|$ sufficiently small. Then we have 
\begin{equation}\label{eq:conv-serie}
\begin{split} e^{-B(\eta)} b_p e^{B (\eta)} &= \sum_{n=0}^\infty \frac{(-1)^n}{n!} \text{ad}_{B(\eta)}^{(n)} (b_p) \\
e^{-B(\eta)} b^*_p e^{B (\eta)} &= \sum_{n=0}^\infty \frac{(-1)^n}{n!} \text{ad}_{B(\eta)}^{(n)} (b^*_p) \end{split} \end{equation}
and the series on the r.h.s. are absolutely convergent. 
\end{lemma}

\section{The excitation Hamiltonian} 
\label{sec:ex}

We define the unitary operator $U: L^2_s (\Lambda^N) \to \cF_+^{\leq N}$ as in (\ref{eq:Uph-def}). In terms of creation and annihilation operators, the map $U$ is given by 
\[ U \psi = \bigoplus_{n=0}^N  (1-|\ph_0 \rangle \langle \ph_0|)^{\otimes n} \frac{a(\ph_0)^{N-n}}{\sqrt{(N-n)!}} \psi \]
for all $\psi \in L^2_s (\Lambda^N)$ (here we identify $\psi \in L^2_s (\Lambda^N)$ with the vector $\{ 0, \dots , 0, \psi, 0, \dots \} \in \cF$). The map $U^* : \cF_{+}^{\leq N} \to L^2_s (\Lambda^N)$ is given, on the other hand, by 
\[ U^* \{ \psi^{(0)}, \dots , \psi^{(N)} \} = \sum_{n=0}^N \frac{a^* (\ph_0)^{N-n}}{\sqrt{(N-n)!}} \psi^{(n)} \]
It is useful to compute the action of $U$ on the product of a creation and an annihilation operators. We find (see \cite{LNSS}):
\begin{equation}\label{eq:U-rules}
\begin{split} 
U a^*_0 a_0 U^* &= N- \cN_+  \\
U a^*_p a_0 U^* &= a^*_p \sqrt{N-\cN_+ } \\
U a^*_0 a_p U^* &= \sqrt{N-\cN_+ } a_p \\
U a^*_p a_q U^* &= a^*_p a_q 
\end{split} \end{equation}
for all $p,q \in \Lambda^*_+ = 2\pi \bZ^3 \backslash \{ 0 \}$. Writing (\ref{eq:Ham0}) in momentum space, we find
\[ H^\beta_N = \sum_{p \in \Lambda^*} p^2 a_p^* a_p + \frac{\kappa}{N} \sum_{p,q,r \in \Lambda^*} \widehat{V} (r/N^\beta) a_p^* a_q^* a_{q-r} a_{p+r} \]
where 
\[ \widehat{V} (q) = \int_{\bR^3} V (x) e^{-i q \cdot x} dx \] 
denotes the Fourier transform of $V$, for all $q \in \bR^3$. With (\ref{eq:U-rules}), we can conjugate $H^\beta_N$ with the map $U$, defining $\cL^\beta_N = U H^\beta_N U^* : \cF_+^{\leq N} \to \cF_+^{\leq N}$. We find
\begin{equation}\label{eq:cLN} \cL^\beta_N =  \cL^{(0)}_{\beta,N} + \cL^{(2)}_{\beta,N} + \cL^{(3)}_{\beta,N} + \cL^{(4)}_{\beta,N} \end{equation}
with
\begin{equation}\label{eq:cLNj} \begin{split} 
\cL_{N,\beta}^{(0)} =\;& \frac{N-1}{2N} \widehat{V} (0) (N-\cN_+ ) + \frac{\widehat{V} (0)}{2N} \cN_+  (N-\cN_+ ) \\
\cL^{(2)}_{N,\beta} =\; &\sum_{p \in \Lambda^*_+} p^2 a_p^* a_p + \sum_{p \in \Lambda_+^*} \widehat{V} (p/N^\beta) \left[ b_p^* b_p - \frac{1}{N} a_p^* a_p \right] \\ &+ \frac{1}{2} \sum_{p \in \Lambda^*_+} \widehat{V} (p/N^\beta) \left[ b_p^* b_{-p}^* + b_p b_{-p} \right] \\
\cL^{(3)}_{N,\beta} =\; &\frac{1}{\sqrt{N}} \sum_{p,q \in \Lambda_+^* : p+q \not = 0} \widehat{V} (p/N^\beta) \left[ b^*_{p+q} a^*_{-p} a_q  + a_q^* a_{-p} b_{p+q} \right] \\
\cL^{(4)}_{N,\beta} =\; & \frac{1}{2N} \sum_{p,q \in \Lambda_+^*, r \in \Lambda^*: r \not = -p,-q} \widehat{V} (r/N^\beta) a^*_{p+r} a^*_q a_p a_{q+r} 
\end{split} \end{equation}

As explained in the introduction, for $\beta \not = 0$, conjugation with $U$ does not yet pull all important contributions for low-energy states into the constant and the quadratic parts of the excitation Hamiltonian  $\cL^{\beta}_N$. In other words, in contrast to the mean-field case $\beta =0$, for $\beta >0$ we cannot expect $\cL^{(3)}_{N,\beta}$ and $\cL^{(4)}_{N,\beta}$ to be small on low-energy states, in the limit $N \to \infty$. For this reason, we need to conjugate $\cL^\beta_N$ with an appropriate generalized Bogoliubov transformation of the form (\ref{eq:eBeta}). 

To choose the function $\eta \in \ell^2 (\Lambda^*_+)$ entering (\ref{eq:defB}) and (\ref{eq:eBeta}), we fix a length $0< \ell < 1/2$, independently of $N$, and we consider the solution of the Neumann problem 
\begin{equation}\label{eq:scatl} 
\left(-\Delta + \frac{\kappa}{2} N^{3\beta-1} V (N^\beta x) \right) f_{N,\ell} (x) = \lambda_{N,\ell} f_{N,\ell} (x) \end{equation}
on the  ball $|x| \leq \ell$, with radial derivative 
$\partial_{|x|} f_{N,\ell} (x) = 0$ and with 
the normalization $f_{N,\ell} (x) = 1$ for $|x| = \ell$ (we omit the $\beta$-dependence in the notation for $f_{N,\ell}$ and for $\lambda_{N,\ell}$). The condition $\ell < 1/2$ guarantees that the ball of radius $\ell$ is contained in $\Lambda$. We extend then $f_{N,\ell}$ to $\Lambda$, by setting $f_{N,\ell} (x) = 1$ for all $|x| > \ell$. Then, for all $x \in \Lambda$, we have 
\begin{equation}\label{eq:scatlN}
 \left( -\Delta + \frac{\kappa}{2} N^{3\beta-1} V (N^\beta x) \right) f_{N,\ell} (x) = \lambda_{N,\ell} f_{N,\ell} (x) \chi_\ell (x) 
\end{equation}
where $\chi_\ell$ is the characteristic function of the ball of radius $\ell$. It is also useful to define $w_{N,\ell} = 1-f_{N,\ell}$ (so that $w_{N,\ell} (x) = 0$ if $|x| > \ell$). Since $w_{N,\ell}$ is compactly supported on $\Lambda$, it can be interpreted as a periodic function. Its Fourier coefficients are given by 
\[ \widehat{w}_{N,\ell} (p) =  \int_{\Lambda} w_{N,\ell} (x) e^{-i p \cdot x} dx  \]
for all $p \in \Lambda^*$. {F}rom (\ref{eq:scatlN}), we find that 
\begin{equation}\label{eq:wellp}
\begin{split}  - p^2 \widehat{w}_{N,\ell} (p) +  \frac{\kappa}{2N} \widehat{V} (p/N^\beta) &- \frac{\kappa}{2N} \sum_{q \in \Lambda^*} \widehat{V} ((p-q)/N^\beta) \widehat{w}_{N,\ell} (q) \\ &	
= \lambda_{N,\ell} \widehat{\chi}_\ell (p) - \lambda_{N,\ell} \sum_{q \in \Lambda^*} \widehat{\chi}_\ell (p-q) \widehat{w}_{N,\ell} (q) \end{split} \end{equation}
for all $p \in \Lambda^*$. In the next lemma we collect some important properties of $\lambda_{N,\ell}$ and of the functions $w_{N,\ell}, f_{N,\ell}$; the proof can be found in \cite[Lemma A.1]{ESY0} and in \cite{BS} (notice that this lemma is the reason we require $V \in L^3 (\bR^3)$; for the rest of the analysis, $V \in L^2 (\bR^3)$ would suffice). 
\begin{lemma} \label{3.0.sceqlemma}
Let $V \in L^3 (\bR^3)$ be non-negative, compactly supported and spherically symmetric. Fix $0 < \ell < 1/2$ and let $f_{N,\ell}$ denote the ground state solution of the Neumann problem \eqref{eq:scatl}. 
\begin{enumerate}
\item [i)] We have 
\[ \lambda_{N,\ell} = \frac{3\kappa \widehat{V} (0)}{8\pi N \ell^3} \left(1 + O(N^{\beta -1}) \right) \]
\item[ii)] We have $0\leq f_{N,\ell}, w_{N,\ell} \leq 1$.
\item[iii)] There exists a constant $C>0$ such that 
	\begin{equation}\label{3.0.scbounds1} 
	w_{N,\ell} (x) \leq \frac{C\kappa}{N (|x|+ N^{-\beta})} \quad\text{ and }\quad |\nabla w_{N,\ell} (x)|\leq \frac{C \kappa}{N (x^2+N^{-2\beta})}. 
	\end{equation}
for all $|x| \leq \ell$. As a result
\[ \int_\Lambda w_{N,\ell} (x) dx \leq \frac{C\kappa \ell^2}{N} \]
\item[iv)] There exists a constant $C > 0$ such that 
\[ |\widehat{w}_{N,\ell} (p)| \leq \frac{C \kappa}{N p^2} \]
for all $p \in \Lambda^*_+$.    
\end{enumerate}        
\end{lemma}
Using the function $w_{N,\ell} = 1- f_{N,\ell}$ defined above, we define $\eta : \Lambda^*_+ \to \bR$ through  
\begin{equation}\label{eq:ktdef} \eta_p = - N \widehat{w}_{N,\ell} (p) 
\end{equation}
{F}rom Lemma \ref{3.0.sceqlemma}, it follows that 
\begin{equation}\label{eq:etap} |\eta_p | \leq \frac{C \kappa}{p^2} \end{equation}
Hence $\eta \in \ell^2 (\Lambda^*_+)$, uniformly in $N$. With Lemma \ref{3.0.sceqlemma} (part iii)), we also obtain   
\begin{equation}\label{eq:etapN} \sum_{p \in \Lambda^*_+} p^2 |\eta_p|^2 = \| \nabla \check{\eta} \|_2^2 \leq C N^\beta \kappa^2 
\end{equation}
Sometimes, it is useful to define $\eta$ also at the point $p=0$. We set $\wt{\eta} (p) = - N \widehat{w}_\ell (p)$ for all $p \in \Lambda^*$. Then $\wt{\eta}_p = \eta_p$ for all $p \not = 0$. By Lemma \ref{3.0.sceqlemma}, part iii), we find \begin{equation}\label{eq:wteta0}  |\wt{\eta}_0| \leq N \int_\Lambda  w_{N,\ell} (x) dx \leq C \kappa \ell^2 \end{equation}
{F}rom (\ref{eq:wellp}), we obtain the following relation for the coefficients $\wt{\eta}$: 
\begin{equation}\label{eq:eta-scat}
\begin{split} 
p^2 \wt{\eta}_p +  \frac{\kappa}{2} \widehat{V} (p/N^\beta) &+ \frac{\kappa}{2N} \sum_{q \in \Lambda^*} \widehat{V} ((p-q)/N^\beta) \wt{\eta}_q  \\ &	
= N \lambda_{N,\ell} \widehat{\chi}_\ell (p) +  \lambda_{N,\ell} \sum_{q \in \Lambda^*} \widehat{\chi}_\ell (p-q) \wt{\eta}_q 
\end{split} 
\end{equation}
With $\eta \in \ell^2 (\Lambda^*_+)$, we construct, as in (\ref{eq:eBeta}), the generalized Bogoliubov transformation $e^{B(\eta)} : \cF_+^{\leq N} \to \cF^{\leq N}_+$. Furthermore, we define the excitation Hamiltonian $\cG^\beta_N : \cF^{\leq N}_+ \to \cF^{\leq N}_+$ by setting (recall the definition (\ref{eq:cLN}) of the operator $\cL^\beta_N$)  
\begin{equation}\label{eq:GN} \cG^\beta_N = e^{-B(\eta)} \cL^\beta_N e^{B(\eta)} = e^{-B(\eta)} U H_N^\beta U^* e^{B(\eta)} 
\end{equation}
In the next theorem, we collect important properties of the self-adjoint operator $\cG_N^\beta$. We will use the notation 
\[ \cK = \sum_{p\in \Lambda^*_+} p^2 a_p^* a_p , \quad \text{and } \quad \cV_N = \frac{\kappa}{2N} \sum_{p,q \in \Lambda^*_+ , r \in \Lambda^* : r \not = -p,-q} \widehat{V} (r/N^\beta) a_{p+r}^* a_q^* a_p a_{q+r} \]
for the kinetic and potential energy operators on the excitation Fock space $\cF_+^{\leq N}$. We also define $\cH_N^\beta = \cK + \cV_N$. 
\begin{theorem}\label{thm:gene}
Let $V \in L^3 (\bR^3)$ be non-negative, compactly supported and spherically symmetric and assume that the coupling constant $\kappa \geq 0$ is small enough. 
\begin{itemize}
\item[a)] Let $E_N^\beta$ denote the ground state energy of the Hamilton operator (\ref{eq:Ham0}). There exists a constant $C > 0$ such that 
\begin{equation}\label{eq:cond} 
\cG_N^\beta -  E_N^\beta \geq  \frac{1}{2} \cH_N^\beta - C 
\end{equation}
and 
\begin{equation}\label{eq:comm} \pm \left[ i\cG_N^\beta , \cN_+  \right] \leq C (\cH_N^\beta + 1) 
\end{equation}
\item[b)] For $p \in \Lambda^*_+$, we set $\sigma_p = \sinh (\eta_p)$ and $\gamma_p = \cosh (\eta_p)$. Let 
\begin{equation}\label{eq:CN}\begin{split} C^\beta_N =\; & \frac{(N-1)}{2} \kappa \widehat{V} (0) \\ & + \sum_{p \in \Lambda^*_+} \Big[ p^2 \sigma^2_p + \kappa \widehat{V} (p/N^\beta) (\sigma^2_p + \sigma_p \gamma_p ) + \frac{\kappa}{2N} \sum_{q \in \Lambda^*_+} \widehat{V} ((p-q)/N^\beta) \eta_p \eta_q \Big]
\end{split} \end{equation}
Moreover, for every $p \in \Lambda^*_+$, we define   
\begin{equation}\label{eq:FpGp} \begin{split} F_{p} &= p^2 (\sigma^2_p + \gamma^2_p) + \kappa \widehat{V} (p/N^\beta) (\sigma_p + \gamma_p)^2 \\
G_{p} &= 2 p^2 \sigma_p \gamma_p + \kappa \widehat{V} (p/N^\beta) (\sigma_p + \gamma_p)^2 + \frac{\kappa}{N} \sum_{q \in \Lambda^*} \widehat{V} ((p-q)/N^\beta) \wt{\eta}_q \end{split} \end{equation}
We use the coefficients $F_{p},G_{p}$ to construct the operator 
\[ \cQ_N^\beta = \sum_{p \in \Lambda^*_+} \Big[ F_p \, b_p^* b_p + \frac{1}{2} \, G_p \, (b_p^* b_{-p}^* + b_p b_{-p}) \Big] \]
quadratic in the $b,b^*$-fields. We define the self-adjoint operator $\cE^\beta_N$ through the identity 
\[ \cG_N^\beta = C_N^\beta + \cQ_N^\beta + \cE^\beta_N \]
Then there exists a constant $C$ such that, on $\cF_+^{\leq N}$,  
\begin{equation}\label{eq:error}  \pm \cE^\beta_N \leq C N^{(\beta-1)/2} (\cN_+ +1) (\cK +1) \end{equation}
\end{itemize}
\end{theorem}
In the last term in the definition of $G_p$, recall that $\wt{\eta}_q = - N \widehat{w}_{N,\ell} (q)$ coincides with $\eta_q$ for all $q \not = 0$ (we find it more convenient to include the contribution with $q=0$ in the definition of $G_p$). The proof of Theorem  \ref{thm:gene} represents the main technical part of our paper. It is deferred to Section \ref{sec:prop} below. In the next three sections, on the other hand, we show how to use the statement of Theorem \ref{thm:gene} to complete the proof of Theorem~\ref{thm:main}.

\section{Bounds on excitation vectors for low-energy states} \label{sec:cond}

In this section, we establish important bounds for  excitation vectors of the form $\xi_N = e^{-B(\eta)} U \psi_N \in \cF_+^{\leq N}$ associated with low energy states $\psi_N \in L^2_s (\Lambda^N)$. We begin with a simple application of the bound (\ref{eq:cond}) in Theorem \ref{thm:gene}. 
\begin{prop} \label{prop:cond}
Let $V \in L^3 (\bR^3)$ be non-negative, compactly supported and spherically symmetric and assume that the coupling constant $\kappa \geq 0$ is small enough. Let $E_N^\beta$ be the ground state energy of the Hamilton operator (\ref{eq:Ham0}). Let $\psi_N \in L^2_s (\Lambda^N)$ be a normalized wave function, with \[ \langle \psi_N , H^\beta_N \psi_N \rangle \leq E_N^\beta + \zeta  \]
for some $\zeta > 0$. Let $\xi_N = e^{-B(\eta)} U \psi_N$ be the excitation vector associated with $\psi_N$ (so that $\psi_N = U^* e^{B(\eta)} \xi_N$). Then there exists a constant $C > 0$ such that 
\begin{equation}\label{eq:condNH} \begin{split} \langle \xi_N , \cN_+  \xi_N \rangle &\leq C (1 + \zeta) \\ \langle \xi_N, \cH_N^\beta \xi_N \rangle &\leq C (1 + \zeta) \end{split} \end{equation}
\end{prop}
\begin{proof}
Since, on $\cF_+^{\leq N}$, $\cN_+  \leq (2\pi)^{-2} \cK \leq (2\pi)^{-2} \cH_N^\beta$, it is enough to show the second bound in (\ref{eq:condNH}). {F}rom (\ref{eq:cond}), we find
\[ \begin{split}  \langle  \xi_N , \cH_N^\beta \xi_N \rangle &\leq C + 2 \langle \xi_N, (\cG_N^\beta - E_N^\beta) \xi_N \rangle \\ &= C + 2 \left[ \langle \xi_N, e^{-B(\eta)} U H_N U^* e^{B(\eta)} \xi_N \rangle - E_N^\beta \right] \\ &= C + 2 \left[ \langle \psi_N, H_N \psi_N \rangle - E_N^\beta \right] \leq C (1 + \zeta) \end{split} \]
\end{proof}

To control the expectation of the error term in (\ref{eq:error}), we need stronger estimates on 
excitation vectors associated with low-energy states. We prove the required bounds in the next proposition, combining (\ref{eq:cond}) with the commutator estimate (\ref{eq:comm}). We remark that the proposition also holds with the same proof in the case $\beta =1$. 
\begin{prop} \label{prop:condNH}
Let $V \in L^3 (\bR^3)$ be non-negative, compactly supported and spherically symmetric and assume that the coupling constant $\kappa \geq 0$ is small enough. Let $E_N^\beta$ be the ground state energy of the Hamilton operator (\ref{eq:Ham0}). Let $\psi_N \in L^2_s (\Lambda^N)$ with $\| \psi_N \| = 1$ belong to the spectral subspace of the Hamiltonian (\ref{eq:Ham0}), with energies below $E_N^\beta + \zeta$, for some $\zeta > 0$. In other words, assume that 
\[ \psi_N = {\bf 1}_{(-\infty ; E_N^\beta + \zeta]} (H_N^\beta) \psi_N \]
Let $\xi_N = e^{-B(\eta)} U \psi_N$ be the excitation vector associated with $\psi_N$. Then there exists a constant $C > 0$ such that 
\[ \langle \xi_N, (\cN_+ +1) ( \cK + 1) \xi_N \rangle \leq \langle \xi_N, (\cN_+ +1) ( \cH_N^\beta + 1) \xi_N \rangle \leq C (1 + \zeta^2) \]
\end{prop}

\begin{proof}
The first inequality follows from $\cV_N \geq 0$ and since $\cK,\cV_N$ both commute with $\cN_+$. We focus on the second inequality. {F}rom (\ref{eq:cond}), we find
\[ \begin{split} \langle \xi_N, (\cN_+ +1) (\cH_N^\beta + 1) \xi_N \rangle &= \langle \xi_N (\cN_+ +1)^{1/2} (\cH_N^\beta + 1) (\cN_+ +1)^{1/2} \xi_N \rangle \\ &\leq 2 \langle \xi_N, (\cN_+ +1)^{1/2} (\wt{\cG}_N^\beta + C) (\cN_+ +1)^{1/2} \xi_N \rangle \end{split} \]
where we introduced the notation $\wt{\cG}_N^\beta = \cG_N^\beta - E^\beta_N$. Next, we commute the operator $(\wt{\cG}_N^\beta + C)$ to the right, through the factor $(\cN_+ +1)^{1/2}$. We obtain
\begin{equation}\label{eq:NHN} 
\begin{split} 
\langle \xi_N, (\cN_+ +1) (\cH_N^\beta + 1) \xi_N \rangle \leq  \; &2 \langle \xi_N, (\cN_+ +1) (\wt{\cG}_N^\beta + C) \xi_N \rangle  \\ &+ 2 \Big\langle \xi_N, (\cN_+ +1)^{1/2} \Big[ \cG_N^\beta , (\cN_+ +1)^{1/2} \Big] \xi_N \Big\rangle \end{split} \end{equation}
With Cauchy-Schwarz, the first term on the r.h.s. of (\ref{eq:NHN}) can be estimated by 
\[ \begin{split} 
\Big| \langle \xi_N, (\cN_+ +1) &(\wt{\cG}_N^\beta + C)  \xi_N \rangle \Big| \\ &\leq \langle \xi_N, (\cN_+ +1) (\wt{\cG}_N^\beta +C)^{-1} (\cN_+ +1) \xi_N \rangle^{1/2} \langle \xi_N, (\wt{\cG}_N^\beta + C)^3 \xi_N \rangle^{1/2} \end{split} \]
Since, by (\ref{eq:cond}), $(\wt{\cG}_N^\beta + C) \geq c (\cN_+ +1)$ for some $c > 0$ (choosing $C > 0$ large enough), and since $\xi_N = e^{-B(\eta)} U \psi_N$ is in the spectral subspace of $\wt{\cG}_N^\beta$, associated with the interval $[0 ; \zeta]$, we conclude that
\begin{equation}\label{eq:NL3} \left| \langle \xi_N, (\cN_+ +1) (\wt{\cG}_N^\beta + C) \xi_N \rangle \right| \leq \langle \xi_N, (\cN_+ +1) \xi_N \rangle^{1/2} (\zeta +C)^{3/2} \leq C (1 + \zeta^2) \end{equation}
where we used Prop. \ref{prop:cond}. 

As for the commutator term on the r.h.s. of (\ref{eq:NHN}), we use the representation 
\[ \frac{1}{\sqrt{z}} = \frac{1}{\pi} \int_0^\infty  \frac{1}{\sqrt{t}} \, \frac{1}{t+z} \, dt \, . \]
We find
\[ \begin{split} i \left[ \cG_N^\beta , (\cN_+ +1)^{1/2} \right] = \; &\frac{1}{\pi} \int_0^\infty dt \, \sqrt{t} \frac{1}{t+\cN_+ +1} i[\cG_N^\beta, \cN_+ ] \frac{1}{t + \cN_+ +1} \\= \; &\frac{1}{\pi} \int_0^\infty dt \, \sqrt{t} \frac{1}{t+\cN_+ +1} (\cH_N^\beta +1)^{1/2}  \cA (\cH_N^\beta + 1)^{1/2} \frac{1}{t + \cN_+ +1} \end{split} \]
where we defined the operator $\cA = (\cH_N^\beta +1)^{-1/2} i[\cG_N^\beta , \cN_+ ] (\cH_N^\beta +1)^{-1/2}$. It follows from (\ref{eq:comm}) that $\cA$ is a bounded operator, with norm $\| \cA \| \leq C$, uniformly in $N$. Hence, we have (since $[\cH_N^\beta , \cN_+ ] = 0$) 
\[ \begin{split}  \Big| \langle \xi_N, & (\cN_+ +1)^{1/2} [ \cG_N^\beta , (\cN_+ +1)^{1/2} ] \xi_N \rangle \Big| \\ \leq \; & \frac{1}{\pi} \int_0^\infty dt \, \sqrt{t} \, \Big| \Big\langle \xi_N , \frac{(\cN_+ +1)^{1/2}(\cH_N^\beta +1)^{1/2}}{t + \cN_+ +1}  \cA  \frac{(\cH_N^\beta +1)^{1/2}}{t+\cN_+ +1} \xi_N \Big\rangle \Big| \\ \leq \; & \frac{1}{\pi} \int_0^\infty dt \,  \sqrt{t} \, \frac{1}{(t+1)^2} \left\| (\cN_+ +1)^{1/2} (\cH_N^\beta + 1)^{1/2} \xi_N \right\| 
\left\| (\cH_N^\beta + 1)^{1/2} \xi_N \right\| \end{split} \]
Therefore, for every $\delta > 0$ we find $C > 0$ such that
\[ \begin{split} \Big| \langle \xi_N, (\cN_+ +1)^{1/2} &[ \cG_N^\beta , (\cN_+ +1)^{1/2} ] \xi_N \rangle \Big|  \\ &\leq \delta \langle \xi_N , (\cN_+ +1) (\cH_N^\beta + 1) \xi_N \rangle + C \langle \xi_N , (\cH_N^\beta + 1) \xi_N \rangle \\ &\leq \delta \langle \xi_N , (\cN_+ +1) (\cH_N^\beta + 1) \xi_N \rangle + C (1+\zeta) \end{split} \]
by Prop. \ref{prop:cond}. Choosing $\delta =1/2$, we conclude from (\ref{eq:NHN}) and (\ref{eq:NL3}) that 
\[ \langle \xi_N, (\cN_+ +1) (\cH_N^\beta + 1) \xi_N \rangle \leq C (1+ \zeta^2) \]
\end{proof}

\section{Diagonalization of quadratic Hamiltonian}\label{sec:diag}

{F}rom Theorem \ref{thm:gene}, we recall that the excitation Hamiltonian $\cG_N^\beta = e^{-B(\eta)} U H_N U^* e^{B(\eta)}$ can be decomposed as 
\begin{equation}\label{eq:GQE} \cG_N^\beta = C_N^\beta + \cQ_N^\beta + \cE^\beta_N 
\end{equation}
with the constant $C^\beta_N$ defined in (\ref{eq:CN}), 
the quadratic part 
\begin{equation}\label{eq:cQ} \cQ_N^\beta = \sum_{p \in \Lambda^*_+} \Big[F_p b_p^* b_p + \frac{1}{2} G_p (b_p^* b_{-p}^* + b_p b_{-p}) \Big] \end{equation}
with the coefficients $F_p,G_p$ defined in (\ref{eq:FpGp}) and with the error term $\cE^\beta_N$ satisfying
\begin{equation}\label{eq:cENNH} \pm \cE^\beta_N \leq C N^{(\beta-1)/2} (\cN_+ +1) (\cK + 1) , \end{equation}
The goal of this section is to diagonalize the quadratic operator (\ref{eq:cQ}). To this end, we will conjugate the excitation Hamiltonian $\cG_N^\beta$ with one more generalized Bogoliubov transformation. 

In order to define the Bogoliubov transformation that is going to diagonalize $\cQ_N^\beta$ we need, first of all, to establish some properties of the coefficients $F_{p}, G_{p}$ defined in (\ref{eq:FpGp}).
\begin{lemma}\label{lm:FpGp}
Let $V \in L^3 (\bR^3)$ be non-negative, compactly supported and spherically symmetric. If the coupling constant $\kappa \geq 0$ is small enough, we find a constant $C > 0$ such that $p^2 /2 \leq F_p \leq C (1+p^2)$, 
\begin{equation}
|G_p| \leq \frac{C \kappa}{p^2} 
\end{equation}  
and 
\begin{equation}\label{eq:FG} \frac{|G_p|}{F_p} \leq \frac{C}{|p|^4} \leq \frac{1}{2} 
\end{equation} 
for all $p\in \Lambda^*_+$. 
\end{lemma}

\begin{proof}
Since $\sigma^2_p + \gamma_p^2 \geq 1$, and since there is a constant $C > 0$ such that $|\widehat{V} (p/N^\beta)| \leq C$ and $|\sigma_p|, \gamma_p \leq C$ for all $p \in \Lambda^*_+$ (using the boundedness (\ref{eq:etap}) of the coefficients $\eta_p$), we easily find that $F_p \geq p^2 - C \kappa \geq p^2/ 2$, if $\kappa > 0$ is small enough (recall that $|p| > (2\pi)$ on $\Lambda^*_+$). To bound $G_p$, we write
\begin{equation}\label{eq:GP-scat} G_p = 2 p^2 \eta_p + \kappa \widehat{V} (p/N^\beta) + \frac{\kappa}{N} \sum_{q \in \Lambda^*} \widehat{V} ((p-q)/N^\beta) \wt{\eta}_q + \wt{G}_{N,p} 
\end{equation}
where $\wt{G}_{N,p}$ is such that $|\wt{G}_{N,p}| \leq C \kappa  p^{-2}$ for all $p \in \Lambda^*_+$. Here we used the fact that
\[ \begin{split} | \sigma_p \gamma_p - \eta_p | &= | \sinh (\eta_p) \cosh (\eta_p) - \eta_p |  \\ &= \left| \frac{1}{2} \sinh  (2\eta_p) - \eta_p \right| \leq \frac{1}{2} \sum_{n \geq 1} \frac{2^{2n+1} |\eta_p|^{2n+1}}{(2n+1)!} \leq  \frac{C \kappa^3}{|p|^6} \end{split} \]
and that, similarly,
\[ |(\sigma_p + \gamma_p)^2 - 1| \leq \frac{C\kappa}{p^2} \]
To estimate the other terms in (\ref{eq:GP-scat}), we use the relation (\ref{eq:eta-scat}). We obtain that
\begin{equation}\label{eq:GNp2} G_p = 2N\lambda_{N,\ell} \widehat{\chi}_\ell (p) + 2\lambda_{N,\ell} \sum_{q \in \Lambda^*} \widehat{\chi}_\ell (p-q) \wt{\eta}_q + \wt{G}_{N,p} \end{equation}
From Lemma \ref{3.0.sceqlemma}, part i), we have $N \lambda_{N,\ell} \leq C \kappa$. A simple computation shows that 
\begin{equation}\label{eq:chip} \widehat{\chi}_\ell (p) = \int_{|x| \leq \ell}  e^{-ip \cdot x} dx = \frac{4\pi}{|p|^2} \left( \frac{\sin (\ell |p|)}{|p|} - \ell \cos (\ell |p|) \right) \end{equation}
which, in particular, implies that $|\widehat{\chi}_\ell (p)| \leq C |p|^{-2}$. Similarly, we find
\[ \lambda_N \sum_{q \in \Lambda^*} \widehat{\chi}_\ell (p-q) \wt{\eta}_q = N \lambda_{N,\ell} \int_\Lambda \chi_\ell (x) w_{N,\ell} (x) e^{-ip\cdot x} dx = N \lambda_{N,\ell} \int_{|x| \leq \ell} w_{N,\ell} (x) e^{-ip \cdot x} dx \]
Switching to spherical coordinates and integrating by parts, we find (abusing slightly the notation by writing $w_{N,\ell} (r)$ to indicate $w_{N,\ell} (x)$ for $|x| = r$),   
\[  \begin{split} 
\int_{|x| \leq \ell} w_{N,\ell} (x) e^{-ip \cdot x} dx &= 2\pi \int_0^\ell dr \, r^2 w_{N,\ell} (r) \int_0^\pi d\theta \, \sin \theta \, e^{-i |p| r \cos \theta} \\ &= \frac{4\pi}{|p|} \int_0^\ell dr \, r w_{N,\ell} (r) \sin (|p|r) \\ &= - \frac{4\pi}{|p|^2} \lim_{r \to 0} r w_{N,\ell} (r) + \frac{4\pi}{|p|^2} \int_0^\ell dr \,  \frac{d}{dr} (r w_{N,\ell} (r)) \cos (|p| r) \end{split}\]
With (\ref{3.0.scbounds1}) and using again the bound $N\lambda_{N,\ell} \leq C \kappa$, we conclude that there is a constant $C > 0$ such that 
\begin{equation}\label{eq:bd-convchi}  \left| \lambda_N \sum_{q \in \Lambda^*} \widehat{\chi}_\ell (p-q) \wt{\eta}_q \right| \leq \frac{C \kappa}{p^2} \end{equation}
for all $p \in \Lambda^*_+$. {F}rom (\ref{eq:GNp2}), we obtain that there is $C > 0$ such that $|G_p| \leq C\kappa/p^2$. Together with the estimate $|F_p| \geq p^2/2$, we find the desired bound, choosing $\kappa > 0$ sufficiently small.
\end{proof} 
 
Since by Lemma \ref{lm:FpGp} we know that $|G_p|/F_p \leq 1/2$ for all $p \in \Lambda^*_+$, we can define a sequence $\tau_p$ by setting
\[ \tanh (2\tau_p) = - \frac{G_p}{F_p} \]
for all $p \in \Lambda^*_+$. Equivalently,
\begin{equation}\label{eq:taup} \tau_p = \frac{1}{4} \log \frac{1- G_p/F_p}{1+G_p/F_p} \end{equation} 
This easily implies that 
\begin{equation}\label{eq:tau-dec} |\tau_p| \leq C \frac{|G_p|}{F_p} \leq \frac{C \kappa}{|p|^{4}} \end{equation}
for all $p \in \Lambda^*_+$. Let us stress the fact that the fast decay of $\tau$ for large momenta (which will be crucial below) is a consequence of the fact that the coefficients $\eta_p$ satisfy the relation (\ref{eq:eta-scat}).

We use the coefficients $\tau_p$ (which are, by definition, real) to define a new generalized Bogoliubov transformation. As in (\ref{eq:eBeta}), we construct the antisymmetric operator 
\[ B(\tau) = \frac{1}{2} \sum_{p\in \Lambda^*_+}  \tau_p (b^*_p b^*_{-p} - b_p b_{-p}) \]
and the generalized Bogoliubov transformation
\begin{equation}\label{eq:eBtau} e^{B(\tau)} = \exp \left[ \frac{1}{2} \sum_{p\in \Lambda^*_+}  \tau_p (b^*_p b^*_{-p} - b_p b_{-p}) \right] \end{equation}
With (\ref{eq:eBtau}), we define a new excitation Hamiltonian $\cM_N^\beta : \cF_+^{\leq N} \to \cF_+^{\leq N}$ by setting 
\begin{equation}\label{eq:MN} 
\begin{split} \cM_N^\beta  &= e^{-B(\tau)} e^{-B(\eta)} U H_N U^* e^{B(\eta)} e^{B(\tau)} \\ &= e^{-B(\tau)} \cG_N^\beta e^{B(\tau)} \\ &= C_N^\beta + e^{-B(\tau)} \cQ_N^\beta e^{B(\tau)} + e^{-B(\tau)} \cE^\beta_N e^{B(\tau)} 
\end{split} \end{equation}
In the next lemma we show that, with (\ref{eq:taup}), the action of the generalized Bogoliubov transformation (\ref{eq:eBtau}) approximately diagonalizes the quadratic operator~$\cQ_N^\beta$.
\begin{lemma}\label{lm:diago}
Let $V \in L^3 (\bR^3)$ be non-negative, compactly supported and spherically symmetric and assume that the coupling constant $\kappa \geq 0$ is small enough, so that the bounds of Lemma \ref{lm:FpGp} hold true. Let $\cQ_N^\beta$ be defined as in (\ref{eq:cQ}) and $\tau_p$ as in (\ref{eq:taup}). Then   
\[ e^{-B(\tau)} \cQ_N^\beta e^{B(\tau)} = \frac{1}{2} \sum_{p \in \Lambda^*_+} \left[ -F_p + \sqrt{F_p^2 - G_p^2} \right] + \sum_{p \in \Lambda^*_+} \sqrt{F_p^2 - G_p^2} \; a_p^* a_p + \delta_{N,\beta} \]
where the self-adjoint operator $\delta_{N,\beta}$ is such that 
\begin{equation} \label{eq:pmdelta} \pm \delta_{N,\beta} \leq C N^{-1} (\cN_+ +1) (\cK + 1) \end{equation}
\end{lemma}

\begin{proof}
For $p \in \Lambda^*_+$, we define a remainder operator $d_p$ through 
\begin{equation}\label{eq:deco-tau} e^{-B(\tau)} b_p e^{B(\tau)} = \cosh (\tau_p) b_p + \sinh (\tau_p) b_{-p}^* + d_p \end{equation}
With (\ref{eq:deco-tau}) and using the short-hand notation $\wt{\gamma}_p = \cosh \tau_p, \wt{\sigma}_p = \sinh \tau_p$, we can write
\begin{equation}\label{eq:diag1} \begin{split}  e^{-B(\tau)} \cQ_N^\beta e^{B(\tau)} = \; & \sum_{p \in \Lambda^*_+} \big(F_p \wt{\sigma}_p^2 + G_p \wt{\gamma}_p \wt{\sigma}_p \big) + \sum_{p \in \Lambda^*_+} \Big[ F_p (\wt{\gamma}^2_p + \wt{\sigma}^2_p) + 2 G_p \wt{\sigma}_p \wt{\gamma}_p \Big] b_p^* b_p \\ &+ \frac{1}{2} \sum_{p \in \Lambda^*_+} \Big[ 2 F_p \wt{\gamma}_p \wt{\sigma}_p + G_p (\wt{\gamma}_p^2 + \wt{\sigma}^2_p) \Big] (b_p b_{-p} + b_p^* b_{-p}^* ) + \wt{\delta}_{N,\beta} \end{split}  \end{equation}
where 
\begin{equation}\label{eq:delta1} \begin{split} \wt{\delta}_{N,\beta} = \; &\sum_{p \in \Lambda_+^*} F_p d_p^* e^{-B(\tau)} b_p e^{B(\tau)} + \sum_{p \in \Lambda^*_+} F_p (\wt{\gamma}_p b_p^* + \wt{\sigma}_p b_p) d_p \\ &+ \frac{1}{2} \sum_{p \in \Lambda^*_+} G_p \Big[ d_p^* e^{-B(\tau)} b_{-p}^* e^{B(\tau)} + \text{h.c.} \Big] + \frac{1}{2} \sum_{p \in \Lambda^*_+} G_p \Big[ (\wt{\gamma}_p b_p^* + \wt{\sigma}_p b_{-p}) d_{-p}^* + \text{h.c.} \Big]  \end{split} \end{equation}
With the definition (\ref{eq:taup}), (\ref{eq:diag1}) simplifies, after a lengthy but straightforward computation, to
\[ e^{-B(\tau)} \cQ_N^\beta e^{B(\tau)} = \; 
\frac{1}{2} \sum_{p \in \Lambda^*_+} \Big[ - F_p + \sqrt{F_p^2 - G_p^2} \Big] + \sum_{p \in \Lambda^*_+} \sqrt{F_p^2 - G_p^2} \; b_p^* b_p + \wt{\delta}_{N,\beta} \]
{F}rom the bound $F_p \leq C (1+p^2)$ in Lemma \ref{lm:FpGp} we obtain  
\[ \begin{split} 
\Big| \sum_{p \in \Lambda^*_+} \sqrt{F_p^2 - G_p^2} \, \Big[ \langle \xi , b_p^* b_p \, \xi \rangle - \langle \xi , a_p^* a_p \xi \rangle \Big]  \Big| = &\; \Big| \frac{1}{N} \sum_{p \in \Lambda^*_+} \sqrt{F_p^2 - G_p^2} \, \langle \xi, a_p^* \cN_+  a_p \xi \rangle \Big| \\ \leq \; &\frac{1}{N} \sum_{p \in \Lambda^*_+} (p^2 + 1) \| a_p (\cN_+ +1)^{1/2} \xi \|^2 \\ = \; &\frac{1}{N} \langle \xi, (\cN_+ +1)(\cK+1) \xi \rangle \end{split} \]
for all $\xi \in \cF_+^{\leq N}$. Hence, the claim follows if we can show that the operator $\wt{\delta}_{N,\beta}$ defined in (\ref{eq:delta1}) satisfies (\ref{eq:pmdelta}). To reach this goal we notice that, by Lemma \ref{lm:conv-series}, 
\[ e^{-B(\tau)} b_p e^{B(\tau)} = \sum_{n \in \bN} \frac{(-1)^n}{n!} \text{ad}^{(n)}_{B(\tau)} (b_p) \]
and therefore 
\[ d_p = \sum_{n \in \bN} \frac{1}{(2n)!} \left[ \text{ad}^{(2n)}_{B(\tau)} (b_p) - \tau_p^{2n} b_p \right] - \sum_{n \in \bN} \frac{1}{(2n+1)!} \left[ \text{ad}^{(2n+1)}_{B(\tau)} (b_p) - \tau_p^{2n+1} b^*_{-p} \right] \]
Let us now consider the expectation of the first term on the r.h.s. of (\ref{eq:delta1}). We find
\begin{equation}\label{eq:delta11} 
\begin{split} 
\Big| \sum_{p \in \Lambda^*_+} F_p \langle d_p &\xi , e^{-B(\tau)} b_p e^{B(\tau)} \xi \rangle \Big| \\ &\leq \sum_{n,m \in \bN} \frac{1}{n!m!} \sum_{p \in \Lambda^*_+} F_p \, \| (\cN_+ +1)^{-1/2} \big[ \text{ad}^{(n)}_{B(\tau)} (b_p) - \tau_p^n b_{\alpha_n p}^{\sharp_n} \big] \xi\| \\ &\hspace{5.5cm} \times  \| (\cN_+ +1)^{1/2} \text{ad}^{(m)}_{B(\tau)} (b_p) \xi \| \end{split} \end{equation}
where $\alpha_n = 1$ and $\sharp_n = \cdot$ if $n$ is even while $\alpha_n = -1$ and $\sharp_n = *$ if $n$ is odd.   
  
{F}rom Lemma \ref{lm:indu} it follows that, for any $m \in \bN$, $\text{ad}^{(m)}_{B(\tau)} (b_p)$ is given by the sum of $2^{m} m!$ terms of the form
\begin{equation}\label{eq:typ-tau} \Lambda_1 \dots \Lambda_{i_1} N^{-k_1} \Pi^{(1)}_{\sharp,\flat} (\tau^{j_1}, \dots , \tau^{j_{k_1}} ; \tau_p^{\ell_1}) \end{equation}
where $i_1,k_1,\ell_1 \in \bN$, $j_1, \dots , j_{k_1} \in \bN \backslash \{ 0 \}$, and where each $\Lambda_j$ is either a factor $(N-\cN_+ )/N$, $(N+1-\cN_+ )/N$ or a $\Pi^{(2)}$-operator having the form
\begin{equation}\label{eq:Pi2-tau} N^{-p} \Pi^{(2)}_{\underline{\sharp}, \underline{\flat}} (\tau^{q_1}, \dots , \tau^{q_p}) \end{equation}
for some $p, q_1, \dots , q_p \in \bN \backslash \{ 0 \}$. Distinguishing the cases $\ell_1 \geq 1$ and $\ell_1 = 0$, this implies that 
\begin{equation}\label{eq:admBtau} \| (\cN_+ +1)^{1/2} \text{ad}^{(m)}_{B(\tau)} (b_p) \xi \| \leq C^m \kappa^m m!  \left[ |p|^{-4} \| (\cN_+ +1) \xi \| + \| b_p (\cN_+ +1)^{1/2} \xi \| \right] \end{equation}

Similarly, the operator $\text{ad}^{(n)}_{B(\tau)} (b_p)$ can be expanded in the sum of $2^n n!$ contributions of the form (\ref{eq:typ-tau}). Part iv) of Lemma \ref{lm:indu} implies that exactly one of these contributions will have the form 
\begin{equation}\label{eq:main-tau1} \left( \frac{N-\cN_+ }{N} \right)^{n/2} \left( \frac{N+1-\cN_+ }{N} \right)^{n/2} \tau_p^n b_p \,  \end{equation}
if $n$ is even or the form
\begin{equation}\label{eq:main-tau2} - \left( \frac{N-\cN_+ }{N} \right)^{(n+1)/2} \left( \frac{N+1-\cN_+ }{N} \right)^{(n-1)/2} \tau_p^n b^*_{-p} \,  \end{equation}
if $n$ is odd. All other terms will have either $k_1 \not = 0$ or at least one of the $\Lambda$-operator having the form (\ref{eq:Pi2-tau}). Notice that the main part of the contribution (\ref{eq:main-tau1}), (\ref{eq:main-tau2}) is exactly $\tau_p^n b_p$ if $n$ is even and $-\tau_p^n b^*_{-p}$ if $n$ is odd and it is canceled exactly by the subtraction of $\tau_p^n b_{\alpha_n p}^{\sharp_n}$. We obtain  
\begin{equation}\label{eq:ad-b-tau} 
\begin{split} \| (\cN_+ +1)^{-1/2} \big[ \text{ad}_{B(\tau)}^{(n)} &(b_p) -   \tau_p^n b_{\alpha_n p}^{\sharp_n} \big] \xi \| \\ &\leq C^n \kappa^n n! N^{-1} \left[ |p|^{-4} \| (\cN_+ +1) \xi \| + \| b_p (\cN_+ +1)^{1/2} \xi \| \right] \end{split} \end{equation}
Inserting the last inequality and (\ref{eq:admBtau}) in (\ref{eq:delta11}), and using the estimate $F_p \leq C (p^2 + 1)$ from Lemma \ref{lm:FpGp}, we conclude that the expectation of the first term on the r.h.s. of (\ref{eq:delta1}) is bounded by 
\[ \Big| \sum_{p \in \Lambda^*_+} F_p \langle d_p \xi , e^{-B(\tau)} b_p e^{B(\tau)} \xi \rangle \Big| \leq C N^{-1} \langle \xi , (\cN_+ +1) (\cK+1) \xi \rangle \]
The expectation of the second term on the r.h.s. of (\ref{eq:delta1}) can be bounded similarly. 

To bound the expectation of the third term on the r.h.s. of (\ref{eq:delta1}) we expand
\[ \begin{split} 
\Big| \sum_{p \in \Lambda^*_+} G_p &\langle d_p \xi , e^{-B(\tau)} b^*_{-p} e^{B(\tau)} \xi \rangle \Big| \\ =\; & \sum_{p \in \Lambda^*_+} |G_p| \| (\cN_+ +1)^{-1/2} d_p \xi \| \| (\cN_+ +1)^{1/2} e^{-B(\tau)} b_{-p}^* e^{B(\tau)} \xi \| \\ \leq \; &C \kappa \| (\cN_+ +1) \xi \| \sum_{n \geq 0} \frac{1}{n!} \sum_{p \in \Lambda^*_+} |p|^{-2} \Big\| (\cN_+ +1)^{-1/2} \big[ \text{ad}^{(n)}_{B(\tau)} (b_p) - \tau_p^n b_{\alpha_n p}^{\sharp_n} \big] \xi \Big\| \end{split} \]
where we used (twice) Lemma \ref{lm:Ngrow} and the bound $|G_p| \leq C \kappa |p|^{-2}$ from Lemma \ref{lm:FpGp}. Inserting (\ref{eq:ad-b-tau}), we find
\[ \Big| \sum_{p \in \Lambda^*_+} G_p \langle d_p \xi , e^{-B(\tau)} b^*_{-p} e^{B(\tau)} \xi \rangle \Big| 
\leq C N^{-1} \| (\cN_+ +1) \xi \|^2 \]
if $\kappa > 0$ is small enough. The last term on the r.h.s. of (\ref{eq:delta1}) can be controlled similarly. 
\end{proof}

Next, we prove precise estimates for the constant term and for the coefficients of the diagonal part of $\cM_N^\beta$, as defined in (\ref{eq:MN}).  
\begin{lemma}\label{lm:gs+exc}
Let $V \in L^3 (\bR^3)$ be non-negative, compactly supported and spherically symmetric and assume that the coupling constant $\kappa \geq 0$ is small enough,  so that the bounds of Lemma \ref{lm:FpGp} hold true (with $F_p, G_p$ defined as in (\ref{eq:FpGp})). Suppose that $C^\beta_N$ is defined as in (\ref{eq:CN}).  Then, for $N \to \infty$,  
\begin{equation}\label{eq:CN+bd} \begin{split} 
&C^\beta_N + \frac{1}{2} \sum_{p \in \Lambda^*_+} \left[ - F_p + \sqrt{F_p^2 - G_p^2} \right] \\ &= 4\pi (N-1) a_N^\beta + \frac{1}{2} \sum_{p \in \Lambda^*_+} \left[ - p^2  - \kappa \widehat{V} (0) + \sqrt{|p|^4 + 2 |p|^2 \kappa \widehat{V} (0)} + \frac{\kappa \widehat{V}^2 (0)}{2p^2} \right] + \cO (N^{-\alpha})  \end{split} 
\end{equation}
with $a_N^\beta$ as defined in (\ref{bN}) and for all $0< \alpha < \beta$ such that $\alpha \leq (1-\beta)/2$. Furthermore, 
on $\cF_+^{\leq N}$, we have 
\begin{equation}\label{eq:harm} \sum_{p \in \Lambda^*_+} \sqrt{F_p^2 - G_p^2} \; a_p^* a_p  = \sum_{p \in \Lambda^*_+} \sqrt{p^4 + 2 p^2 \kappa \widehat{V} (0)} a_p^* a_p + \vartheta_{N,\beta} \end{equation}
where 
\[ \pm \vartheta_{N,\beta} \leq C N^{-\alpha} (\cN_+ +1)^2 \]
for all $\alpha \leq \min (\beta, (1-\beta))$. 
\end{lemma}

\begin{proof}
{F}rom (\ref{eq:CN}) and from the definition of the coefficients $F_p,G_p$ in (\ref{eq:FpGp}) we obtain 
\[  C^\beta_N - \frac{1}{2} \sum_{p \in \Lambda^*_+} F_p =  \frac{(N-1)}{2} \kappa \widehat{V} (0) - \frac{1}{2} \sum_{p \in\Lambda^*_+} \left[ p^2 + \kappa \widehat{V} (p/N^\beta) \right] + \frac{\kappa}{2N} \sum_{p,q \in \Lambda^*_+} \widehat{V} (p/N^\beta) \eta_p \eta_q
\]
On the other hand, setting
\begin{equation}\label{eq:Ap-def} \begin{split}  A_p = \; & -2 \Big[ \kappa \widehat{V} (p/N^\beta) (\gamma_p + \sigma_p)^2 + 2 p^2 \gamma_p \sigma_p  \Big] \frac{\kappa}{N} \sum_{q \in \Lambda^*} \widehat{V} ((p-q)/N^\beta) \wt{\eta}_q \\ &- \Big[ \frac{\kappa}{N} \sum_{q \in \Lambda^*} \widehat{V} ((p-q)/N^\beta) \wt{\eta}_q \Big]^2 \end{split} \end{equation}
we find that 
\begin{equation}\label{eq:F-G} F_p^2 - G_p^2 = |p|^4 + 2p^2 \kappa \widehat{V} (p/N^\beta) + A_p 
\end{equation}
Notice that with (\ref{eq:etap}) and (\ref{eq:wteta0}), we have
\begin{equation}\label{eq:bd-b-1} \Big|  \frac{\kappa}{N} \sum_{q \in \Lambda^*} \widehat{V} ((p-q)/N^\beta) \wt{\eta}_q \Big| \leq C \frac{\kappa^2}{N} \sum_{q \in \Lambda^*} \frac{|\widehat{V} ((p-q)/N^\beta)|}{q^2 + 1} \leq C \kappa^2 N^{\beta-1} \end{equation}
which implies that \begin{equation}\label{eq:Ap-est} |A_p| \leq C N^{\beta-1}
\end{equation} 
for every fixed $p \in \Lambda^*$. Choosing $\kappa > 0$ so small that $|p|^4 + 2p^2 \kappa \widehat{V} (p/N^\beta)$ and $|p|^4 + 2p^2 \kappa \widehat{V} (p/N^\beta) + A_p$ are positive and bounded away from $0$, uniformly in $p \in \Lambda^*_+$, we observe that 
\[ \begin{split} \sqrt{|p|^4 + 2 p^2 \kappa \widehat{V} (p/N^\beta) + A_p} = \; &\sqrt{|p|^4 + 2p^2 \kappa \widehat{V} (p/N^\beta)} \\ &+ \frac{A_p}{\sqrt{|p|^4 + 2 p^2 \kappa \widehat{V} (p/N^\beta) + A_p} + \sqrt{|p|^4 + 2p^2 \kappa \widehat{V} (p/N^\beta)}}  \end{split} \]
The denominator in the last term is such that 
\[ \begin{split} 2 p^2 &\leq \sqrt{|p|^4 + 2 p^2 \kappa \widehat{V} (p/N^\beta) + A_p} + \sqrt{|p|^4 + 2p^2 \kappa \widehat{V} (p/N^\beta)} \\ &\hspace{6cm} \leq  2 p^2 \left[ 1 + C \left( \frac{A_p}{|p|^4} + \frac{\kappa \widehat{V} (p/N^\beta)}{p^2} \right) \right] \end{split} \]
This implies that 
\[ \begin{split} \frac{A_p}{2p^2} &\left[ 1- C \left( \frac{A_p}{|p|^4} + \frac{\kappa \widehat{V} (p/N^\beta)}{p^2} \right) \right] \\ &\hspace{2cm} \leq \frac{A_p}{\sqrt{|p|^4 + 2 p^2 \kappa \widehat{V} (p/N^\beta) + A_p} + \sqrt{|p|^4 + 2p^2 \kappa \widehat{V} (p/N^\beta)}} \leq \frac{A_p}{2p^2} \end{split}  \]
for all $p\in \Lambda^*_+$. Since, from (\ref{eq:Ap-est}),  
\[ \sum_{p \in \Lambda^*_+} \frac{A_p^2}{|p|^6} \leq C N^{2(\beta-1)} , \quad \text{ and } \quad \sum_{p\in \Lambda^*_+} \frac{A_p |\widehat{V} (p/N^\beta)|}{|p|^4} \leq C N^{\beta-1} \]
we conclude that 
\begin{equation}\label{eq:inter1} \begin{split} C_N^\beta + & \frac{1}{2} \sum_{p \in \Lambda^*_+} \Big[ -F_p + \sqrt{F_p^2 - G_p^2} \Big] \\ = \; & \frac{(N-1)}{2} \kappa \widehat{V} (0) + \frac{1}{2} \sum_{p \in\Lambda^*_+} \left[ - p^2 - \kappa \widehat{V} (p/N^\beta) + \sqrt{|p|^4 + 2p^2 \kappa \widehat{V} (p/N^\beta)} \right] 
\\ &+\sum_{p \in \Lambda^*_+} \left[ \frac{A_p}{4p^2} + \frac{\kappa}{2N} \sum_{q \in \Lambda^*_+} \widehat{V} ((p-q)/N^\beta) \eta_p \eta_q \right] + \cO (N^{\beta-1}) \end{split} \end{equation}
We still have to compute 
\begin{equation}\label{eq:B-def} \text{B} := \sum_{p\in \Lambda^*_+} \left[ \frac{A_p}{4p^2} + \frac{\kappa}{2N} \sum_{q \in \Lambda^*_+} \widehat{V} ((p-q)/N^\beta) \eta_p \eta_q \right] \end{equation}
To this end, we decompose $A_p = A_{1,p} + A_{2,p}$ with
\[ A_{1,p} = -\Big[ \kappa \widehat{V} (p/N^\beta) + 2 p^2 \eta_p + \frac{\kappa}{2N} \sum_{q \in \Lambda^*} \widehat{V} ((p-q)/N^\beta) \wt{\eta}_q \Big] \Big[ \frac{2\kappa}{N}  \sum_{q \in \Lambda^*_+} \widehat{V} ((p-q)/N^\beta) \wt{\eta}_q \Big] 
 \]
In other words, we define $A_{1,p}$ by replacing, in (\ref{eq:Ap-def}), $(\gamma_p + \sigma_p)^2$ by $1$ and $\gamma_p \sigma_p$ by $\eta_p$; recalling the bound (\ref{eq:bd-b-1}), we conclude that the rest term $A_{2,p}$ is such that 
\begin{equation}\label{eq:A2p} \sum_{p\in \Lambda^*_+} \frac{A_{2,p}}{p^2} \leq C N^{\beta-1} \end{equation}
{F}rom (\ref{eq:B-def}), we obtain  
\[ \begin{split} \text{B} = \; & -
\sum_{p\in \Lambda^*_+} \frac{1}{p^2} \Big[ \frac{\kappa}{2N}  \sum_{q \in \Lambda^*} \widehat{V} ((p-q)/N^\beta) \wt{\eta}_q \Big]\\ &\hspace{3cm} \times  \Big[ \kappa \widehat{V} (p/N^\beta) + p^2 \eta_p + \frac{\kappa}{2N} \sum_{q \in \Lambda^*} \widehat{V} ((p-q)/N^\beta) \wt{\eta}_q \Big] \\ &+ \cO (N^{\beta- 1})
\end{split} 
\]
Notice here that, in contrast with (\ref{eq:B-def}), the sum on the r.h.s. includes the point $q=0$ (which gives a contribution of order $N^{\beta-1}$). Using the relation (\ref{eq:eta-scat}), we find 
\[ \begin{split} \text{B} = \; & -
\sum_{p\in \Lambda^*_+} \frac{1}{p^2} \Big[ \frac{\kappa}{2N}  \sum_{q \in \Lambda^*} \widehat{V} ((p-q)/N^\beta) \wt{\eta}_q \Big]\\ &\hspace{3cm} \times  \Big[ \frac{\kappa}{2} \widehat{V} (p/N^\beta) + N\lambda_{N,\ell} \widehat{\chi}_\ell (p) + \lambda_{N,\ell} \sum_{q \in\Lambda^*} \widehat{\chi}_\ell (p-q) \wt{\eta}_q \Big] \\ &+ \cO (N^{\beta- 1})
\end{split} 
\]
With (\ref{eq:chip}) and the bounds (\ref{eq:bd-convchi}) and (\ref{eq:bd-b-1}), we can simplify the last identity to 
\begin{equation}\label{eq:iter} \begin{split} \text{B} = \; & -  
\sum_{p\in \Lambda^*_+} \frac{\kappa \widehat{V} (p/N^\beta)}{2p^2} \frac{\kappa}{2N}  \sum_{q \in \Lambda^*} \widehat{V} ((p-q)/N^\beta) \wt{\eta}_q  + \cO (N^{\beta- 1})
\\ = \; & -  
\sum_{p\in \Lambda^*_+} \frac{\kappa \widehat{V} (p/N^\beta)}{2p^2} \frac{\kappa}{2N}  \sum_{q \in \Lambda^*_+} \widehat{V} ((p-q)/N^\beta) \eta_q  + \cO (N^{\beta- 1})
\end{split} \end{equation}
since the contribution from the term with $q=0$ is of the order $N^{\beta-1}$ (and since $\wt{\eta}_q  = \eta_q$ for $q \not = 0$). The r.h.s. is of the order $N^{2\beta-1}$ (the sum over $q$ is of the order $N^{\beta-1}$, but it does not decay in $p$; summing over $p$ produces an additional factor $N^\beta$). For $\beta < 1/2$, the whole r.h.s. is negligible, in the limit $N \to \infty$. For $\beta \geq 1/2$, on the other hand, we have to expand it further. To this end, we use again the relation (\ref{eq:eta-scat}) to write
\begin{equation}\label{eq:eta-scat2}
\begin{split}  q^2 \eta_q = \; &-\frac{\kappa}{2} \widehat{V} (q/N^\beta) - \frac{\kappa}{2N} \sum_{q_2 \in \Lambda^*} \widehat{V} ((q-q_2)/N^\beta) \wt{\eta}_{q_2} \\ &+ N \lambda_{N,\ell} \widehat{\chi}_\ell (q) + \lambda_{N,\ell} \sum_{q_2 \in \Lambda^*} \widehat{\chi}_\ell (q-q_2) \wt{\eta}_{q_2}
\end{split} \end{equation}
Inserting this identity in the r.h.s. of (\ref{eq:iter}) we notice that the contribution of the last two terms on the r.h.s. is negligible, in the limit of large $N$ (after summing over $p,q,q_2$, it is of the order $N^{\beta-1} \ll 1$). Also the contribution associated with $q_2 = 0$ in the second term on the r.h.s. of (\ref{eq:eta-scat2}) vanishes, as $N \to \infty$ (it is of order $N^{2(1-\beta)}$). We arrive at 
\begin{equation}\label{eq:iter2}
\begin{split} \text{B} = \; & 
\sum_{p\in \Lambda^*_+} \frac{\kappa \widehat{V} (p/N^\beta)}{2p^2}  \frac{\kappa}{2N}  \sum_{q \in \Lambda^*} \widehat{V} ((p-q)/N^\beta) \frac{\kappa \widehat{V} (q/N^\beta)}{2q^2} \\ &+ 
\sum_{p\in \Lambda^*_+} \frac{\kappa \widehat{V} (p/N^\beta)}{2p^2}  \frac{\kappa}{2N}  \sum_{q_1 \in \Lambda^*} \widehat{V} ((p-q_1)/N^\beta) \frac{\kappa}{2 q_1^2 N} \sum_{q_2 \in \Lambda^*_+} \widehat{V} ((q_1 - q_2)/N^\beta) \eta_{q_2}  \\ &+ \cO (N^{\beta-1})
\end{split} 
\end{equation}
While the first term on the r.h.s. is of the order $N^{2\beta -1}$, the second term is now of the order $N^{3\beta -2}$. If $\beta < 2/3$, it is negligible. If instead $\beta \geq 2/3$, we iterate again the same procedure, expressing $\eta_{q_2}$ using (\ref{eq:eta-scat2}). After $k$ iterations, we obtain
\begin{equation}\label{eq:iter-last}
\begin{split} \text{B}  = \; & \sum_{j=1}^k \frac{(-1)^{j+1} \kappa^{j+2}}{2^{j+2} N^j} \sum_{p,q_1, \dots , q_j \in \Lambda^*_+} 
\frac{\widehat{V} (p/N^\beta)}{p^2} 
\frac{\widehat{V} ((p-q_1)/N^\beta)}{q_1^2} \frac{\widehat{V} ((q_1-q_2)/N^\beta)}{q_2^2} \dots \\ &\hspace{6.5cm} \dots \frac{\widehat{V} ((q_{j-1} - q_j)/N^\beta)}{q_j^2} \widehat{V} (q_j/N^\beta) \\ & + \cO (N^{(k+1)\beta - k}) +\cO ( N^{\beta-1}). 
\end{split} \end{equation}
Choosing $k = m_\beta$ the largest integer with $m_\beta \leq 1/(1-\beta) + \min (1/2 , \beta/(1-\beta))$, we obtain that $(k+1)\beta - k < - \min ((1-\beta)/2, \beta)$. Inserting (\ref{eq:iter-last}) in (\ref{eq:inter1}), we obtain 
\[ \begin{split}
C_N^\beta + &\frac{1}{2} \sum_{p \in \Lambda^*_+} \Big[ - F_p + \sqrt{F_p^2 - G_p^2} \Big] \\ = \; &\frac{(N-1)}{2} \kappa \widehat{V} (0) + \sum_{j=1}^{m_\beta} \frac{(-1)^{j+1} \kappa^{j+2}}{2^{j+2} N^j} \\ &\hspace{1cm} \times \sum_{p,q_1, \dots, q_j \in \Lambda^*_+} \frac{V(p/N^\beta)}{p^2} \frac{\widehat{V} ((p-q_1)/N^\beta)}{q_1^2} \dots \frac{\widehat{V} ((q_{j-1}- q_j)/N^\beta)}{q_j^2} \widehat{V} (q_j/N^\beta) \\ &+ \frac{1}{2} \sum_{p\in \Lambda^*_+} \Big[ -p^2 -\kappa \widehat{V} (p/N^\beta) + \sqrt{|p|^4 + 2p^2 \kappa \widehat{V} (p/N^\beta)} \Big] + \cO (N^{-\alpha}) \end{split} \]
for all $\alpha < \min (\beta, (1-\beta)/2)$. Adding and subtracting $\sum_{p \in \Lambda^*_+} \kappa^2 \widehat{V}^2 (p/N^\beta) / (4p^2)$ and comparing with the definition (\ref{bN}) of $a_N^\beta$, we get
\begin{equation}\label{eq:GS-fin} \begin{split}
C_N^\beta + &\frac{1}{2} \sum_{p \in \Lambda^*_+} \Big[ -F_p + \sqrt{F_p^2 - G_p^2} \Big] \\ = \; &4 \pi (N-1) a_N^\beta \\ &+ \frac{1}{2} \sum_{p\in \Lambda^*_+} \Big[ -p^2 -\kappa \widehat{V} (p/N^\beta) + \sqrt{|p|^4 + 2p^2 \kappa \widehat{V} (p/N^\beta)} + \frac{\kappa^2 \widehat{V}^2 (p/N^\beta)}{2p^2} \Big] + \cO (N^{-\alpha}) \end{split} \end{equation}
for every $\alpha < \min (\beta, (1-\beta)/2)$. Expanding the square root in the last sum as  
\begin{equation}\label{eq:sqrt-exp} \begin{split} &\sqrt{|p|^4 + 2 p^2 \kappa \widehat{V} (p/N^\beta)}  \\ &\hspace{1cm} = p^2 \left\{ 1 + \frac{\kappa \widehat{V} (p/N^\beta)}{p^2} - \frac{\kappa^2 \widehat{V}^2 (p/N^\beta)}{2|p|^4} \right. \\ &\hspace{2.5cm} \left. + \frac{3 \kappa^3 \widehat{V}^3 (p/N^\beta)}{|p|^6} \int_0^1 ds_1\,  s_1^2 \int_0^1 ds_2 \, s_2 \int_0^1 ds_3 \, \frac{1}{\left[1+ \frac{2 \kappa s_1 s_2 s_3 \widehat{V} (p/N^\beta)}{p^2} \right]^{5/2}} \right\} \end{split} \end{equation}
it is easy to check that
\[ \Big| -p^2 -\kappa \widehat{V} (p/N^\beta) + \sqrt{|p|^4 +2p^2 \kappa \widehat{V} (p/N^\beta)} + \frac{\kappa^2 \widehat{V}^2 (p/N^\beta)}{2p^2} \Big| \leq \frac{C}{|p|^{4}} \] 
uniformly in $N$ and, comparing (\ref{eq:sqrt-exp}) with a similar expansion with $\widehat{V} (p/N^\beta)$ replaced by $\widehat{V} (0)$, that 
\[ \begin{split} \Big| \Big[ -p^2 - \kappa \widehat{V} (p/N^\beta) &+ \sqrt{|p|^4 + 2p^2 \kappa \widehat{V} (p/N^\beta)} + \frac{\kappa^2 \widehat{V}^2 (p/N^\beta)}{2p^2} \Big] \\ &- \Big[ -p^2 -\kappa \widehat{V} (0) + \sqrt{|p|^4 + 2p^2 \kappa \widehat{V} (0)} + \frac{\kappa^2 \widehat{V}^2(0)}{2p^2} \Big] \Big| \leq C N^{-\beta} |p|^{-3}\end{split} \]
Here, we used the fact that $\kappa > 0$ is so small that the denominator in the integral on the r.h.s. of (\ref{eq:sqrt-exp}) is bounded away from $0$, uniformly in $p \in \Lambda^*_+$. Separating the sum in two regions $|p| \leq N^\beta$ and $|p| \geq N^\beta$, we conclude that 
\[ \begin{split} 
\Big| \sum_{p \in \Lambda^*_+} &\left[ - p^2 -\kappa \widehat{V} (p/N^\beta) + \sqrt{|p|^4 + 2p^2 \kappa \widehat{V} (p/N^\beta)} + \frac{\kappa^2 \widehat{V}^2 (p/N^\beta)}{2p^2} \right] \\ &\hspace{2cm} - \sum_{p \in \Lambda^*_+} \left[ - p^2 -\kappa \widehat{V} (0) + \sqrt{|p|^4 + 2p^2 \kappa \widehat{V} (0)} + \frac{\kappa^2 \widehat{V}^2 (0)}{2p^2} \right] \Big| \leq C N^{-\alpha}  \end{split} \]
for every $\alpha < \beta$. 
Inserting in (\ref{eq:GS-fin}), we obtain (\ref{eq:CN+bd}). 

Let us now prove (\ref{eq:harm}). {F}rom (\ref{eq:F-G}), we find
\begin{equation}\label{eq:B1B2} \begin{split} \sum_{p \in \Lambda^*_+} 
&\sqrt{F_p^2 - G_p^2} \, a_p^* a_p \\ = \; &\sum_{p \in \Lambda^*_+} \sqrt{|p|^4 + 2p^2 \kappa \widehat{V} (p/N^\beta) + A_p} \, a_p^* a_p \\ = \; &\sum_{p \in \Lambda^*_+} \sqrt{|p|^4 + 2p^2 \kappa \widehat{V} (p/N^\beta)} \, a_p^* a_p \\ &+ \sum_{p \in \Lambda^*_+} \frac{A_p}{\sqrt{|p|^4 +2p^2 \kappa \widehat{V} (p/N^\beta) + A_p} + \sqrt{|p|^4 + 2 p^2 \kappa \widehat{V} (p/N^\beta)}} \, a_p^* a_p \\ =: \; & \text{B}_1 + \text{B}_2 \end{split} \end{equation}
With (\ref{eq:Ap-est}), we find
\begin{equation}\label{eq:B2-bdd} \begin{split}  |\langle \xi , \text{B}_2 \xi \rangle| &\leq \sum_{p \in \Lambda^*_+} \frac{|A_p|}{\sqrt{|p|^4 +2p^2 \kappa \widehat{V} (p/N^\beta) + A_p} + \sqrt{|p|^4 + 2 p^2 \kappa \widehat{V} (p/N^\beta)}} \| a_p \xi \|^2  \\ 
&\leq C N^{\beta-1} \langle \xi, \cN_+  \xi \rangle \end{split} \end{equation}
As for $\text{B}_1$, we write
\[ \begin{split} \text{B}_1 =\; & \sum_{p \in \Lambda^*_+} \sqrt{|p|^4 +2p^2 \kappa \widehat{V} (0) + 2 N^{-\beta} p^2 \kappa \int_0^1 ds \, p \cdot \nabla \widehat{V} (s \, p / N^\beta)} \; a_p^* a_p \\ =\; & \sum_{p \in \Lambda^*_+} \sqrt{|p|^4 +2p^2 \kappa \widehat{V} (0)} \; a_p^* a_p \\ &+ \sum_{p \in \Lambda^*} \frac{2 N^{-\beta} p^2 \kappa \int_0^1 ds \, p \cdot \nabla \widehat{V} (s \, p / N^\beta)}{\sqrt{|p|^4 +2p^2 \kappa \widehat{V} (0)} + \sqrt{|p|^4 +2p^2 \kappa \widehat{V} (p/N^\beta)}} \, a_p^* a_p \\ =\; &  \sum_{p \in \Lambda^*_+} \sqrt{|p|^4 +2p^2 \kappa \widehat{V} (0)} \; a_p^* a_p + \wt{\text{B}}_{1}  
 \end{split} \]
The expectation of the second term can be bounded by 
\[ |\langle \xi, \wt{\text{B}}_{1} \xi \rangle| \leq C N^{-\beta} \sum_{p\in \Lambda^*_+} |p| \| a_p \xi \|^2 \leq  C N^{-\beta} \sum_{p\in \Lambda^*_+} (1+p^2) \| a_p \xi \|^2 \leq C N^{-\beta} \langle \xi , (\cK + 1) \xi \rangle \]
Combining the last bound with (\ref{eq:B1B2}) and (\ref{eq:B2-bdd}) we obtain
\[ \sum_{p \in \Lambda^*_+} \sqrt{F_p^2 - G_p^2} \,  a_p^* a_p = \sum_{p \in \Lambda^*_+} \sqrt{|p|^4 + 2\kappa p^2 \widehat{V} (0)} \, a_p^* a_p + \vartheta_{N,\beta} \]
where $\vartheta_{N,\beta}$ is such that 
\[ \pm \vartheta_{N,\beta} \leq C N^{-\alpha} (\cH_N^\beta + 1) ( \cN_+ +1) \]
for all $\alpha \leq \min (\beta, (1-\beta))$. 
\end{proof}

To show that the error term $\cE_N$ appearing in the decomposition (\ref{eq:GQE}) of $\cG_N^\beta$ remains negligible after conjugation with the generalized Bogoliubov transformation $e^{B(\tau)}$, we use the following lemma. 
\begin{lemma}\label{lm:err-con}
Let $V \in L^3 (\bR^3)$ be non-negative, spherically symmetric, compactly supported and suppose that the coupling constant $\kappa > 0$ is small enough. Suppose that, for $p \in \Lambda^*_+$, $\tau_p$ is defined as in (\ref{eq:taup}).  Then there exists $C > 0$ such that 
\begin{equation}\label{eq:err-con} e^{-B(\tau)} (\cN_+ +1)(\cH_N^\beta+1) e^{B(\tau)} \leq C  (\cN_+ +1)(\cH_N^\beta+1) \end{equation}
\end{lemma} 

\begin{proof}
We apply Gronwall's inequality. For $ \xi\in\cF_+^{\leq N}$ and $s \in \bR$, we consider
        \[\begin{split} &\partial_s  \langle \xi, e^{-sB(\tau)}  (\cH_N^\beta +1 )(\cN_+ +1) e^{sB(\tau)} \xi \rangle = -\langle \xi, e^{-s B(\tau)} [B(\tau),  (\cH_N^\beta+1)(\cN_+ +1) ] e^{sB(\tau)} \xi \rangle \end{split} \]
We have
\begin{equation}\label{eq:commBHN} \begin{split} 
[ B(\tau),  (\cH_N^\beta+1)(\cN_+ +1) ] = \; &[ B(\tau),  \cK](\cN_+ +1) \\ &+ [B(\tau), \cV_N] (\cN_+ +1) + (\cH_N^\beta +1)  
[B(\tau), \cN_+  ]  \end{split}
\end{equation}
Consider first the last term on the r.h.s. of (\ref{eq:commBHN}). With
\[ [ B(\tau), \cN_+  ] = \sum_{p \in \Lambda^*_+} \tau_p (b_p b_{-p} + b^*_p b^*_{-p} ) \]
we find 
\begin{equation}\label{eq:commHBN} \begin{split}
\langle \xi, e^{-sB(\tau)} (\cH_N^\beta + 1) &[B(\tau), \cN_+ ] e^{s B(\tau)} \xi \rangle \\ = \; &\sum_{p,q \in \Lambda^*_+} \tau_p q^2 \langle \xi,  e^{-sB(\tau)} a_q^* a_q (b_p b_{-p}+b_p^*b_{-p}^*) e^{sB(\tau)} \xi \rangle \\ &+ \sum_{p \in \Lambda^*_+} \tau_p \langle \xi,  e^{-sB(\tau)} \cV_N (b_p b_{-p}+b_p^*b_{-p}^*) e^{sB(\tau)} \xi \rangle \\ = \; &\text{I} + \text{II} 
\end{split} \end{equation}
where
\[ \begin{split} 
| \text{I} | &\leq \sum_{p,q \in \Lambda^*_+} \tau_p q^2 \| a_q (\cN_+ +1)^{1/2} e^{sB(\tau)} \xi \| \\ &\hspace{2cm} \times  \big[ \| (b_p b_{-p} + b_p^* b_{-p}^*) (\cN_+ +1)^{-1/2} a_q e^{s B(\tau)} \xi \| + \delta_{p,q} \| \xi \| \big] \\ &\leq C \sum_{p,q \in \Lambda^*_+} \tau_p q^2 \| a_q (\cN_+ +1)^{1/2} e^{sB(\tau)} \xi \|^2 
+ \sum_p \tau_p p^2 \| a_p (\cN_+ +1)^{1/2} \xi \| \| \xi \| 
\\ &\leq C \| \cK^{1/2} (\cN_+ +1)^{1/2} e^{sB(\tau)} \xi \|^2 + \| (\cN_+ +1) e^{sB(\tau)} \xi \| \| \xi \|\\ &\leq C \langle e^{sB(\tau)} \xi , (\cN_+ +1) ( \cK+1) e^{sB(\tau)} \xi \rangle \end{split} \]
and, expressing the potential energy operator in position space, 
\[\begin{split} | \text{II} | &\leq \sum_{p \in \Lambda^*_+} \tau_p  \int dx dy \, N^{-1+3\beta} V(N^\beta (x-y)) \Big| \langle \check{a}_x \check{a}_y e^{sB(\tau)} \xi , \check{a}_x \check{a}_y (b_p b_{-p} + b_p^* b_{-p}^*) e^{sB(\tau)} \xi \rangle \Big| 
\\ &\leq  \sum_{p \in \Lambda^*_+} \tau_p \int dx dy \, N^{-1+3\beta} V(N^\beta (x-y)) \| \check{a}_x \check{a}_y (\cN_+ +1)^{1/2} e^{sB(\tau)} \xi \| \\ &\hspace{1.2cm} \times \Big[ \| (b_p b_{-p} + b_p^* b_{-p}^*) (\cN_+ +1)^{-1/2} \check{a}_x \check{a}_y   e^{sB(\tau)} \xi \| + \| \check{a}_y e^{s B(\tau)} \xi \| + \| \check{a}_x e^{ sB(\tau)} \xi \| \Big] \\ &\leq C \int dx dy \, N^{-1+3\beta} V(N^\beta (x-y)) \| \check{a}_x \check{a}_y (\cN_+ +1)^{1/2} e^{sB(\tau)} \xi \|^2 \\ &\hspace{.4cm} + C \| (\cN_+ +1)^{1/2} e^{s B(\tau)} \xi \|^2 
\\ &\leq C \langle \xi, e^{-sB(\tau)} (\cV_N + 1) (\cN_+ +1) e^{sB(\tau)} \xi \rangle \end{split} \]
Here we used (\ref{eq:tau-dec}) to conclude that $\sum_{p\in \Lambda^*_+} |\tau_p| < \infty$. {F}rom (\ref{eq:commHBN}) we obtain that 
\begin{equation}\label{eq:BN-fin} \Big| \langle \xi, e^{-sB(\tau)} (\cH_N^\beta + 1) [B(\tau), \cN_+ ] e^{s B(\tau)} \xi \rangle \Big| \leq C 
\langle \xi, e^{-sB(\tau)}  (\cH_N^\beta +1 )(\cN_+ +1) e^{sB(\tau)} \xi \rangle 
\end{equation}

Let us now consider the first term on the r.h.s. of (\ref{eq:commBHN}). Since
\[ [B(\tau), \cK] = \sum_{p\in \Lambda^*_+} p^2 \tau_p \left( b_p b_{-p} + b_p^* b^*_{-p}\right) \]
we obtain by Cauchy-Schwarz that
\begin{equation}\label{eq:Kcom}  \begin{split} \Big|\langle \xi, &e^{-sB(\tau)} [B(\tau) ,  \cK ] (\cN_+ +1) e^{sB(\tau)} \xi \rangle \Big| \\ &\leq \sum_{p\in\Lambda^*_+} p^2 |\tau_p|  \| b_p b_{-p} e^{s B(\tau)} \xi \| \| (\cN_+ +1) e^{sB(\tau)} \xi \|   \\  &\leq C \| (\cN_+ +1) e^{s B(\tau)} \xi \| \Big[ \sum_{p\in\Lambda^*_+} \| b_p (\cN_+ +1)^{1/2} e^{sB(\tau)}\xi \|^2  \Big]^{1/2}  \\
        &\leq C \langle \xi , e^{-s B(\tau)}(\cN_+ +1)^2 e^{sB(\tau)} \xi \rangle \leq C \langle \xi, e^{-s B(\tau)}  (\cH_N^\beta+1)(\cN_+ +1) e^{sB(\tau)} \xi \rangle \end{split} \end{equation}
Here, we used the estimate (\ref{eq:tau-dec}) (to make sure that $\sum_{p\in \Lambda^*_+} p^4 \tau_p^2 < \infty$) and again the fact that, on $\cF_+^{\leq N}$, $\cN_+  \leq C \cH_N^\beta$.
       
Finally, let us consider the second term on the r.h.s. of (\ref{eq:commBHN}). It is convenient to express the potential energy operator $\cV_N$ in position space. We find
\[ \begin{split} 
\langle \xi, &e^{-s B(\tau)} [ B(\tau), \cV_N ] (\cN_+ +1) e^{s B(\tau)} \xi \rangle \\ =\; &\frac{\kappa}{2N} \int_{\Lambda \times \Lambda} dx dy \, N^{3\beta} V(N^\beta (x-y)) \check{\tau} (x-y) \langle e^{sB(\tau)} \xi, (b_x^* b_y^* +  b_x b_y) (\cN_+ +1) e^{s B(\tau)} \xi \rangle \\ &+ \frac{\kappa}{N} \int_{\Lambda \times \Lambda} dx dy N^{3\beta} V(N^\beta (x-y)) \langle e^{s B(\tau)} \xi , \big[ b_x^* b_y^* a^* (\check{\tau}_y) \check{a}_x + \text{h.c.} \big] (\cN_+ +1) e^{s B(\tau)} \xi \rangle \\ = \; & \text{III} + \text{IV} \end{split} \]
where we set $\check{\tau} (x) = \sum_{p \in \Lambda^*_+} \tau_p e^{ip \cdot x}$. Since $\| \check{\tau} \|_\infty \leq \| \tau \|_1 \leq C < \infty$ uniformly in $N$, it is simple to estimate 
\[ |\text{I}| \leq C \langle \xi, e^{-s B(\tau)} (\cV_N +1) (\cN_+ +1) e^{s B(\tau)} \xi \rangle \]
Similarly, since $\| \check{\tau}_y \|_2 = \| \check{\tau} \|_2 = \| \tau \|_2 \leq C < \infty$ independently of $y \in \Lambda$ and of $N$, we find
\[ |\text{II}| \leq C \langle \xi, e^{-s B(\tau)} (\cV_N +1) (\cN_+ +1) e^{s B(\tau)} \xi \rangle \]
We conclude therefore that
\[ \Big| \langle \xi, e^{-s B(\tau)} [ B(\tau), \cV_N ] (\cN_+ +1) e^{s B(\tau)} \xi \rangle \Big| \leq C \langle \xi, e^{-s B(\tau)} (\cV_N +1) (\cN_+ +1) e^{s B(\tau)} \xi \rangle \]
Combining this bound with (\ref{eq:BN-fin}) and (\ref{eq:Kcom}), we obtain from (\ref{eq:commBHN}) that
\[ \Big| \partial_s  \langle \xi, e^{-sB(\tau)}  (\cH_N^\beta +1 )(\cN_+ +1) e^{sB(\tau)} \xi \rangle \Big| \leq C   \langle \xi, e^{-sB(\tau)}  (\cH_N^\beta +1 )(\cN_+ +1) e^{sB(\tau)} \xi \rangle \]
By Gronwall's inequality, we arrive at (\ref{eq:err-con}).  
\end{proof}

In the next corollary, we summarize the properties of the excitation Hamiltonian $\cM_N^\beta$ defined in (\ref{eq:MN}) that follow from Lemma \ref{lm:diago}, Lemma \ref{lm:gs+exc} and Lemma \ref{lm:err-con} above. This corollary will be the starting point for the proof of Theorem \ref{thm:main} in the next section.
\begin{cor}\label{cor}
Fix $0 < \beta < 1$. Let $V \in L^3 (\bR^3)$ be non-negative, spherically symmetric and compactly supported with sufficiently small coupling constant $\kappa > 0$. Then the excitation Hamiltonian $\cM_N^\beta = e^{-B(\tau)} e^{-B(\eta)} U H_N U^* e^{B(\eta)} e^{B(\tau)} : \cF_+^{\leq N} \to \cF_+^{\leq N}$ is such that
\begin{equation}\label{eq:cor} \begin{split} \cM_N^\beta = &4\pi (N-1) a_N^\beta + \frac{1}{2} \sum_{p \in \Lambda^*_+} \left[ - p^2  - \kappa \widehat{V} (0) + \sqrt{|p|^4 + 2 |p|^2 \kappa \widehat{V} (0)} + \frac{\kappa \widehat{V}^2 (0)}{2p^2} \right] \\ &+ \sum_{p \in \Lambda^*_+} \sqrt{p^4 + 2 p^2 \kappa \widehat{V} (0)} \; a_p^* a_p + \rho_{N,\beta} \end{split} \end{equation}
where, for all $0< \alpha < \beta$ such that $\alpha \leq (1-\beta)/2$ there exists $C > 0$ with 
\[ \pm \rho_{N,\beta} \leq C N^{-\alpha} (\cN_+ +1) ( \cH_N^\beta + 1) \]
Furthermore, let $E_N^\beta$ be the ground state energy of the Hamiltonian $H_N^\beta$ and let $\psi_N \in L^2_s (\bR^{3N})$ with $\| \psi_N \|  = 1$ belong to the spectral subspace of $H_N^\beta$ with energies below $E_N^\beta + \zeta$, for some $\zeta > 0$. In other words, assume that 
\[ \psi_N = {\bf 1}_{(-\infty; E_N^\beta + \zeta]} (H_N^\beta) \psi_N \]
Let $\xi_N = e^{-B(\tau)} e^{-B(\eta)} U \psi_N \in \cF_+^{\leq N}$. Then there exists a constant $C > 0$ such that 
\[ \langle \xi_N , (\cN_+ +1) ( \cH_N^\beta + 1) \xi_N \rangle \leq C (1+\zeta^2) \]
\end{cor}

\section{Proof of Theorem \ref{thm:main}}
\label{sec:proof}

Let 
\[ \wt{E}^\beta_N = 4\pi (N-1) a_N^\beta + \frac{1}{2} \sum_{p \in \Lambda^*_+} \left[ - p^2  - \kappa \widehat{V} (0) + \sqrt{|p|^4 + 2 |p|^2 \kappa \widehat{V} (0)} + \frac{\kappa \widehat{V}^2 (0)}{2p^2} \right] \]
with $a_N^\beta$ defined as in (\ref{bN}). To prove Theorem \ref{thm:main}, we will compare the eigenvalues of $\cM_N^\beta - \wt{E}_N^\beta$ (which of course coincide with the eigenvalues of $H_N^\beta - \wt{E}_N^\beta$) with those of the diagonal quadratic operator 
\begin{equation}\label{eq:cDdef} \cD = \sum_{p \in \Lambda^*_+} \eps_p  a_p^* a_p ,  \end{equation}
acting on $\cF_+^{\leq N}$. Here we defined $\eps_p = (|p|^4 +2p^2 \kappa \widehat{V} (0))^{1/2}$ for all $p \in \Lambda^*_+$. For $m \in \bN$, let $\lambda_m$ denote the $m$-th eigenvalue of $\cM_N^\beta - \wt{E}_N^\beta$ and $\wt{\lambda}_m$ the $m$-th eigenvalue of $\cD$ (eigenvalues are counted with multiplicity). We will show that 
\begin{equation}\label{eq:lambdas} |\lambda_m - \wt{\lambda}_m | \leq C N^{-\alpha} (1+\zeta^3)  \end{equation}
for all $0< \alpha < \beta$ such that $\alpha \leq (1-\beta)/2$ and for all $m \in \bN$ such that $\wt{\lambda}_m < \zeta$. 

Since $\wt{\lambda}_0 = 0$, (\ref{eq:lambdas}) implies first of all that $E_N^\beta = \wt{E}_N^\beta + \cO (N^{-\alpha})$, for all $0< \alpha < \beta$ such that $\alpha \leq (1-\beta)/2$. Furthermore, since the eigenvalues $\wt{\lambda}_m$ of (\ref{eq:cDdef}) have the form
\[ \sum_{j=1}^k n_j \eps_{p_j} \]
for $k \in \bN$, $n_1, \dots , n_k \in \bN$ and $p_1, \dots ,p_k \in \Lambda^*_+$, (\ref{eq:lambdas}) also implies the relation (\ref{1.excitationSpectrum}) for the low-lying excitation energies of $H_N^\beta$.

To show (\ref{eq:lambdas}), we first prove an upper bound for $\lambda_m$, valid for all $m\in \bN$ with $\wt{\lambda}_m < \zeta$. To this end, we use the min-max principle, which implies that
\begin{equation}\label{eq:minmax-up} \lambda_m = \inf_{\substack{Y \subset \cF_+^{\leq N} : \\ \text{dim }Y=m} } \; \sup_{\substack{\xi \in Y : \\ \| \xi \| = 1}} \langle \xi , (\cM_N^\beta - \wt{E}_N^\beta) \xi \rangle \leq \sup_{\substack{\xi \in Y_{\cD}^m :\\ \| \xi \| = 1}} \langle \xi, (\cM_N^\beta - \wt{E}_N^\beta) \xi \rangle  
\end{equation}
where $Y_{\cD}^m$ denotes the space spanned by normalized eigenvectors $\xi_1, \dots , \xi_m$ of $\cD$, associated with the eigenvalues $\wt{\lambda}_1 \leq \dots \leq \wt{\lambda}_m < \zeta$. Without loss of generality, since $\cD$ commutes with $\cN_+ $ we may assume that $\xi_0, \dots , \xi_m$ are also eigenvectors 
of $\cN_+ $; we denote the corresponding eigenvalue by $r_0, \dots , r_m \in \bN$, i.e. $\cN_+  \xi_j = r_j \xi_j$. Since $\cD \geq c \cN_+ $, we find $r_j \leq C \zeta$. {F}rom Lemma \ref{lm:cVN} and since $\cK \leq \cD$ we obtain 
\[ \langle \xi , (\cN_+ +1) ( \cH_N^\beta + 1) \xi \rangle \leq C \langle \xi, (\cN_+ +1)^2 (\cK+1) \xi \rangle \leq C \langle \xi, (\cN_+ +1)^2 (\cD+1) \xi \rangle \leq C (1+\zeta^3) \]
for all $\xi \in Y_{\cD}^m$. With (\ref{eq:cor}), we conclude that 
\[ \langle \xi, (\cM_N^\beta -  \wt{E}_N^\beta) \xi \rangle \leq  \langle \xi, \cD \xi \rangle + C N^{-\alpha} (1+\zeta^3) \]
for all $\xi \in Y_{\cD}^m$ and all $0< \alpha < \beta$ such that $\alpha \leq (1-\beta)/2$. {F}rom (\ref{eq:minmax-up}), we obtain 
\[ \lambda_m \leq  \sup_{\xi \in Y_\cD^m : \| \xi \| =1} \langle \xi , \cD \xi \rangle + C N^{-\alpha} (1+\zeta^3) \leq \wt{\lambda}_m + C N^{-\alpha} (1+\zeta^3) \]
again for all $0< \alpha < \beta$ such that $\alpha \leq (1-\beta)/2$. 

Next, we prove the lower bound for $\lambda_m$. {F}rom the upper bound above and since we assumed that $\wt{\lambda}_m < \zeta$, we find that $\lambda_m \leq \zeta$ if $N$ is large enough. Denoting by $P_\zeta$ the spectral projection of $\cM_N^\beta - \wt{E}_N^\beta$ associated with the interval $(-\infty; \zeta]$, we find
\[ \begin{split} \lambda_m  &= \inf_{\substack{Y \subset 
\cF_+^{\leq N} :\\ \text{dim } Y = m}} \, \sup_{\substack{\xi \in Y : \\ \| \xi \| =1}} \, \langle \xi, (\cM_N^\beta - \wt{E}_N^\beta) \xi 
\rangle  
 \\ &\geq \inf_{\substack{Y \subset 
P_\zeta (\cF_+^{\leq N}) :\\ \text{dim } Y = m}} 
\, \sup_{\substack{\xi \in Y : \\ \| \xi \| =1}} \langle \xi, \cD \xi \rangle  - C N^{-\alpha} (1+\zeta^3) \\ &\geq \inf_{\substack{Y \subset 
\cF_+^{\leq N} :\\ \text{dim } Y = m}} \,
 \sup_{\substack{\xi \in Y : \\ \| \xi \| =1}} \langle \xi, \cD \xi \rangle  - C N^{-\alpha} (1+\zeta^3) \\ &= \wt{\lambda}_m - C N^{-\alpha} (1+\zeta^3) 
\end{split} \]
for all $0< \alpha < \beta$ such that $\alpha \leq (1-\beta)/2$. This concludes the proof of (\ref{eq:lambdas}) and the proof of Theorem~\ref{thm:main}.

\medskip

{\it Remark:} Theorem \ref{thm:main} states the convergence of low-lying eigenvalues of the Hamilton operator (\ref{eq:Ham0}) towards the eigenvalues of the quadratic Hamiltonian 
\begin{equation}\label{eq:Q} 
\mathcal{Q}_\infty = E_N^\beta + \sum_{p \in \Lambda^*_+} \sqrt{|p|^4 + 2p^2 \kappa \widehat{V} (0)} \, a_p^* a_p 
\end{equation}
In fact, using ideas from \cite{GS}, one can also show convergence of the corresponding eigenvectors. More precisely, for a fixed $j \in \bN$, let $P^{(j)}_{H_N^\beta}$ denote the orthogonal projection onto the subspace of $L^2_s (\bR^{3N})$ spanned by the eigenvectors associated with the $j$ smallest eigenvalues of $H_N^\beta$. Similarly, let $P_{Q}^{(j)}$ denote the orthogonal projection onto the subspace of $\cF_+^{\leq N}$ spanned by the eigenvectors associated with the $j$ smallest eigenvalues $E_N^\beta = \mu_1 \leq \mu_2 \leq \dots \leq \mu_j$ of the quadratic Hamiltonian $\cQ_\infty$. Then, assuming that $\mu_{j+1} > \mu_j$, we find
\begin{equation}\label{eq:eigvec} \Big\| e^{-B (\tau)} e^{-B(\eta)} U P^{(j)}_{H_N^\beta} U^* e^{B(\eta)} e^{B(\tau)} - P_{\cQ_\infty}^{(j)} \Big\|_\text{HS}^2 \leq \frac{C}{\mu_{j+1} - \mu_j} N^{-\alpha} \end{equation}
for all $0< \alpha < \beta$ such that $\alpha \leq (1-\beta)/2$. In particular, if $\psi_N^\beta$ denotes the ground state of the Hamiltonian $H_N^\beta$ defined in (\ref{eq:Ham0}), then there exists an appropriate 
phase $\theta \in [0 ; 2\pi )$ such that 
\[ \big\| \psi_N^{\beta} - e^{i\theta} U^* e^{B(\eta)} e^{B(\tau)} \Omega \big\|^2 \leq \frac{C}{\mu_1 - \mu_0} N^{-\alpha} \]
for all $0< \alpha < \beta$ such that $\alpha \leq (1-\beta)/2$. The proof of (\ref{eq:eigvec}) follows very closely the arguments used in Section 7 of \cite{GS}.

\section{Analysis of the excitation Hamiltonian}
\label{sec:prop}

The goal of this section is to show Theorem \ref{thm:gene}. We decompose 
\begin{equation}\label{eq:GN-dec} \cG^\beta_N = \cG_{N,\beta}^{(0)} + \cG_{N,\beta}^{(2)} + \cG_{N,\beta}^{(3)} + \cG_{N,\beta}^{(4)} \end{equation}
with
\[ \cG_{N,\beta}^{(j)} = e^{-B(\eta)} \cL_{N,\beta}^{(j)} e^{B(\eta)} \]
and with $\cL_{N,\beta}^{(j)}$ as defined in (\ref{eq:cLNj}), for $j=0,2,3,4$. We study the four contributions on the r.h.s. of (\ref{eq:GN-dec}) in the following subsections. 

First, in the next three lemmas, we collect some preliminary results that will be used later to analyze the operators $\cG_{N,\beta}^{(j)}$, $j=0,2,3,4$. 
In the first lemma, we show how to bound typical operators arising from expansions of nested commutators, as described in Lemma \ref{lm:indu} above.
\begin{lemma}\label{lm:aux}
Let $\xi \in \cF^{\leq N}$, $p,q \in \Lambda^*_+$, $i_1, i_2, k_1, k_2, \ell_1, \ell_2 \in \bN$, $j_1, \dots , j_{k_1}$, $m_1, \dots , m_{k_2} \in \bN \backslash \{0 \}$ and let $\alpha_{\ell_i} = (-1)^{\ell_i}$, for $i=1,2$. For every $s = 1, \dots , \max \{ i_1 , i_2 \}$, let $\Lambda_s$, $\Lambda'_s$ be either a factor $(N-\cN_+ )/N$, a factor $(N+1-\cN_+ )/N$ or a $\Pi^{(2)}$-operator of the form 
\begin{equation}\label{eq:Pi2-prel}
N^{-h} \Pi^{(2)}_{\underline{\sharp}, \underline{\flat}} (\eta^{z_1}, \dots , \eta^{z_h}) 
\end{equation}
for some $h \in \bN \backslash \{ 0 \}$ and $z_1, \dots, z_h \in \bN \backslash \{ 0 \}$. Suppose that the operators 
\[ \begin{split} &\Lambda_1 \dots \Lambda_{i_1} N^{-k_1} \Pi^{(1)}_{\sharp,\flat} (\eta^{j_1}, \dots , \eta^{j_{k_1}} ; \eta^{\ell_1}_p \ph_{\alpha_{\ell_1} p} ) \\ 
&\Lambda'_1 \dots \Lambda'_{i_2} N^{-k_2} \Pi^{(1)}_{\sharp',\flat'} (\eta^{m_1}, \dots , \eta^{m_{k_2}} ; \eta^{\ell_2}_q \ph_{\alpha_{\ell_2} q}) 
\end{split} \] 
appear in the expansion of $\text{ad}^{(n)}_{B(\eta)} (b_p)$ and of  $\text{ad}^{(k)}_{B(\eta)} (b_q)$, as described in Lemma \ref{lm:indu}, for some 
$n,k \in \bN$. 
\begin{itemize}
\item[i)] For any $\beta \in \bZ$, let 
\begin{equation}\label{eq:B} \text{B} = (\cN_+ +1)^{(\beta-1)/2} \Lambda_1 \dots \Lambda_{i_1} N^{-k_1} \Pi^{(1)}_{\sharp , \flat} (\eta^{j_1} , \dots , \eta^{j_{k_1}} ; \eta^{\ell_1}_p \ph_{\alpha_{\ell_1} p} ) \xi \end{equation}
and
\[\wt{B} = (\cN_+ +1)^{(\beta-1)/2}  N^{-k_1} \Pi^{(1)}_{\sharp , \flat} (\eta^{j_1} , \dots , \eta^{j_{k_1}} ; \eta^{\ell_1}_p \ph_{\alpha_{\ell_1} p} )^* \Lambda_{i_1}^* \dots \Lambda^*_1 \xi \] 
Then, we have
\begin{equation}\label{eq:D1} \| \text{B} \| , \| \wt{B}  \|  \leq C^{n} \kappa^{n} p^{-2\ell_1} \| (\cN_+ +1)^{\beta/2} \xi \| \end{equation}
If $\ell_1$ is even, we also find  
\begin{equation}\label{eq:D0} \| \text{B} \| \leq C^{n} \kappa^{n} p^{-2\ell_1} \| a_p (\cN_+ +1)^{(\beta-1)/2} \xi \| \end{equation}
\item[ii)] For $\beta \in \bZ$, let
\begin{equation}\label{eq:D} \begin{split} \text{D} = \; &(\cN_+ +1)^{(\beta-1)/2} \Lambda_1 \dots \Lambda_{i_1} N^{-k_1} \Pi^{(1)}_{\sharp,\flat} (\eta^{j_1}, \dots , \eta^{j_{k_1}} ; \eta^{\ell_1}_p \ph_{\alpha_{\ell_1} p}) \\ &\hspace{3cm} \times \Lambda'_1 \dots \Lambda'_{i_2} N^{-k_2} \Pi^{(1)}_{\sharp',\flat'} (\eta^{m_1}, \dots , \eta^{m_{k_2}} ; \eta^{\ell_2}_q \ph_{\alpha_{\ell_2} q} ) \xi \end{split} \end{equation}
Then, we have 
\begin{equation}\label{eq:E11} \| \text{D} \| \leq C^{n+k} \kappa^{n+k} p^{-2\ell_1} q^{-2\ell_2} \| (\cN_+ +1)^{(\beta + 1)/2} \xi \| \end{equation}
If $\ell_2$ is even, we find 
\begin{equation}\label{eq:E10} \| \text{D} \| \leq C^{n+k} \kappa^{n+k} p^{-2\ell_1} q^{-2\ell_2} \| a_q (\cN_+ +1)^{\beta/2} \xi \| \end{equation}
If $\ell_1$ is even, we have
\begin{equation}\label{eq:E01-b} \begin{split} \| \text{D} \| \leq \; &C^{n+k} k N^{-1} \kappa^{n+k} p^{-2(\ell_1+1)} q^{-2\ell_2} \| (\cN_+ +1)^{(\beta+1)/2} \xi \| \\ &+ C^{n+k} \kappa^{n+k}  p^{-2(\ell_1 + \ell_2)} \mu_{\ell_2} \delta_{p,-q}  \| (\cN_+ +1)^{(\beta-1)/2}\xi \| \\ &+ C^{n+k} \kappa^{n+k} p^{-2\ell_1} q^{-2\ell_2} \| a_p (\cN_+ +1)^{\beta/2} \xi \| \end{split} 
\end{equation}
where $\mu_{\ell_2} = 1$ if $\ell_2$ is odd and $\mu_{\ell_2}= 0$ if $\ell_2$ is even. If $\ell_1$ is even and either $k_1 > 0$ or $k_2 >0$ or there is at least one $\Lambda$- or $\Lambda'$-operator having the form (\ref{eq:Pi2-prel}), we obtain the improved bound 
\begin{equation}\label{eq:E01} \begin{split} \| \text{D} \| \leq \; &C^{n+k} k N^{-1} \kappa^{n+k} p^{-2(\ell_1+1)} q^{-2\ell_2} \| (\cN_+ +1)^{(\beta+1)/2} \xi \| \\ &+ C^{n+k} N^{-1} \kappa^{n+k}  p^{-2(\ell_1 + \ell_2)} \mu_{\ell_2} \delta_{p,-q} \| (\cN_+ +1)^{(\beta+1)/2}\xi \| \\ &+ C^{n+k} \kappa^{n+k} p^{-2\ell_1} q^{-2\ell_2} \| a_p (\cN_+ +1)^{\beta/2} \xi \| \end{split} \end{equation}
Finally, if $\ell_1 = \ell_2 = 0$, we can write
\begin{equation}\label{eq:Edeco}  \text{D} = \text{D}_1 (p,q) + D_2 \, a_p a_q \xi 
\end{equation}
where 
\[ \| \text{D}_1 (p,q) \| \leq C^{n+k} k N^{-1} \kappa^{n+k}  p^{-2} \| a_q (\cN_+ +1)^{\beta/2} \xi \| \]
and $D_2$ is a bounded operator on $\cF^{\leq N}_+$ with \begin{equation}\label{eq:E2-est} \| D_2^\natural \zeta \| \leq C^{n+k} \kappa^{n+k} \| (\cN_+ +1)^{(\beta-1)/2} \zeta \| \end{equation} 
for $\natural \in \{ \cdot , * \}$ and for all $\zeta \in \cF^{\leq N}_+$. If $k_1 > 0$ or $k_2 > 0$ or at least one of the $\Lambda$- or $\Lambda$'-operators has the form (\ref{eq:Pi2-prel}), we also have the improved bound \begin{equation}\label{eq:E2impr} \| D_2^{\natural} \zeta \| \leq C^{n+k} N^{-1} \kappa^{n+k}   \| (\cN_+ +1)^{(\beta+1)/2} \zeta \| \end{equation}
for $\natural \in \{ \cdot ,* \}$ and all $\zeta \in \cF^{\leq N}_+$.
\item[iii)] All the bounds in part ii) remain true if, in the definition of $\text{D}$, we replace the operator $\Lambda_1 \dots \Lambda_{i_1} N^{-k_1} \Pi^{(1)}_{\sharp,\flat} (\eta^{j_1}, \dots , \eta^{j_{k_1}} ; \eta_p^{\ell_1} \ph_{\alpha_{\ell_1}p})$ by the operator $\eta^n b^{\natural_n}_{\alpha_n p}$, where $\natural_n = \cdot$ and $\alpha_n =1$ if $n$ is even while $\natural_n = *$ and $\alpha_n =-1$ if $n$ is odd (in this case, $\ell_1 =n$). 
\end{itemize}
\end{lemma}
The proof of Lemma \ref{lm:aux}, part i) and ii) can be found in \cite[Lemma 4.1]{BBCS}. The proof of part iii) is very similar to the proof of part ii). Notice that part iii) states essentially that all bounds in part ii) remain true if in the definition of $\text{D}$, we replace all operators $\Lambda_1,\dots, \Lambda_{i_1}$ by the identity). We will use part iii) of Lemma \ref{lm:aux} in the proof of Prop. \ref{prop:K} and Prop. \ref{prop:V} below. In some occasions, it will also be important to bound 
vectors of the form (\ref{eq:D}), expressed as functions in position space. To this end, we will use the following lemma, whose proof follows closely the proof of Lemma 5.2 in \cite{BS}. 
\begin{lemma} \label{lm:aux2}
Let $\xi \in \cF^{\leq N}$, $\beta \in \bN$, $i_1, i_2, k_1, k_2, \ell_1, \ell_2 \in \bN$, $j_1, \dots, j_{k_1}, m_1, \dots , m_{k_2} \in \bN\backslash \{ 0 \}$, For every $s=1, \dots, \max \{ i_1, i_2 \}$, let $\Lambda_s, \Lambda'_s$ be either a factor $(N-\cN_+ )/N$, $(N+1-\cN_+ )/N$ or a $\Pi^{(2)}$-operator of the form
\begin{equation}\label{eq:Pi-prel4} N^{-h} \Pi^{(2)}_{\underline{\sharp}, \underline{\flat}} (\eta^{z_1}, \dots , \eta^{z_h}) \end{equation}
for some $h, z_1, \dots ,z_h \in \bN \backslash \{ 0 \}$. Suppose that the operators 
\[ \begin{split} \Lambda_1 \dots \Lambda_{i_1} N^{-k_1} \Pi^{(1)}_{\sharp, \flat} (\eta^{j_1}, \dots , \eta^{j_{k_1}} ; \check{\eta}^{\ell_1}_x ) \\ \Lambda'_1 \dots \Lambda'_{i_2} N^{-k_2} \Pi^{(1)}_{\sharp', \flat'} (\eta^{m_1}, \dots , \eta^{m_{k_2}} ; \check{\eta}^{\ell_2}_y )
 \end{split} \]
appear in the expansion of $\text{ad}^{(n)}_{B(\eta)} (\check{b}_x), \text{ad}^{(k)}_{B(\eta)} (\check{b}_y)$, respectively, for some $n,k \in \bN$. Here we use the notation $\check{\eta}^{\ell_1}_x$ for the function $z \to \check{\eta}^{\ell_1} (x-z)$, where $\check{\eta}^{\ell_1}$ denotes the Fourier transform of the function $\eta^{\ell_1}$ defined on $\Lambda^*_+$. Let
\[ \begin{split} 
\text{S} = & \; (\cN_+ +1)^{\beta/2}  \Lambda'_1 \dots \Lambda'_{i_2} N^{-k_2} \Pi^{(1)}_{\sharp', \flat'} (\eta^{m_1}, \dots , \eta^{m_{k_2}} ; \check{\eta}^{\ell_2}_y )\\ &\hspace{2cm} \times  \Lambda_1 \dots \Lambda_{i_1} N^{-k_1} \Pi^{(1)}_{\sharp, \flat} (\eta^{j_1}, \dots , \eta^{j_{k_1}} ; \check{\eta}^{\ell_1}_x ) \xi  \end{split} \]
Then we have the following bounds. If $\ell_1, \ell_2 \geq 1$, 
\begin{equation}\label{eq:S1} \| \text{S} \| \leq C^{n+k} \kappa^{n+k} \| (\cN_+ +1)^{(\beta+2)/2} \xi \|  \end{equation}
If $\ell_1 =0$ and $\ell_2 \geq 1$, \[ \| \text{S} \| \leq C^{n+k} \kappa^{n+k}  \| \check{a}_x (\cN_+ +1)^{(\beta+1)/2} \xi \| \]
If $\ell_1 \geq 1$ and $\ell_2 =0$,
\begin{equation}\label{eq:S2} \begin{split} \| \text{S} \| \leq \; &C^{n+k} \kappa^{n+k} n N^{-1} \| (\cN_+ +1)^{(\beta+2)/2} \xi \| \\ &+ C^{n+k} \kappa^{n+k-\ell_2} \mu_{\ell_2} |\check{\eta}^{\ell_2} (x-y)| \| (\cN_+ +1)^{\beta/2} \xi \| \\ &+C^{n+k} \kappa^{n+k} \| \check{a}_y (\cN_+ +1)^{(\beta+1)/2} \xi \| \end{split} \end{equation}
where $\mu_{\ell_2} = 1$ if $\ell_2$ is odd, while $\mu_{\ell_2} = 0$ if $\ell_2$ is even. If $\ell_1 \geq 1$ and $\ell_2 =0$ and we additionally assume that $k_1 > 0$ or $k_2 > 0$ or at least one of the $\Lambda$- or $\Lambda'$-operators is a $\Pi^{(2)}$-operator of the form (\ref{eq:Pi-prel4}), we obtain the improved estimate
\begin{equation}\label{eq:S2-imp} 
\begin{split} \| \text{S} \|  \leq &C^{n+k} \kappa^{n+k} n N^{-1} \| (\cN_+ +1)^{(\beta+2)/2} \xi \| \\ &+ C^{n+k} \kappa^{n+k-\ell_2} \mu_{\ell_2} N^{-1} |\check{\eta}^{\ell_2} (x-y)| \| (\cN_+ +1)^{(\beta+2)/2} \xi \| \\ &+ C^{n+k} \kappa^{n+k} \| \check{a}_y (\cN_+ +1)^{(\beta+1)/2} \xi \| \end{split} \end{equation}
Finally, if $\ell_1 = \ell_2 = 0$, 
\[ \| \text{S} \| \leq C^{n+k} \kappa^{n+k} n N^{-1} \| \check{a}_x (\cN_+ +1)^{(\beta+1)/2} \xi \| +C^{n+k} \kappa^{n+k}  \| \check{a}_x \check{a}_y (\cN_+ +1)^{\beta/2} \xi \| \]
\end{lemma}

Finally, in the next lemma, we show that the expectation of the potential energy operator is small, of the order $N^{\beta-1}$, on states with bounded expectation of $(\cN_+ +1)(\cK+1)$. This lemma will be important to show that, asymptotically the quadratic part of the generator $\cG_N^\beta$ is dominant.
\begin{lemma}\label{lm:cVN}
Suppose $V \in L^2 (\bR^3)$. Then there exists $C > 0$ such that
\[ \begin{split} \langle \xi, \cV_N \xi \rangle &= \frac{\kappa}{2N} \sum_{p,q \in \Lambda^*_+ , r\in \Lambda^*: r \not = -p,-q} \widehat{V} (r/N^\beta) a_{p+r}^* a_{q}^* a_{q+r} a_p \\ &\leq C \kappa N^{\beta-1} \| (\cK+1)^{1/2} (\cN_+ +1)^{1/2} \xi \|^2 \end{split} \]
for every $\xi \in \cF_+^{\leq N}$. Here $\cK = \sum_{p \in \Lambda^*_+} p^2 a_p^* a_p$ is the kinetic energy operator.  
\end{lemma}

\begin{proof}
We observe that
\[ \begin{split} \langle \xi, \cV_N \xi \rangle &\leq \frac{\kappa}{2N} \sum_{p,q \in \Lambda_+^* , r \in \Lambda^* : r \not = -p,-q} |\widehat{V} (r/N^\beta)| \frac{|p+r|}{|q+r|} \| a_{p+r} a_q \xi \|  \frac{|q+r|}{|p+r|} \| a_{q+r} a_p \xi \|  \\ &\leq \frac{\kappa}{N} \sum_{p,q \in \Lambda_+^* , r \in \Lambda^* : r \not = -p,-q} \frac{|\widehat{V} (r/N^\beta)|}{(q+r)^2}  (p+r)^2 \| a_{p+r} a_q \xi \|^2 \\  &\leq \kappa \left\{ \sup_{q \in \Lambda^*_+} \frac{1}{N} \sum_{r \in \Lambda^* : r \not = -q} \frac{|\widehat{V} (r/N^\beta)|}{(q+r)^2} \right\} \| \cN_+ ^{1/2} \cK^{1/2} \xi \|^2 \end{split} \]
The claim follows from the estimate  
\begin{equation}\label{eq:eta-bd1} \begin{split} 
\frac{1}{N} \sum_{r \in \Lambda^* : r \not = -q} \frac{|\widehat{V} (r/N^\beta)|}{(q+r)^2} \leq \; &\frac{\| \widehat{V} \|_\infty}{N} \sum_{r \in \Lambda^* : |r+q| \leq N^{\beta}} \frac{1}{(r+q)^2} \\ &+ \frac{1}{N} \Big[ \sum_{r \in \Lambda^*} |\widehat{V} (r/N^\beta)|^2 \Big]^{1/2} \Big[ \sum_{r \in \Lambda^* : |r+q| > N^\beta} \frac{1}{|r+q|^4} \Big]^{1/2} \\ \leq \; &C N^{\beta-1} \end{split} \end{equation}
uniformly in $q \in \Lambda_+^*$. 
\end{proof}

\subsection{Analysis of $\cG_N^{(0)}$}

{F}rom (\ref{eq:cLNj}), we have
\[ \cG_{N,\beta}^{(0)} = e^{-B(\eta)} \cL_{N,\beta}^{(0)} e^{B(\eta)} = \frac{(N-1)}{2} \kappa \widehat{V} (0) + \cE_{N,\beta}^{(0)} \]
where 
\[ \cE_{N,\beta}^{(0)} = \frac{\kappa \widehat{V} (0)}{2N} e^{-B(\eta)} \cN_+  e^{B(\eta)} - \frac{\kappa \widehat{V} (0)}{2N} e^{-B(\eta)} \cN_+ ^2 e^{B(\eta)} \]
We collect the properties of $\cE_{N,\beta}^{(0)}$ in the next proposition. 
\begin{prop}\label{prop:G0} Under the assumptions of Theorem \ref{thm:gene}, there exists $C >0$ such that, on $\cF_+^{\leq N}$,  
\begin{equation}\label{eq:GN0} \begin{split} \pm \cE_{N,\beta}^{(0)} &\leq \frac{C \kappa}{N} (\cN_+ +1)^2 \leq C \kappa (\cN_+ +1) \\
\pm [\cE_{N,\beta}^{(0)} , i \cN_+ ] &\leq C (\cN_+ +1)  \end{split} \end{equation}
\end{prop}

\begin{proof}
The first bound in (\ref{eq:GN0}) follows directly from Lemma \ref{lm:Ngrow}. To prove the second estimate in (\ref{eq:GN0}), we write
\[ \begin{split} e^{-B(\eta)} \cN_+  e^{B(\eta)} &= \cN_+  + \sum_{p \in \Lambda^*_+}  \int_0^1 ds \, e^{-sB(\eta)} [a_p^* a_p , B(\eta) ] e^{sB(\eta)} \\ &= \cN_+  + \sum_{p \in \Lambda^*_+} \eta_p \int_0^1 ds \, e^{-sB(\eta)} (b_p b_{-p} + b^*_p b^*_{-p}) e^{sB(\eta)} \end{split} \]
With Lemma \ref{lm:conv-series}, we obtain
\begin{equation}\label{eq:S-sum}  e^{-B(\eta)} \cN_+  e^{B(\eta)} = \cN_+  + \sum_{n,m \geq 0} \frac{(-1)^{n+m}}{n!m! (n+m+1)} \sum_{p \in \Lambda^*_+} \eta_p \left( \text{ad}^{(n)}_{B(\eta)} (b_p) \text{ad}^{(m)}_{B(\eta)} (b_{-p}) + \text{h.c.} \right) \end{equation}
It follows from Lemma \ref{lm:indu} that the operator 
\[ \sum_{p \in \Lambda^*_+} \eta_p \, \text{ad}^{(n)}_{B(\eta)} (b_p) \text{ad}^{(m)}_{B(\eta)} (b_{-p}) \]
can be written as the sum of $2^n n!$ terms of the form
\begin{equation}\label{eq:E-def}
\begin{split} \text{E} = \; &\sum_{p \in \Lambda^*_+} \eta_p \Lambda_1 \dots \Lambda_{i_1} N^{-k_1} \Pi^{(1)}_{\sharp, \flat} (\eta^{j_1}, \dots , \eta^{j_{k_1}} ; \eta^{\ell_1}_p \ph_{\alpha_{\ell_1} p} ) \\ &\hspace{3cm} \times \Lambda'_1 \dots \Lambda'_{i_2} N^{-k_2} \Pi^{(1)}_{\sharp, \flat} (\eta^{m_1}, \dots , \eta^{m_{k_2}} ; \eta^{\ell_2}_p \ph_{-\alpha_{\ell_2} p} )  \end{split} \end{equation}
where $i_1, i_2, k_1, k_2, \ell_1 , \ell_2 \in \bN$, $j_1, \dots , j_{k_1}, m_1, \dots , m_{k_2} \in \bN \backslash \{ 0 \}$, $\alpha_{\ell_1} = (- 1)^{\ell_1}$ and where each $\Lambda_r, \Lambda'_r$ is either a factor $(N-\cN_+ )/N$, a factor $(N+1-\cN_+ )/N$ or a $\Pi^{(2)}$-operator of the form 
\[ N^{-h} \Pi^{(2)}_{\underline{\sharp}, \underline{\flat}} (\eta^{z_1}, \dots , \eta^{z_h}) \] 
with $h, z_1, \dots, z_h \in \bN \backslash \{ 0 \}$. Lemma \ref{lm:aux}, part ii), allows us to bound matrix-elements of (\ref{eq:E-def}) by 
\begin{equation} \label{eq:S-bd} \begin{split} 
|\langle \xi_1, \text{E} \, \xi_2 \rangle | \leq \; &\sum_{p \in \Lambda^*_+} |\eta_p| \| (\cN_+ +1)^{1/2} \xi_1 \| \\ &\hspace{.5cm} \times   \| (\cN_+ +1)^{-1/2} \Lambda_1 \dots \Lambda_{i_1}   N^{-k_1} \Pi^{(1)}_{\sharp, \flat} (\eta^{j_1}, \dots , \eta^{j_{k_1}} ; \eta^{\ell_1} \ph_{\alpha_{\ell_1} p} ) \\ &\hspace{2.5cm} \times    \Lambda'_1 \dots \Lambda'_{i_2} N^{-k_2} \Pi^{(1)}_{\sharp, \flat} (\eta^{m_1}, \dots , \eta^{m_{k_2}} ; \eta^{\ell_2} \ph_{-\alpha_{\ell_2} p} )  \xi_2 \| \\ \leq \; &C^{n+m} \kappa^{n+m+1} \| (\cN_+ +1)^{1/2} \xi_1 \| \sum_{p \in  \Lambda^*_+} \left\{ \|p|^{-4} \| (\cN_+ +1)^{1/2} \xi_2 \| + |p|^{-2} \| a_p \xi_2 \| \right\} 
\\ \leq \; &C^{n+m} \kappa^{n+m+1} \| (\cN_+ +1)^{1/2} \xi_1 \| \| (\cN_+ +1)^{1/2} \xi_2 \| \end{split} \end{equation}
Since $[\cN_+ , \text{E}]$ has again the form $\text{E}$, up to a multiplicative constant bounded by $(n+m)$,  the bound (\ref{eq:S-bd}), with (\ref{eq:S-sum}), also implies that 
\begin{equation}\label{eq:comm1} \Big|\langle \xi_1, \left[ e^{-B(\eta)} \cN_+  e^{B(\eta)} , \cN_+  \right] \xi_2 \rangle \Big| \leq C \kappa \| (\cN_+ +1)^{1/2} \xi_1 \| \| (\cN_+ +1)^{1/2} \xi_2 \| \end{equation}
for all $\xi_1, \xi_2 \in \cF_+^{\leq N}$. With Lemma \ref{lm:Ngrow}, we obtain 
\[ \begin{split} \Big| \big\langle \xi, \left[ e^{-B(\eta)} \cN_+ ^2 e^{B(\eta)} , \cN_+  \right] \xi_2 \big\rangle \Big| &= \Big| \langle \xi, e^{-B(\eta)} \cN_+   e^{B(\eta)} \left[ e^{-B(\eta)} \cN_+  e^{B(\eta)} , \cN_+  \right] \xi \rangle \\ &\hspace{2cm} + \langle \xi , \left[ e^{-B(\eta)} \cN_+  e^{B(\eta)} , \cN_+  \right] e^{-B(\eta)} \cN_+  e^{B(\eta)} \xi \rangle \Big| \\ &\leq C \kappa \| (\cN_+ +1)^{1/2} e^{-B(\eta)} \cN_+  e^{B(\eta)}  \xi \| \| (\cN_+ +1)^{1/2} \xi \| \\ &\leq C \kappa \| (\cN_+ +1)^{3/2} \xi \| \| (\cN_+ +1)^{1/2} \xi \| \\ &\leq C \kappa N \| (\cN_+ +1)^{1/2} \xi \|^2 \end{split} \]
Together with (\ref{eq:comm1}), this concludes the proof of the second estimate in (\ref{eq:GN0}).

\end{proof}

\subsection{Analysis of $\cG_N^{(2)}$}

{F}rom (\ref{eq:cLNj}), we have 
\[ \cG_{N,\beta}^{(2)} = e^{-B(\eta)} \cL_{N,\beta}^{(2)} e^{B(\eta)}  = e^{-B(\eta)} \cK e^{B(\eta)} + e^{-B(\eta)} \cL_{N,\beta}^{(V)} e^{B(\eta)}  \]
where $\cK = \sum_{p \in  \Lambda_+^*} p^2 a_p^* a_p$ and 
\begin{equation}\label{eq:wtL2N} 
\cL^{(V)}_{N,\beta} = \sum_{p \in \Lambda_+^*} \kappa \widehat{V} (p/N^\beta) \left[ b_p^* b_p - \frac{1}{N} a_p^* a_p \right] + \frac{\kappa}{2} \sum_{p \in \Lambda_+^*} \widehat{V} (p/N^\beta) \big[ b_p^* b_{-p}^* + b_p b_{-p} \big] 
\end{equation}

We study first the contribution arising from the kinetic energy operator $\cK$. We define the operator $\wt{\cE}_{N,\beta}^{(K)}$ through 
\begin{equation}\label{eq:E_NK} e^{-B(\eta)} \cK e^{B(\eta)} = \cK + \sum_{p \in \Lambda^*_+} p^2 \eta_p^2 + \sum_{p\in \Lambda^*_+} p^2 \eta_p \big[ b_p^* b_{-p}^* + b_p b_{-p} \big] + \wt{\cE}_{N,\beta}^{(K)} \end{equation}
To prove part b) of Theorem \ref{thm:gene}, we need to keep track of more order one terms arising from the conjugation of $\cK$. We define the operator $\cE_{N,\beta}^{(K)}$ through 
\begin{equation}\label{eq:wtE_NK} \begin{split} e^{-B(\eta)} \cK e^{B(\eta)} = \; &\cK + \sum_{p \in \Lambda^*_+} \Big[ p^2 \sigma_p^2 + p^2 \sigma_p \gamma_p  \big( b_p b_{-p} + b_p^* b_{-p}^* \big) + 2 p^2 \sigma_p^2 b_p^* b_p \Big] + \cE_{N,\beta}^{(K)} \end{split} \end{equation}
In the next proposition, we study the properties of the error terms $\wt{\cE}_{N,\beta}^{(K)}$, 
$\cE_{N,\beta}^{(K)}$. 
\begin{prop} \label{prop:K}
Under the assumptions of Theorem \ref{thm:gene}, for every $\delta > 0$ there exists $C >0$ such that, on $\cF_+^{\leq N}$, 
\begin{equation}\label{eq:bdK} \begin{split} 
\pm \, \wt{\cE}_{N,\beta}^{(K)} &\leq \delta \cH_N^\beta + C \kappa (\cN_+ +1) \\ \pm \, \left[ \wt{\cE}_{N,\beta}^{(K)} , i \cN_+  \right] &\leq C (\cH_N^\beta + 1) \end{split}
\end{equation}
Furthermore, there exists $C > 0$ such that 
\begin{equation}\label{eq:bdK2} 
\pm \, \cE_{N,\beta}^{(K)} \leq C N^{\beta -1} (\cN_+ +1) (\cK +1) \end{equation}
\end{prop}

\begin{proof}
We compute 
\[ \begin{split} e^{-B(\eta)} \cK e^{B(\eta)} = \; &\cK + \int_0^1 ds \frac{d}{ds} \, e^{-sB(\eta)} \cK e^{s B(\eta)} \\ = \; &\cK + \int_0^1 ds \, e^{-s B(\eta)} [\cK , B(\eta) ] e^{s B(\eta)} \\ = \; &\cK + \int_0^1 ds \sum_{p \in \Lambda^*_+} p^2 \eta_p \, e^{-s B(\eta)} \left(  b_p b_{-p}  + b_p^* b_{-p}^* \right)  e^{s B(\eta)} \end{split} \]
With Lemma \ref{lm:conv-series} we find
\[ \begin{split} e^{-B(\eta)} \cK e^{B(\eta)} = \; &\cK + \int_0^1 ds \sum_{n,m \geq 0} \frac{(-1)^{n+m}}{n!m!}   \sum_{p \in \Lambda^*_+} p^2 \eta_p \left[ \text{ad}^{(n)}_{sB(\eta)} (b_p) \text{ad}^{(m)}_{sB(\eta)} ( b_{-p}) + \text{h.c.} \right]  \\
= \; &\cK + \int_0^1 ds \sum_{n,m \geq 0} \frac{(-1)^{n+m}}{n!m!}  \sum_{p \in \Lambda^*_+} p^2 \eta_p \\ &\hspace{1cm} \times \Big\{ \Big[ s^n \eta_p^n b_{\alpha_n}^{\sharp_n} + \text{ad}_{sB(\eta)}^{(n)} (b_p) - s^n \eta_p^n b_{\alpha_n}^{\sharp_n} \Big] \\ &\hspace{3cm} \times  \Big[ s^m \eta_p^m b_{\alpha_m}^{\sharp_m} + \text{ad}_{sB(\eta)}^{(m)} (b_p) - s^m \eta_p^m b_{\alpha_m}^{\sharp_m} \Big] + \text{h.c.} \Big\} 
\end{split} \]
where we defined $\alpha_n = +1$ and $\sharp_n =\cdot$ if $n$ is even while $\alpha_n = -1$ and $\sharp_n = *$ if $n$ is odd. Integrating over $s$, and using \[ \begin{split} \eta_p \int_0^1  \big( \cosh^2 (s \eta_p) + \sinh^2 (s\eta_p) \big) \, ds &= \cosh (\eta_p) \sinh (\eta_p) \\ 2\eta_p \int_0^1  \sinh (s\eta_p) \cosh (s\eta_p) \, ds &=  \sinh^2 (\eta_p)  \end{split} \]
we easily find, with the notation $\gamma_p = \cosh \eta_p$ and $\sigma_p = \sinh \eta_p$,
\[ \begin{split} e^{-B(\eta)} &\cK e^{B(\eta)} \\ = \; &\cK + \sum_{p \in \Lambda^*_+} p^2 \sigma^2_p + \sum_{p \in \Lambda^*_+} p^2 \gamma_p \sigma_p \big( b_p b_{-p} + b_p^* b_{-p}^* \big) + 2 \sum_{p \in \Lambda^*_+}  p^2 \sigma^2_p b_p^* b_p \\ &+ \sum_{n,m \geq 0} \frac{(-1)^{n+m}}{n!m!(n+m+1)} \sum_{p\in \Lambda^*_+} p^2  \eta_p^{n+1} b_{\alpha_n p}^{\sharp_n} \left[ \text{ad}_{B(\eta)}^{(m)} (b_{-p}) - \eta_p^m b_{-\alpha_m p}^{\sharp_m} \right] + \text{h.c.} \\ &+ \sum_{n,m \geq 0} \frac{(-1)^{n+m}}{n!m!(n+m+1)} \sum_{p\in \Lambda^*_+} p^2 \left[ \text{ad}_{B(\eta)}^{(n)} (b_p) - \eta_p^n b_{\alpha_n p}^{\sharp_n} \right] \eta_p^{m+1} b_{-\alpha_m p}^{\sharp_m}   + \text{h.c.}
\\ &+ \sum_{n,m \geq 0} \frac{(-1)^{n+m}}{n!m!(n+m+1)}\sum_{p\in \Lambda^*_+} p^2 \eta_p \left[ \text{ad}_{B(\eta)}^{(n)} (b_p) -  \eta_p^n b_{\alpha_n p}^{\sharp_n} \right] \\ &\hspace{7cm} \times \left[ \text{ad}_{B(\eta)}^{(m)} (b_{-p}) - \eta_p^m b_{-\alpha_m p}^{\sharp_m} \right]  + \text{h.c.} \\ 
=: \; & \cK + \sum_{p \in \Lambda^*_+} \Big[ p^2 \sigma^2_p + p^2  \gamma_p \sigma_p \big( b_p b_{-p} + b^*_p b^*_{-p} \big) + 2p^2 \sigma_p^2 b_p^* b_p \Big] \\ &+ \cE^{(K)}_1 + \cE^{(K)}_2 + \cE^{(K)}_3 \, .
\end{split} \]
Comparing with (\ref{eq:E_NK}) and (\ref{eq:wtE_NK}), we conclude that $\cE_{N,\beta}^{(K)} = \cE^{(K)}_1 + \cE^{(K)}_2 + \cE^{(K)}_3$ and 
\[ \begin{split} \wt{\cE}_{N,\beta}^{(K)} = \; &\sum_{p \in \Lambda^*_+} p^2 \big[\sigma_p^2 - \eta_p^2 \big] + 2p^2 \sigma_p^2 b_p^* b_p +  p^2 \big[ \sigma_p \gamma_p - \eta_p \big] \big[ b_p^* b_{-p}^* + b_p b_{-p} \big] \\ &+ \cE^{(K)}_1 + \cE^{(K)}_2 + \cE^{(K)}_3 \\
=: \; & \cE_0^{(K)} + \cE^{(K)}_1 + \cE^{(K)}_2 + \cE^{(K)}_3 \end{split} \]
Since, by (\ref{eq:etap}), $|\sigma_p^2 - \eta_p^2| \leq C \kappa^2 |p|^{-8}$, $p^2 \sigma_p^2 \leq C \kappa^2$ and $|\sigma_p \gamma_p - \eta_p| \leq C \kappa^3 |p|^{-6}$, it is easy to check that 
\[ \begin{split}  | \langle \xi , \cE_0^{(K)} \xi \rangle | &\leq C \kappa \| (\cN_+ +1)^{1/2} \xi \|^2 \\  |\langle \xi , [ \cE_0^{(K)} , \cN_+  ] \xi \rangle | &\leq C \kappa \| (\cN_+ +1)^{1/2} \xi \|^2 \end{split} \]
Hence, Proposition \ref{prop:K} follows if we can show that the three error terms $\cE_1^{(K)}, \cE^{(K)}_2, \cE_3^{(K)}$ satisfy the three bounds in (\ref{eq:bdK}), (\ref{eq:bdK2}). 

We consider first the term $\cE^{(K)}_1$. According to Lemma \ref{lm:indu}, the operator
\begin{equation}\label{eq:EK1-tms} \sum_{p \in \Lambda^*_+} p^2 \eta_p^{n+1} b^{\sharp_n}_{\alpha_n p} \big[\text{ad}^{(m)}_{B(\eta)} (b_{-p}) - \eta_p^m b^{\sharp_m}_{-\alpha_m p}\big] \end{equation}
is given by the sum of one term of the form 
\begin{equation}\label{eq:F1}\begin{split} 
\text{F}_{1} = \; & \sum_{p \in \Lambda^*_+} p^2 \eta_p^{m+n+1} b^{\sharp_n}_{\alpha_n p} \\ &\hspace{2cm} \times \left\{ \left( \frac{N-\cN_+ }{N} \right)^{\frac{m+(1-\alpha_m)/2}{2}} \left( \frac{N+1-\cN_+ }{N} \right)^{\frac{m-(1-\alpha_m)/2}{2}} -1 \right\} b^{\sharp_m}_{-\alpha_m p} \end{split} \end{equation}
and of $2^m m! -1$ terms of the form
\begin{equation}\label{eq:F2} \text{F}_{2} = \sum_{p \in \Lambda^*_+} p^2 \eta_p^{n+1} b^{\sharp_n}_{\alpha_n p} \Lambda_1 \dots \Lambda_{i_1} N^{-k_1} \Pi^{(1)}_{\sharp,\flat} (\eta^{j_1} , \dots , \eta^{j_{k_1}} ; \eta^{\ell_1}_p \ph_{-\alpha_{\ell_1} p}) \end{equation}
where $i_1, k_1, \ell_1 \in \bN$, $j_1, \dots , j_{k_1} \in \bN \backslash \{ 0 \}$, $\alpha_{\ell_1} = (-1)^{\ell_1}$ and where each $\Lambda_r$ is either a factor $(N-\cN_+ )/N$, $(N+1-\cN_+ )/N$ or a $\Pi^{(2)}$-operator of the form 
\begin{equation}\label{eq:Pi2-K} N^{-h} \Pi^{(2)}_{\underline{\sharp}, \underline{\flat}} (\eta^{z_1}, \dots , \eta^{z_h}) \end{equation}
with $h, z_1, \dots, z_h \in \bN \backslash \{ 0 \}$. Furthermore, since we are considering the term (\ref{eq:F1}) separately, each term of the form (\ref{eq:F2}) must have either $k_1 > 0$ or it must contain at least one $\Lambda$-operator of the form (\ref{eq:Pi2-K}) for some $p > 0$. 

To estimate (\ref{eq:F1}), we define 
\[ f(\cN_+ ) = \left\{ 1 - \left( \frac{N-\cN_+ }{N} \right)^{\frac{m+(1-\alpha_m)/2}{2}} \left( \frac{N+1-\cN_+ }{N} \right)^{\frac{m-(1-\alpha_m)/2}{2}} \right\} \]
and we notice that 
\begin{equation}\label{eq:fN} -C m /N \leq f(\cN_+ ) \leq C m \cN_+ /N \end{equation}
Since $f (\cN_+) = 0$ when $m=0$, distinguishing the two cases $n+m \geq 2$ and $n=0,m=1$ we conclude that 
\begin{equation}\label{eq:F1-1} \begin{split} 
|\langle \xi , \text{F}_1 \xi \rangle | \leq \; &C^{n+m+1} \kappa^{n+m+1} \sum_{p \in \Lambda^*_+} \left\{ \frac{(m+1)}{N |p|^4} \| (\cN_+ +1) \xi \|^2 + \frac{m}{N p^2} \| b_p (\cN_+ +1)^{1/2} \xi \|^2 \right\}  \\ &+ \frac{m}{N} \| (\cN_+ +1)^{1/2} \xi \|^2 \sum_{p \in \Lambda^*_+} p^2 \eta^2_p \\ \leq \; &C^{n+m+1} \kappa^{n+m+1} (m+1) \Big\{ N^{-1} \| (\cN_+ +1) \xi \|^2 + N^{\beta -1} \| (\cN_+ +1)^{1/2} \xi \|^2 \Big\} \end{split} \end{equation}
for all $n,m \in \bN$ (the second line bounds the term with the commutator $[b_p, b_p^*]$ arising when $n=0$ and $m=1$). Since $\cN_+ \leq N$ on $\cF_+^{\leq N}$, (\ref{eq:F1-1}) also implies that
\begin{equation}\label{eq:F1-2}
|\langle \xi , \text{F}_1 \xi \rangle | \leq C^{n+m+1} \kappa^{n+m+1} (m+1) \| (\cN_+ + 1)^{1/2} \xi \|^2 
\end{equation}
Eq. (\ref{eq:F1-1}) will be used in the proof of (\ref{eq:bdK2}), while (\ref{eq:F1-2}) will be used to show (\ref{eq:bdK}).

Let us now consider the expectation of (\ref{eq:F2}). First, assume that $\ell_1 + n \geq 1$. Then, Lemma \ref{lm:aux}, part iii), implies that 
\begin{equation}\label{eq:F2-1}\begin{split} |\langle \xi , \text{F}_2 \xi \rangle | \leq \; &\sum_{p \in \Lambda^*_+} p^2 |\eta_p| \| (\cN_+ +1)^{1/2} \xi \| \\  &\hspace{.1cm} \times \| (\cN_+ +1)^{-1/2} b^{\sharp_n}_{\alpha_n p} \Lambda_1 \dots \Lambda_{i_1} N^{-k_1} \Pi^{(1)}_{\sharp,\flat} (\eta^{j_1}, \dots, \eta^{j_{k_1}} ; \eta_p^{\ell_1} \ph_{-\alpha_{\ell_1} p} ) \xi \| \\
\leq \; & C^{n+m} \kappa^{n+m+1} \| (\cN_+ +1)^{1/2} \xi \| \\ &\hspace{.1cm} \times \sum_{p \in \Lambda^*_+} \left\{ \frac{(1+m/N)}{|p|^4} \| (\cN_+ +1)^{1/2} \xi \| + \frac{1}{|p|^2} \| a_p \xi \| \right\} \\ &+ C^{n+m} \kappa^{n+m-1} \| (\cN_+ +1)^{1/2} \xi \|^2 \frac{1}{N} \sum_{p \in \Lambda^*_+} p^2 \eta_p^2  \\ 
\leq \; &C^{n+m} \kappa^{n+m+1} (m+1) \| (\cN_+ +1)^{1/2} \xi \|^2 \end{split} \end{equation}
by (\ref{eq:etapN}), which will be used in the proof of (\ref{eq:bdK}). Also here we will need a slightly different estimate to show (\ref{eq:bdK2}). Using again Lemma \ref{lm:aux}, part iii), under the assumption $\ell_1 + n \geq 1$, we find 
\begin{equation}\label{eq:F2-2} \begin{split} |\langle \xi , \text{F}_2 \xi \rangle | \leq \; &\sum_{p \in \Lambda^*_+} p^2 |\eta_p| \| (\cN_+ +1) \xi \| \\  &\hspace{.1cm} \times \| (\cN_+ +1)^{-1} b^{\sharp_n}_{\alpha_n p} \Lambda_1 \dots \Lambda_{i_1} N^{-k_1} \Pi^{(1)}_{\sharp,\flat} (\eta^{j_1}, \dots, \eta^{j_{k_1}} ; \eta_p^{\ell_1} \ph_{-\alpha_{\ell_1} p} ) \xi \| \\
\leq \; & \frac{C^{n+m} \kappa^{n+m+1}}{N} \| (\cN_+ +1) \xi \| \\ &\hspace{.1cm} \times \sum_{p \in \Lambda^*_+} \left\{ \frac{(1+m)}{|p|^4} \| (\cN_+ +1) \xi \| + \frac{1}{|p|^2} \| a_p (\cN_+ +1)^{1/2} \xi \| \right\} \\ &+ \frac{C^{n+m} \kappa^{n+m-1}}{N} \| (\cN_+ +1) \xi \|  \| \xi \| \sum_{p \in \Lambda^*_+} p^2 \eta_p^2  \\ 
\leq \; & C^{n+m} \kappa^{n+m+1} (m+1) N^{\beta-1} \| (\cN_+ +1) \xi \|^2 \end{split} \end{equation}

In the case $n=\ell_1=0$, Lemma \ref{lm:aux}, part iii), allows us to write
\[ \langle \xi , \text{F}_{2} \xi \rangle = \sum_{p \in \Lambda_+^*} p^2 \eta_p \langle \xi , \text{D}_{1} (p) \rangle + \sum_{p \in \Lambda_+^*} p^2 \eta_p \langle \xi ,\text{D}_2 \, a_p a_{-p} \xi \rangle  \]
where
\[ \| (\cN_+ +1)^{-1} \text{D}_1 (p) \| \leq C^{m} \kappa^{m} m N^{-1} |p|^{-2} \| a_p (\cN_+ +1)^{-1/2} \xi \| \]
and $\| \text{D}^*_2 \xi \| \leq C^{m} \kappa^{m} N^{-1}\| (\cN_+ +1) \xi \|$. Hence, in this case,
\[ \begin{split} |\langle \xi, \text{F}_{2} \xi \rangle| \leq \; &\frac{C^{m} \kappa^{m+1} m}{N} \| (\cN_+ +1) \xi \| \sum_{p \in \Lambda_+^*} |p|^{-2} \| a_p (\cN_+ +1)^{-1/2} \xi \| \\ &+ \Big| \sum_{p \in \Lambda^*_+} p^2 \eta_p \langle \xi , \text{D}_{2} \, a_p a_{-p} \xi \rangle  \Big| \\
\leq \; &C^{m} \kappa^{m+1} m N^{-1} \| (\cN_+ +1) \xi \|^2 + \Big| \sum_{p \in \Lambda^*_+} p^2 \eta_p \langle  \text{D}^*_2 \, \xi ,  a_p a_{-p} \xi \rangle  \Big| \end{split} \]
To control the last term, we  use (\ref{eq:eta-scat}) to replace
\begin{equation}\label{eq:eta-scat2b} 
p^2 \eta_p = -\frac{\kappa}{2} \widehat{V} (p/N^\beta) - \frac{\kappa}{2N} \sum_{q \in \Lambda^*} \widehat{V} ((p-q)/N^\beta) \wt{\eta}_q + N \lambda_{N,\ell} \widehat{\chi}_\ell (p) + \lambda_{N,\ell} \sum_{q \in \Lambda^*} \widehat{\chi}_\ell (p-q) \wt{\eta}_q 
 \end{equation}
To bound the contribution proportional to $\kappa \widehat{V} (p/N^\beta)$, we switch to position space. We find 
\[ \begin{split} \Big| \kappa \sum_{p \in \Lambda^*_+} \widehat{V} (p/N^\beta) &\langle \text{D}^*_2 \, \xi ,  a_p a_{-p} \xi \rangle \Big| \\ 
&= \left| \kappa \int_{\Lambda \times \Lambda} dx dy N^{3\beta} V(N^\beta (x-y))\langle \text{D}^*_2 \xi,  \check{a}_x \check{a}_y \xi \rangle \right| \\ &\leq \frac{C^{m} \kappa^{m+1}}{N} \int_{\Lambda \times \Lambda} dx dy N^{3\beta} V(N^\beta (x-y)) \| (\cN_+ +1) \xi \| \| \check{a}_x \check{a}_y \xi  \| \\ &\leq  C^{m} \kappa^{m+1/2} N^{-1/2} \| \cV_N^{1/2} \xi \| \| (\cN_+ +1) \xi \|  \end{split} \]
The contribution of the other terms on the r.h.s. of (\ref{eq:eta-scat2b}) can be bounded similarly. We conclude that, for $n=\ell_1 = 0$, 
\[ |\langle \xi , \text{F}_2 \xi \rangle| \leq \frac{C^m \kappa^{m+1} (m+1)}{N} \| (\cN_+ +1) \xi \|^2 + \frac{C^m \kappa^{m+1/2}}{\sqrt{N}} \| (\cN_+ +1) \xi \| \| \cV_N^{1/2} \xi \| \]
Since $\cN_+  \leq N$ on $\cF_+^{\leq N}$, the last estimate also implies that
\[ |\langle \xi , \text{F}_2 \xi \rangle| \leq C^m \kappa^{m+1} (m+1) \| (\cN_+ +1)^{1/2} \xi \|^2 + C^m \kappa^{m+1/2} \| (\cN_+ +1)^{1/2} \xi \| \| \cV_N^{1/2} \xi \| \]
Combining the last two bounds with (\ref{eq:F2-1}) and (\ref{eq:F2-2}) we obtain that, for 
every $n,m \in \bN$, 
\begin{equation}\label{eq:F2-1-f} |\langle \xi , \text{F}_2 \xi \rangle|  \leq C^{n+m} \kappa^{n+m+1} (m+1) \| (\cN_+ +1)^{1/2} \xi \|^2 + C^{n+m} \kappa^{n+m+1/2}  \| (\cN_+ +1)^{1/2} \xi \| \| \cV_N^{1/2} \xi \| \end{equation}
and, with Lemma \ref{lm:cVN},  
\begin{equation}\label{eq:F2-2-f} \begin{split} 
|\langle \xi , \text{F}_2 \xi \rangle| \leq \; &C^{n+m} \kappa^{n+m+1} (m+1) N^{\beta-1} \| (\cN_+ +1) \xi \|^2 + C^{n+m} \kappa^{n+m} \| \cV_N^{1/2} \xi \|^2 \\ \leq \; &C^{n+m} \kappa^{n+m} (m+1) N^{\beta-1} \|(\cN_+ +1)^{1/2} (\cK+1)^{1/2} \xi \|^2 \, . \end{split} \end{equation}

From (\ref{eq:F1-2}) and (\ref{eq:F2-1-f}) we conclude that, if $\kappa > 0$ is small enough, 
\[ |\langle \xi, \cE^{(K)}_1 \xi \rangle| \leq C \kappa \| (\cN_+ +1)^{1/2} \xi \| + C \kappa^{1/2} \|(\cN_+ +1)^{1/2} \xi \| \| \cV_N^{1/2} \xi \| \]
Hence, for every $\delta > 0$ we can find $C > 0$ such that 
\[ |\langle \xi, \cE^{(K)}_1 \xi \rangle| \leq \delta \| \cV_N^{1/2} \xi \|^2 + C \kappa \| (\cN_+ +1)^{1/2} \xi \|^2 \]
From (\ref{eq:F2-2-f}) and (\ref{eq:F1-2}), we can also estimate, if $\kappa > 0$ is small enough, 
\[  |\langle \xi, \cE^{(K)}_1 \xi \rangle| \leq C  N^{\beta-1} \| (\cN_+ +1)^{1/2} (\cK +1)^{1/2} \xi \|^2 \]
This proves that the error term $\cE^{(K)}_1$ satisfies the first bound in (\ref{eq:bdK}) and (\ref{eq:bdK2}). In fact, it also satisfies the second bound in (\ref{eq:bdK}), because the commutator of every term of the form (\ref{eq:EK1-tms}) with $\cN_+ $ has again the same form, up to multiplication with a constant, bounded by $C (m+1)$ (because the difference between the number of creation and the number of annihilation operators in (\ref{eq:F1}), (\ref{eq:F2}) is at most proportional to $m$). 

The error term $\cE^{(K)}_2$ can be controlled exactly as we did with $\cE^{(K)}_1$. Also the error term $\cE^{(K)}_3$ can be controlled similarly. The difference is that, now, the operator 
\[ \sum_{p \in \Lambda^*_+} p^2 \eta_p \big[ \text{ad}^{(n)}_{B(\eta)} (b_p) - \eta_p^n b^{\sharp_n}_{\alpha_n p} \big] \big[ \text{ad}^{(m)}_{B(\eta)} (b_{-p}) - \eta_p^m b^{\sharp_m}_{-\alpha_m p} \big] \]
can be written as the sum of $(2^m m! -1) (2^n n! -1)$ terms of the form 
\begin{equation}\label{eq:F3} 
\begin{split} \text{F}_3 = \; &\sum_{p \in \Lambda^*_+} p^2 \eta_p \Lambda_1 \dots \Lambda_{i_1} N^{-k_1} \Pi^{(1)}_{\sharp,\flat} (\eta^{j_1} , \dots , \eta^{j_{k_1}} ; \eta^{\ell_1}_p \ph_{\alpha_{\ell_1} p}) \\ &\hspace{3cm} \times \Lambda'_1 \dots \Lambda'_{i_2} N^{-k_2} \Pi^{(1)}_{\sharp',\flat'} (\eta^{m_1} , \dots , \eta^{m_{k_2}} ; \eta^{\ell_2}_p \ph_{-\alpha_{\ell_2} p})\end{split} \end{equation}
of $(2^m m!-1)$ terms of the form
\begin{equation}\label{eq:F4} \begin{split} \text{F}_4 = \; &\sum_{p \in \Lambda^*_+} p^2 \eta_p \left\{ \left( \frac{N-\cN_+ }{N} \right)^{\frac{n+(1-\alpha_n)/2}{2}} \left( \frac{N+1-\cN_+ }{N} \right)^{\frac{n-(1-\alpha_n)/2}{2}} -1 \right\} b^{\sharp_n}_{\alpha_n p}  \\ &\hspace{3cm} \times  \Lambda'_1 \dots \Lambda'_{i_2} N^{-k_2} \Pi^{(1)}_{\sharp',\flat'} (\eta^{m_1} , \dots , \eta^{m_{k_2}} ; \eta^{\ell_2}_p \ph_{-\alpha_{\ell_2} p}) \end{split} \end{equation}
of $(2^n n! -1)$ terms of the form 
\begin{equation}\label{eq:F5} \begin{split} \text{F}_5 = \; &\sum_{p \in \Lambda^*_+} p^2 \eta_p  \, \Lambda_1 \dots \Lambda_{i_1} N^{-k_1} \Pi^{(1)}_{\sharp,\flat} (\eta^{j_1} , \dots , \eta^{j_{k_1}} ; \eta^{\ell_1}_p \ph_{\alpha_{\ell_1} p}) \\ &\hspace{1cm} \times \left\{ \left( \frac{N-\cN_+ }{N} \right)^{\frac{m+(1-\alpha_m)/2}{2}} \left( \frac{N+1-\cN_+ }{N} \right)^{\frac{m-(1-\alpha_m)/2}{2}} -1 \right\} b^{\sharp_m}_{-\alpha_m p} 
\end{split} \end{equation}
and of one term of the form
\begin{equation}\label{eq:F6} \begin{split} \text{F}_6 = \; &\sum_{p \in \Lambda^*_+} p^2 \eta_p \left\{ \left( \frac{N-\cN_+ }{N} \right)^{\frac{n+(1-\alpha_n)/2}{2}} \left( \frac{N+1-\cN_+ }{N} \right)^{\frac{n-(1-\alpha_n)/2}{2}} -1 \right\} b^{\sharp_n}_{\alpha_n p}  \\ &\hspace{1cm} \times \left\{ \left( \frac{N-\cN_+ }{N} \right)^{\frac{m+(1-\alpha_m)/2}{2}} \left( \frac{N+1-\cN_+ }{N} \right)^{\frac{m-(1-\alpha_m)/2}{2}} -1 \right\} b^{\sharp_m}_{-\alpha_m p} \end{split} 
\end{equation}
where $i_1, i_2, k_1, k_2, \ell_1, \ell_2 \in \bN$, $j_1, \dots , j_{k_1}, m_1, \dots , m_{k_2} \in \bN \backslash \{ 0 \}$, $\alpha_r = (-1)^r$ and where each $\Lambda_r$- and $\Lambda'_r$-operator is either a factor $(N-\cN_+ )/N$, a factor $(N+1-\cN_+ )/N$ or a $\Pi^{(2)}$-operator of the form (\ref{eq:Pi2-K}). Furthermore, in (\ref{eq:F3}), we must have $k_1 > 0$ or at least one
$\Lambda$-operator of the form (\ref{eq:Pi2}) and $k_2 > 0$ or at least one $\Lambda'$-operator of the form (\ref{eq:Pi2}). Similarly, in (\ref{eq:F4}) we must have $k_2 > 0$ or at least one $\Lambda'$-operator of the form (\ref{eq:Pi2}) and in (\ref{eq:F5}) we must have $k_1 > 0$ or at least one $\Lambda$-operator of the form (\ref{eq:Pi2}). The terms (\ref{eq:F3}), (\ref{eq:F4}), (\ref{eq:F5}) and (\ref{eq:F6}) can therefore be estimated using Lemma \ref{lm:aux} 
as we did above with the terms $\text{F}_1$ defined in (\ref{eq:F1}) and the terms $\text{F}_2$ defined in (\ref{eq:F2}). We omit the details. 
\end{proof}

Next, we focus on the quadratic terms in (\ref{eq:wtL2N}). We define the operator $\wt{\cE}^{(V)}_N$ through
\begin{equation}\label{eq:cEV} e^{-B(\eta)} \cL^{(V)}_{N,\beta} e^{B(\eta)}  = \sum_{p \in \Lambda_+^*} \left[ \kappa \widehat{V} (p/N^\beta) \eta_p + \frac{\kappa  \widehat{V} (p/N^\beta)}{2} (b_p b_{-p} + b_p^* b_{-p}^*) \right] + \wt{\cE}^{(V)}_{N,\beta} \end{equation}
To prove part b) of Theorem \ref{thm:gene}, we will need to keep track of more contributions to $\cL^{(V)}_N$, so that the error has a vanishing expectation, in the limit of large $N$, on low-energy states. We define therefore the operator $\cE^{(V)}_{N,\beta}$ through
\begin{equation}\label{eq:cEV-def} \begin{split} e^{-B(\eta)} \cL^{(V)}_{N,\beta}  e^{B(\eta)} = \; &\sum_{p \in \Lambda^*_+} \left[ \kappa \widehat{V} (p/N^\beta) \sigma_p^2 + \kappa \widehat{V} (p/N^\beta) \sigma_p \gamma_p \right]  \\ &+ \sum_{p \in \Lambda^*_+} \kappa \widehat{V} (p/N^\beta) (\gamma_p + \sigma_p)^2 b_p^* b_p \\ &+ \frac{1}{2} \sum_{p \in \Lambda^*_+} \kappa \widehat{V} (p/N^\beta) (\gamma_p+\sigma_p)^2 (b_pb_{-p} + b_p^* b_{-p}^*) + \cE_{N,\beta}^{(V)} \end{split} \end{equation}
In the next proposition, we establish bounds for the error terms $\wt{\cE}_{N,\beta}^{(V)}$ and $\cE_{N,\beta}^{(V)}$.
\begin{prop}\label{prop:V}
Under the assumptions of Theorem \ref{thm:gene}, for every $\delta > 0$ there exists $C > 0$ such that, on $\cF_+^{\leq N}$,  
\[ \begin{split} 
\pm \wt{\cE}_{N,\beta}^{(V)} &\leq \delta \cV_N + C \kappa (\cN_+ +1) 
\\ \pm \left[ \wt{\cE}^{(V)}_{N,\beta} , i \cN_+  \right]
&\leq C (\cH_N^\beta +1) \end{split} \]
Furthermore,  \[  \pm \cE_{N,\beta}^{(V)} \leq C N^{\beta-1} \| (\cN_+ +1)^{1/2} (\cK+1)^{1/2} \xi \|^2 \]
\end{prop}

\begin{proof}
From the definition of $\cL^{(V)}_{N,\beta}$ in (\ref{eq:wtL2N}), we find
\begin{equation}\label{eq:G2-deco} \begin{split} 
\cG_{N,\beta}^{(2)} =\; &\kappa \sum_{p \in \Lambda^*_+}  \widehat{V} (p/N^\beta) e^{-B(\eta)} b_p^* b_p e^{B(\eta)} - \frac{\kappa}{N} \sum_{p \in \Lambda^*_+} \widehat{V} (p/N^\beta) e^{B(\eta)} a_p^* a_p e^{-B(\eta)} \\ &+ \frac{\kappa}{2} \sum_{p \in \Lambda^*_+} \widehat{V} (p/N^\beta) e^{-B(\eta)} \big[ b_p b_{-p} + b_p^* b_{-p}^* \big] e^{B(\eta)} \\ =: \; &\cG^{(2,1)}_{N,\beta} + \cG_{N,\beta}^{(2,2)} + \cG_{N,\beta}^{(2,3)} \end{split} \end{equation}
{F}rom Lemma \ref{lm:conv-series}, the term $\cG_{N,\beta}^{(2,1)}$ can be written (using again the notation $\gamma_p = \cosh \eta_p$, $\sigma_p = \sinh \eta_p$) as
\begin{equation}\label{eq:G21-deco} \begin{split} 
\cG_{N,\beta}^{(2,1)}  = \; & \sum_{m,n \geq 0} \frac{(-1)^{m+n}}{m! n!} \kappa \sum_{ p\in \Lambda^*_+} \widehat{V} (p/N^\beta) \, \left[ \text{ad}^{(m)}_{B(\eta)} (b^*_p) - \eta_p^m b_{\alpha_m p}^{\bar{\sharp}_m} + \eta_p^m b_{\alpha_m p}^{\bar{\sharp}_m} \right] \\ &\hspace{5cm} \times  \left[ \text{ad}^{(n)}_{B(\eta)} (b_p) - \eta_p^n b_{\alpha_n p}^{\sharp_n} + \eta_p^n b_{\alpha_n p}^{\sharp_n} \right] \\ =\; &\sum_{m,n \geq 0} \frac{(-1)^{m+n}}{m! n!} \kappa \sum_{ p\in \Lambda^*_+} \widehat{V} (p/N^\beta) \eta_p^{m+n} b^{\bar{\sharp}_m}_{\alpha_m p} b^{\sharp_n}_{\alpha_n p} + \cE^{(V)}_1  \\
=\; &\kappa \sum_{ p\in \Lambda^*_+} \widehat{V} (p/N^\beta) \big[ \gamma_p b_p^* + \sigma_p b_{-p} \big] \big[ \gamma_p b_p + \sigma_p b_{-p}^*]  + \cE^{(V)}_1 \end{split} \end{equation}
with $\alpha_n = 1$ and $\sharp_n = \cdot$ if $n$ is even while $\alpha_n = -1$ and $\sharp_n =*$ if $n$ is odd (and $\bar{\sharp}_n = *$ if $\sharp_n = \cdot$ and $\bar{\sharp}_n = \cdot$ if $\sharp_n = *$) and with the error term
\begin{equation}\label{eq:EV1} \begin{split} \cE^{(V)}_1 =  \; &\sum_{m,n \geq 0} \frac{(-1)^{m+n}}{m! n!} \kappa \sum_{ p\in \Lambda^*_+} \widehat{V} (p/N^\beta) \eta_p^m b_{\alpha_m p}^{\bar{\sharp}_m} \big[ \text{ad}^{(n)}_{B(\eta)} (b_p) - \eta_p^n b^{\sharp_n}_{\alpha_n p} \big] \\ &+ \sum_{m,n \geq 0} \frac{(-1)^{m+n}}{m! n!} \kappa \sum_{ p\in \Lambda^*_+} \widehat{V} (p/N^\beta)  \big[ \text{ad}^{(m)}_{B(\eta)} (b^*_p) - \eta_p^m b^{\bar{\sharp}_m}_{\alpha_m p} \big] \eta_p^n b_{\alpha_n p}^{\sharp_n} \\ &+ \sum_{m,n \geq 0} \frac{(-1)^{m+n}}{m! n!} \kappa \sum_{ p\in \Lambda^*_+} \widehat{V} (p/N^\beta)  \big[ \text{ad}^{(m)}_{B(\eta)} (b^*_p) - \eta_p^m b^{\bar{\sharp}_m}_{\alpha_m p} \big] \big[ \text{ad}^{(n)}_{B(\eta)} (b_p) - \eta_p^n b^{\sharp_n}_{\alpha_n p} \big] 
\end{split} \end{equation}
According to Lemma \ref{lm:indu}, the operator 
\[ \kappa \sum_{p \in \Lambda^*_+} \widehat{V} (p/N^\beta) \eta_p^m b^{\bar{\sharp}_m}_{\alpha_m p} \big[ \text{ad}^{(n)}_{B(\eta)} (b_p) - \eta_p^n b^{\sharp_n}_{\alpha_n p} \big]  \]
can be written as the sum of one term of the form
\begin{equation}\label{eq:G1-EV1} \begin{split} \text{G}_1 = &\; \kappa \sum_{p \in \Lambda^*_+} \widehat{V} (p/N^\beta) \eta_p^{m+n} b^{\bar{\sharp}_m}_{\alpha_m p} \\ &\hspace{2cm} \times \left\{ \left( \frac{N-\cN_+ }{N} \right)^{\frac{n+(1-\alpha_n)/2}{2}} \left( \frac{N+1-\cN_+ }{N} \right)^{\frac{n-(1-\alpha_n)/2}{2}} - 1 \right\} b^{\sharp_n}_{\alpha_n p} \end{split} \end{equation}
and of $2^n n! -1$ terms of the form
\begin{equation}\label{eq:G2-EV1} \text{G}_2 = \kappa \sum_{p \in \Lambda^*_+} \widehat{V} (p/N^\beta) \eta_p^m b^{\bar{\sharp}_m}_{\alpha_m p} \Lambda_1 \dots \Lambda_{i_1} N^{-k_1} \Pi^{(1)}_{\sharp, \flat} (\eta^{j_1}, \dots , \eta^{j_{k_1}} ; \eta_p^{\ell_1} \ph_{\alpha_{\ell_1} p}) 
\end{equation}
with $i_1, k_1, \ell_1 \in \bN$, $j_1, \dots, j_{k_1} \in \bN \backslash \{ 0 \}$, $\alpha_r = (-1)^r$ and where each $\Lambda_r$ is either a factor $(N-\cN_+ )/N$, a factor $(N+1-\cN_+ )/N$ or a $\Pi^{(2)}$-operator of the form
\begin{equation}\label{eq:Pi2-EV1} N^{-h} \Pi^{(2)}_{\underline{\sharp}, \underline{\flat}} (\eta^{z_1}, \dots , \eta^{z_p}) \end{equation}
for $h, z_1, \dots , z_p \in \bN \backslash \{ 0 \}$.
Furthermore, each operator of the form (\ref{eq:G2-EV1}) must have either $k_1 > 0$ or at least one $\Lambda$-operator having the form (\ref{eq:Pi2-EV1}). 

Noticing that $\text{G}_1 = 0$ if $n=0$, the expectation of (\ref{eq:G1-EV1}) can be bounded by 
\[ \begin{split} 
|\langle \xi , \text{G}_1 \xi \rangle | \leq \; &\frac{C^{n+m} \kappa^{n+m+1}}{N}  \| (\cN_+ +1) \xi \| \\ &\hspace{2cm} \times \sum_{p \in \Lambda^*_+} \frac{|\widehat{V} (p/N^\beta)|}{p^2} \left[ \| a_p (\cN_+ +1)^{1/2} \xi \| + \frac{1}{p^{2}} \| (\cN_+ +1) \xi \| \right]  \\ \leq \; & 
\frac{C^{n+m} \kappa^{n+m+1}}{N} \, \| (\cN_+ +1) \xi \|^2  \end{split} \]

As for the term $\text{G}_2$ defined in (\ref{eq:G2-EV1}), its expectation can be bounded with Lemma~\ref{lm:aux} part iii) by
\[ \begin{split}
|\langle \xi , \text{G}_2 \xi \rangle | &\leq \frac{C^{n+m} \kappa^{n+m+1}}{N} \sum_{p \in \Lambda^*_+} |\widehat{V} (p/N^\beta)| \left\{ \| a_p (\cN_+ +1)^{1/2} \xi \| + \frac{1}{p^2} \| (\cN_+ +1) \xi \| \right\}^2 \\ &\leq \frac{C^{n+m} \kappa^{n+m+1}}{N}  \| (\cN_+ +1) \xi \|^2 
\end{split} \]
for all $\xi \in \cF_+^{\leq N}$. The expectation of the operators appearing on the second and third line in (\ref{eq:EV1}) can be controlled similarly, using again Lemma \ref{lm:aux}. Therefore, if $\kappa > 0$ is small enough, we can sum over $m,n \in \bN$, and from (\ref{eq:EV1}) we conclude that 
\begin{equation}\label{eq:EV1-bd} |\langle \xi , \cE_1^{(V)} \xi \rangle | \leq \frac{C \kappa}{N} \, \| (\cN_+ +1) \xi \|^2 \, \leq C \kappa \| (\cN_+ +1)^{1/2} \xi \|^2 . 
\end{equation}
Since commutators of $\cN_+ $ with operators of the form (\ref{eq:G1-EV1}), (\ref{eq:G2-EV1}) have again the same form (up to a multiplicative constant bounded by $C(n+1)$), we also find
\begin{equation}\label{eq:EV1-comm} |\langle \xi , \big[ \cN_+  , \cE_1^{(V)} \big] \xi \rangle | \leq C \kappa \| (\cN_+ +1)^{1/2} \xi \|^2 . 
\end{equation}

Let us now consider the second contribution to $\cG_{N,\beta}^{(2)}$ on the r.h.s. of (\ref{eq:G2-deco}). 
We observe that
\begin{equation}\label{eq:G22-def} \begin{split} -\cG_{N,\beta}^{(2,2)} = \; & \frac{\kappa}{N} \sum_{p \in \Lambda^*_+} \widehat{V} (p/N^\beta) e^{-B(\eta)} a_p^* a_p e^{B(\eta)}  \\ = \; &\frac{\kappa}{N} \sum_{p \in \Lambda^*_+} \widehat{V} (p/N^\beta) \left[ a_p^* a_p  + \int_0^1 ds \, e^{-sB(\eta)} [a_p^* a_p , B(\eta)] e^{sB(\eta)} \right] \\
= \; &\frac{\kappa}{N} \sum_{p \in \Lambda^*_+} \widehat{V} (p/N^\beta) a_p^* a_p  + \int_0^1 ds \, \frac{\kappa}{N} \sum_{p \in \Lambda^*_+} \widehat{V} (p/N^\beta) e^{-sB(\eta)} (b_p b_{-p} + b^*_p b^*_{-p}) e^{sB(\eta)} \\  
= \; &\frac{\kappa}{N} \sum_{p \in \Lambda^*_+} \widehat{V} (p/N^\beta) a_p^* a_p  \\ &+ \sum_{n,m \geq 0} \frac{(-1)^{m+n}}{m!n!(m+n+1)}  \frac{\kappa}{N} \sum_{p \in \Lambda^*_+} \widehat{V} (p/N^\beta) \left[ \text{ad}^{(n)}_{B(\eta)} (b_p) \text{ad}^{(m)}_{B(\eta)} (b_{-p}) + \text{h.c.} \right] 
\end{split} \end{equation}
The first term on the r.h.s. of (\ref{eq:G22-def}) is clearly bounded by $C\kappa \cN_+ /N$. Let us focus now 
on the sum over $m,n$. By Lemma \ref{lm:indu}, the operator
\[ \frac{\kappa}{N} \sum_{p \in \Lambda^*_+} \widehat{V} (p/N^\beta) \text{ad}^{(n)}_{B(\eta)} (b_p) \text{ad}^{(m)}_{B(\eta)} (b_{-p}) \]
can be written as the sum of $2^{n+m} n!m!$ terms of the form
\begin{equation}\label{eq:Ldef}\begin{split} 
\text{L} = \; & \frac{\kappa}{N} \sum_{p \in \Lambda_+^*} \widehat{V} (p/N^\beta) \Lambda_1 \dots \Lambda_{i_1} N^{-k_1} \Pi^{(1)}_{\sharp, \flat} (\eta^{j_1}, \dots \eta^{j_{k_1}} ; \eta_p^{\ell_1} \ph_{\alpha_{\ell_1} p}) \\ &\hspace{3cm} \times \Lambda'_1 \dots \Lambda'_{i_2} N^{-k_2} \Pi^{(1)}_{\sharp',\flat'} (\eta^{m_1}, \dots , \eta^{m_{k_2}} ; \eta^{\ell_2}_p \ph_{-\alpha_{\ell_2} p}) \end{split} \end{equation}
with $i_1, i_2, k_1, k_2, \ell_1, \ell_2 \in \bN$, $j_1, \dots, j_{k_1}, m_1, \dots , m_{k_2} \in \bN \backslash \{ 0 \}$, $\alpha_r =(- 1)^r$ and where each $\Lambda_r$ and $\Lambda'_r$-operator is either a factor $(N-\cN_+ )/N$, a factor $(N+1-\cN_+ )/N$ or a $\Pi^{(2)}$-operator of the form (\ref{eq:Pi2-EV1}). 

If $\ell_1 + \ell_2 \geq 1$, Lemma \ref{lm:aux}, part ii), implies that  
\begin{equation}\label{eq:Ln00} 
\begin{split} 
|\langle \xi , \text{L} \xi \rangle | &\leq \frac{C^{n+m} \kappa^{n+m+1}}{N} \sum_{p \in \Lambda^*_+} \frac{|\widehat{V} (p/N^\beta)|}{p^2} \| (\cN_+ +1)^{1/2}  \xi \|^2 \\ &\leq C^{n+m} \kappa^{n+m+1} N^{\beta-1} \| (\cN_+ +1)^{1/2} \xi \|^2 \end{split} \end{equation} 
If instead $\ell_1 = \ell_2 = 0$, we use Lemma \ref{lm:aux}, part ii), to write
\begin{equation}\label{eq:L00} \langle \xi , \text{L} \xi \rangle = \frac{\kappa}{N} \sum_{p \in \Lambda^*_+} \widehat{V} (p/N^\beta) \langle \xi , \text{D}_1 (p) \rangle + \frac{\kappa}{N} \sum_{p \in \Lambda^*_+} \widehat{V} (p/N^\beta) \langle \xi , \text{D}_2 \, a_p a_{-p} \xi \rangle  \end{equation}
with 
\begin{equation}\label{eq:E1-L} \| (\cN_+ +1)^{-1/2} \text{D}_1 (p) \| \leq \frac{C^{m+n} \kappa^{m+n} m}{N p^2} \| a_p \xi \| \qquad \text{and } \quad \| \text{D}_2 \| \leq C^{m+n} \kappa^{m+n} \, . \end{equation}
Switching to position space to estimate the second term on the r.h.s. of (\ref{eq:L00}) we find, for $\ell_1 = \ell_2 = 0$,  
\begin{equation}\label{eq:L00-f} \begin{split} |\langle  \xi , \text{L} \xi \rangle | \leq \; &\frac{C^{m+n} \kappa^{m+n+1} m}{N} \| (\cN_+ +1)^{1/2} \xi \|^2 + \left| \frac{\kappa}{N} \int_{\Lambda \times \Lambda} N^{3\beta} V(N^\beta (x-y)) \langle \xi , \text{D}_2 \check{a}_x \check{a}_y \xi \rangle \right| \\ \leq \; &\frac{C^{m+n} \kappa^{m+n+1} m}{N} \| (\cN_+ +1)^{1/2} \xi \|^2 \\ &+ \frac{C^{m+n} \kappa^{m+n+1}}{N} \int_{\Lambda \times \Lambda} N^{3\beta} V(N^\beta (x-y)) \| \xi \| \| \check{a}_x \check{a}_y \xi \| 
\\ \leq \; &\frac{C^{m+n} \kappa^{m+n+1} m}{N} \| (\cN_+ +1)^{1/2} \xi \|^2 + \frac{C^{m+n} \kappa^{m+n+1/2}}{\sqrt{N}}  \| \cV_N^{1/2} \xi \| \| \xi \| \end{split} \end{equation}
Combining (\ref{eq:Ln00}) with (\ref{eq:L00-f}) we conclude that 
\[ |\langle \xi , \text{L} \xi \rangle | \leq C^{m+n} \kappa^{m+n+1} (m+1) N^{\beta-1} \| (\cN_+ +1)^{1/2} \xi \|^2 + \frac{C^{m+n} \kappa^{m+n+1/2}}{\sqrt{N}} \| \xi \| \| \cV_N^{1/2} \xi \| \]
Hence, for $\kappa > 0$ small enough (so that we can sum over $m,n \in \bN$), (\ref{eq:G22-def}) implies 
\begin{equation}\label{eq:fin-G220} |\langle \xi , \cG_{N,\beta}^{(2,2)} \xi \rangle | \leq C \kappa N^{\beta-1} \| (\cN_+ +1)^{1/2} \xi \|^2 + \frac{C\kappa^{1/2}}{\sqrt{N}} \| \xi \| \| \cV_N^{1/2} \xi \| \end{equation}
This shows, on the one hand, that for every $\delta > 0$ there exists $C >0$ such that 
\begin{equation}\label{eq:fin-G22} |\langle \xi , \cG_{N,\beta}^{(2,2)} \xi \rangle | \leq \delta \| \cV_N^{1/2} \xi \|^2 + C \kappa \| (\cN_+ +1)^{1/2} \xi \|^2 \end{equation}
and, since as usual the commutator of $\cN_+ $ with operators of the form (\ref{eq:Ldef}) has again the same form, 
\begin{equation}\label{eq:fin-G22-comm} |\langle \xi , \big[ \cG_{N,\beta}^{(2,2)}, \cN_+  \big] \xi \rangle | \leq \delta \| \cV_N^{1/2} \xi \|^2 + C \kappa \| (\cN_+ +1)^{1/2} \xi \|^2 \end{equation}
On the other hand, taking into account Lemma \ref{lm:cVN}, (\ref{eq:fin-G220}) also proves that 
\[ |\langle \xi , \cG_{N,\beta}^{(2,2)} \xi \rangle | \leq C N^{\beta-1} \| (\cN_+ +1)^{1/2} (\cK +1)^{1/2} \xi \|^2 \]
 
Finally, let us consider the third contribution to $\cG^{(2)}_{N,\beta}$ on the r.h.s. of (\ref{eq:G2-deco}). We have, with Lemma \ref{lm:conv-series}, 
\begin{equation}\label{eq:G23-deco} \begin{split} \cG_{N,\beta}^{(2,3)} = \; & \frac{\kappa}{2} \sum_{p \in \Lambda^*_+} \widehat{V} (p/N^\beta) e^{-B(\eta)} \big[ b_p b_{-p} + b_p^* b_{-p}^* \big] e^{B(\eta)} \\
= \; &\sum_{m,n \geq 0} \frac{(-1)^{m+n}}{m!n!} \kappa \sum_{p \in \Lambda^*_+} \widehat{V} (p/N^\beta) \, \left[ \text{ad}^{(m)}_{B(\eta)} (b_p) \text{ad}^{(n)}_{B(\eta)} (b_{-p}) + \text{h.c.} \right] \\
= \; &\frac{\kappa}{2} \sum_{p \in \Lambda^*_+} \widehat{V} (p/N^\beta) \\ &\hspace{.1cm} \times \left\{ \left[ \gamma_p b_p + \sigma_p b_{-p}^* \right] \left[ \gamma_p b_{-p} + \sigma_p b_p^* \right] + \left[ \gamma_p b^*_p + \sigma_p b_{-p} \right] \left[ \gamma_p b^*_{-p} + \sigma_p b_p \right] \right\}\\ & + \cE^{(V)}_3 \end{split} \end{equation}
with the error term
\begin{equation}\label{eq:EV2} \begin{split} \cE_3^{(V)} = \; &\sum_{m,n \geq 0} \frac{(-1)^{m+n}}{m!n!} \frac{\kappa}{2} \sum_{p \in \Lambda^*_+} \widehat{V} (p/N^\beta) \, \eta_p^m b_{\alpha_m p}^{\sharp_m} \, \big[ \text{ad}^{(n)}_{B(\eta)} (b_{-p}) - \eta^n_p b_{-\alpha_n p}^{\sharp_n} \big] \\ &+ \sum_{m,n \geq 0} \frac{(-1)^{m+n}}{m!n!} \frac{\kappa}{2} \sum_{p \in \Lambda^*_+} \widehat{V} (p/N^\beta)  \, \big[ \text{ad}^{(m)}_{B(\eta)} (b_{p}) - \eta^m_p b_{\alpha_m p}^{\sharp_m} \big] \, \eta_p^n b_{-\alpha_n p}^{\sharp_n}
\\ &+ \sum_{m,n \geq 0} \frac{(-1)^{m+n}}{m!n!} \frac{\kappa}{2} \sum_{p \in \Lambda^*_+} \widehat{V} (p/N^\beta)  \\ &\hspace{3cm} \times  \big[ \text{ad}^{(m)}_{B(\eta)} (b_{p}) - \eta^m_p b_{\alpha_m p}^{\sharp_m} \big] \,
\big[ \text{ad}^{(n)}_{B(\eta)} (b_{-p}) - \eta^n_p b_{-\alpha_n p}^{\sharp_n} \big] \\&+ \text{h.c.} 
\end{split} \end{equation}
We consider the first sum on the r.h.s. of (\ref{eq:EV2}). According to Lemma \ref{lm:conv-series}, the operator

\[ \frac{\kappa}{2} \sum_{p \in \Lambda^*_+} \widehat{V} (p/N^\beta) \eta_p^{m} b^{\sharp_m}_{\alpha_m p} \big[ \text{ad}^{(n)}_{B(\eta)} (b_{-p}) - \eta_p^n b^{\sharp_n}_{-\alpha_n p} \big] \]
can be written as the sum of the one term
\begin{equation}\label{eq:M1} \begin{split} \text{M}_1 = \; &\frac{\kappa}{2} \sum_{p \in \Lambda^*_+} \widehat{V} (p/N^\beta) \eta_p^{m+n} b_{\alpha_m p}^{\sharp_m} \\ &\hspace{1cm} \left\{ \left( \frac{N-\cN_+ }{N} \right)^{\frac{n+(1-\alpha_n)/2}{2}} \left( \frac{N+1-\cN_+ }{N} \right)^{\frac{n-(1-\alpha_n)/2}{2}} - 1 \right\} b^{\sharp_n}_{-\alpha_n p} \end{split}\end{equation}
and of $2^n n! -1$ terms of the form 
\begin{equation}\label{eq:M2}\text{M}_2 = \frac{\kappa}{2} \sum_{p \in \Lambda^*_+} \widehat{V} (p/N^\beta) \eta_p^{m} b_{\alpha_m p}^{\sharp_m} \Lambda_1 \dots \Lambda_{i_1} N^{-k_1} \Pi^{(1)}_{\sharp, \flat} (\eta^{j_1}, \dots , \eta^{j_{k_1}} ; \eta^{\ell_1}_p \ph_{-\alpha_{\ell_1} p}) \end{equation}
where $i_1, k_1, \ell_1 \in \bN$, $j_1, \dots ,j_{k_1} \in \bN \backslash \{ 0 \}$ and where each $\Lambda_r$-operator is either a factor $(N-\cN_+ )/N$, a factor $(N+1-\cN_+ )/N$ or a $\Pi^{(2)}$-operator of the form (\ref{eq:Pi2-EV1}). In every term of the form (\ref{eq:M2}) we have $k_1 > 0$ or at least one of the $\Lambda$-operators must have the form (\ref{eq:Pi2-EV1}). 

To estimate the expectation of (\ref{eq:M1}) we proceed very similarly as we did in the proof of Prop. \ref{prop:K}, to show (\ref{eq:F1-1}), (\ref{eq:F1-2}). We obtain 
\begin{equation}\label{eq:M1-1} |\langle \xi , \text{M}_1 \xi \rangle | \leq C^{m+n} \kappa^{m+n+1} (n+1) \| (\cN_+ +1)^{1/2} \xi \|^2 \end{equation}
and also 
\begin{equation}\label{eq:M1-2} \begin{split} |\langle \xi, \text{M}_1 \xi \rangle | &\leq  C^{m+n} \kappa^{m+n+1} (n+1) N^{\beta-1} \| (\cN_+ +1) \xi \|^2
\end{split} \end{equation}

Next, we bound the expectation of the term $\text{M}_2$, defined in (\ref{eq:M2}). If $m+\ell_1 \geq 1$, we can use Lemma \ref{lm:aux}, part iii), to estimate 
\[ \begin{split} 
|\langle \xi , \text{M}_2 \xi \rangle | \leq \; & \frac{\kappa}{2} \sum_{p \in \Lambda^*_+} |\widehat{V} (p/N^\beta)| |\eta_p|^m \| (\cN_+ +1)^{1/2} \xi \| \\ &\hspace{.5cm} \times 
\| (\cN_+ +1)^{-1/2} b_{\alpha_m p}^{\sharp_m} \Lambda_1 \dots \Lambda_{i_1} N^{-k_1} \Pi^{(1)}_{\sharp,\flat} (\eta^{j_1} , \dots ,\eta^{j_{k_1}} ; \eta_p^{\ell_1} \ph_{-\alpha_{\ell_1} p}) \xi \| \\
\leq \; & C^{n+m} \kappa^{n+m+1} \| (\cN_+ +1)^{1/2} \xi \|\sum_{p \in \Lambda^*_+} |\widehat{V} (p/N^\beta)|  \\ &\hspace{.5cm} \times \left\{ \frac{(1+m/N)}{|p|^{4}} \| (\cN_+ +1)^{1/2} \xi \| + \frac{1}{|p|^2} \| a_p \xi \| + \frac{1}{N|p|^2} \| (\cN_+ +1)^{1/2} \xi \| \right\} \\
\leq \; &C^{n+m} \kappa^{n+m+1} \| (\cN_+ +1)^{1/2} \xi \|^2 \end{split} \]
Alternatively, again for $m+\ell_1 \geq 1$, we can also use Lemma \ref{lm:aux}, part iii), to show the bound
\[ \begin{split} |\langle \xi, \text{M}_2 \xi \rangle| \leq \; & \frac{\kappa}{2} \sum_{p \in \Lambda^*_+} |\widehat{V} (p/N^\beta)| |\eta_p|^{m} \| (\cN_+ +1) \xi \| \\ &\hspace{1cm} \times \| (\cN_+ +1)^{-1} b_{\alpha_m p}^{\sharp_m} \Lambda_1 \dots \Lambda_{i_1} N^{-k_1} \Pi^{(1)}_{\sharp,\flat} (\eta^{j_1} , \dots , \eta^{j_{k_1}} ; \eta_p^{\ell_1} \ph_{-\alpha_{\ell_1} p}) \xi \| \\
\leq \; &\frac{C^{m+n} \kappa^{m+n+1} (m+1)}{N} \| (\cN_+ +1) \xi \| \\ &\hspace{1cm} \times \sum_{p \in \Lambda^*_+} |\widehat{V} (p/N^\beta)| \left[ \frac{1}{|p|^4} \| (\cN_+ +1) \xi \| + \frac{1}{p^2} \| a_p (\cN_+ +1)^{1/2} \xi \| + \frac{1}{p^2} \| \xi \| \right] \\
\leq \; &C^{m+n} \kappa^{m+n+1} N^{\beta-1} \| (\cN_+ +1) \xi \|^2 \end{split} \] 
If now $m=\ell_1 = 0$, we write
\begin{equation}\label{eq:M2-00} \langle \xi, \text{M}_2 \xi \rangle = \frac{\kappa}{2} \sum_{p \in \Lambda^*_+} \widehat{V} (p/N^\beta) \langle \xi, \text{D}_1 (p) \rangle + \frac{\kappa}{2} \sum_{p \in \Lambda^*_+} \widehat{V} (p/N^\beta) \langle  \xi , \text{D}_2 \, a_p a_{-p} \xi \rangle \end{equation}
with $\text{D}_1 (p)$ and the operator $\text{D}_2$ satisfying 
\begin{equation}\label{eq:E1-M} \begin{split} \| (\cN_+ +1)^{-1/2} \text{D}_1 (p) \| &\leq \frac{C^{n} \kappa^{n+1} n}{N p^2} \| a_p \xi \| \qquad \text{and } \\ \| \text{D}_2^* \xi \| &\leq \frac{C^{n} \kappa^{n+1}}{N} \| (\cN_+ +1) \xi \| \, . \end{split} \end{equation}
Switching to position space to estimate the second term on the r.h.s. of (\ref{eq:M2-00}), we conclude that
\begin{equation} \label{eq:M2-1-f} \begin{split} 
|\langle \xi, \text{M}_2 \xi \rangle| \leq \; & \frac{C^n \kappa^{n+1} n}{N} \| (\cN_+ +1)^{1/2} \xi \|^2 \\ &+ \frac{C^n \kappa^{n+1}}{N} \int_{\Lambda \times \Lambda} dx dy N^{3\beta} V(N^\beta (x-y)) \| \check{a}_x \check{a}_y \xi \| \| (\cN_+ +1) \xi \| \\ 
\leq \; &\frac{C^n \kappa^{n+1}}{N} \| (\cN_+ +1)^{1/2} \xi \|^2 + \frac{C^n \kappa^{n+1/2}}{\sqrt{N}} \| \cV_N^{1/2} \xi \| \| (\cN_+ +1) \xi \|  \end{split} \end{equation}
With (\ref{eq:M1-1}) and (\ref{eq:M2-1-f}), we can control the first sum on the r.h.s. of (\ref{eq:EV2}). The second and third sum can be controlled similarly. We conclude that, if $\kappa > 0$ is small enough (so that we can sum over $n,m \in \bN$), 
\[ |\langle \xi, \cE_3^{(V)} \xi \rangle | \leq C \kappa \| (\cN_+ +1)^{1/2} \xi \|^2 + C \kappa^{1/2} \| (\cN_+ +1)^{1/2} \xi \| \| \cV_N^{1/2} \xi \| \]
Hence, for every $\delta > 0$ we can find $C > 0$ such that 
\begin{equation} \label{eq:fin-EV3} |\langle \xi, \cE_3^{(V)} \xi \rangle | \leq \delta \| \cV_N^{1/2} \xi \|^2 + C \kappa \| (\cN_+ +1)^{1/2} \xi \|^2 \end{equation}
and (since the commutator with $\cN_+ $ of every term of the form (\ref{eq:M1}), (\ref{eq:M2}) is again an operator with the same form, up to a constant bounded by $C(n+1)$)
\begin{equation}\label{eq:EV3-comm} |\langle \xi, \big[ \cN_+ , \cE_3^{(V)} \big] \xi \rangle | \leq \delta \| \cV_N^{1/2} \xi \|^2 + C \kappa \| (\cN_+ +1)^{1/2} \xi \|^2 \end{equation}
Combining (\ref{eq:M1-2}) with (\ref{eq:M2-1-f}), we arrive moreover with Lemma \ref{lm:cVN} at the bound
\begin{equation}\label{eq:EV3-bd} |\langle  \xi , \cE_3^{(V)} \xi \rangle | \leq C N^{\beta-1} \| (\cN_+ +1)^{1/2} (\cK+1)^{1/2} \xi \|^2 \, . \end{equation} 

From (\ref{eq:G2-deco}), (\ref{eq:G21-deco}), (\ref{eq:G23-deco}) and from the definition (\ref{eq:cEV-def}), we obtain  
\[ \cE_{N,\beta}^{(V)} = \cE_1^{(V)} + \cG_{N,\beta}^{(2,2)} + \cE^{(V)}_3 \]
Hence, from the bounds (\ref{eq:EV1-bd}), (\ref{eq:fin-G22}) and (\ref{eq:EV3-bd}), we conclude that
\[ |\langle \xi , \cE_{N,\beta}^{(V)} \xi \rangle | \leq C N^{\beta-1} \| (\cK+1)^{1/2} (\cN_+ +1)^{1/2}  \xi \|^2   \]
Furthermore, with the definition (\ref{eq:cEV}), we find that 
\[ \wt{\cE}_{N,\beta}^{(V)} =  \cE_0^{(V)} + \cE_1^{(V)} + \cG_{N,\beta}^{(2,2)} + \cE^{(V)}_3 \]
with the additional term 
\[ \begin{split} \cE^{(V)}_0 = \; & \sum_{p \in \Lambda^*_+} \kappa \widehat{V} (p/N^\beta) [\sigma_p^2 + \sigma_p \gamma_p - \eta_p] + \sum_{p \in \Lambda^*_+} \kappa \widehat{V} (p/N^\beta) (\gamma_p + \sigma_p)^2 b_p^* b_p \\ &+ \frac{\kappa}{2} \sum_{p \in \Lambda^*_+} \widehat{V} (p/N^\beta) [(\gamma_p + \sigma_p)^2 - 1] (b_p b_{-p} + b_p^* b_{-p}^* ) \end{split}  \] 
Since $|\sigma_p^2 + \sigma_p \gamma_p - \eta_p| \leq C \kappa^2 |p|^{-4}$, $|\gamma_p + \sigma_p|^2 \leq C$ and $|(\gamma_p + \sigma_p)^2 -1| \leq C \kappa |p|^{-2}$, we easily find that 
\[ \begin{split} 
|\langle \xi , \cE_0^{(V)} \xi \rangle | &\leq C \kappa \| (\cN_+ +1)^{1/2} \xi \|^2 \\ 
|\langle \xi , \big[ \cN_+ , \cE_0^{(V)} \big] \xi \rangle | &\leq C \kappa \| (\cN_+ +1)^{1/2} \xi \|^2 
\end{split} \]
for all $\xi \in \cF_+^{\leq N}$. Together with the estimates (\ref{eq:EV1-bd}), (\ref{eq:EV1-comm}),  (\ref{eq:fin-G22}), (\ref{eq:fin-G22-comm}), (\ref{eq:fin-EV3}), (\ref{eq:EV3-comm}), we conclude that for every $\delta > 0$ there exists $C > 0$ such that 
\[ \begin{split} 
|\langle \xi , \wt{\cE}_{N,\beta}^{(V)} \xi \rangle | &\leq \delta \| \cV_N^{1/2} \xi \|^2 + C \kappa \| (\cN_+ +1)^{1/2} \xi \|^2 \\ 
|\langle \xi , \big[ \cN_+ , \wt{\cE}_{N,\beta}^{(V)} \big] \xi \rangle | &\leq \delta \| \cV_N^{1/2} \xi \|^2 + C \kappa \| (\cN_+ +1)^{1/2} \xi \|^2 
\end{split} \]
\end{proof}

\subsection{Analysis of $\cG_N^{(3)}$}

Recall from (\ref{eq:cLNj}) that
\begin{equation}\label{eq:cGN3} \begin{split} 
\cG_{N,\beta}^{(3)} &= e^{-B(\eta)} \cL^{(3)}_{N,\beta} e^{B(\eta)} \\ &= \frac{1}{\sqrt{N}} \sum_{p,q \in \Lambda^*_+ : p+q \not = 0} \widehat{V} (p/N^\beta) e^{-B(\eta)} \big[ b^*_{p+q} a_{-p}^* a_q + a_q^* a_{-p} b_{p+q} \big] e^{B(\eta)} \end{split} 
\end{equation}
In the next proposition, we show how to control the operator $\cG_{N,\beta}^{(3)}$. 
\begin{prop}\label{prop:G3}
Under the assumptions of Theorem \ref{thm:gene}, for every $\delta > 0$ there exists $C > 0$ such that, on $\cF_+^{\leq N}$,  
\[ \begin{split} 
\pm \cG_{N,\beta}^{(3)} &\leq \delta \cV_N + C \kappa (\cN_+ +1) \\
\pm \left[ \cG_{N,\beta}^{(3)} , i \cN_+  \right] &\leq C (\cH_N^\beta + 1) 
\end{split} \]
Furthermore, we have 
\[ \pm \cG_{N,\beta}^{(3)} \leq C N^{(\beta-1)/2} (\cK + 1) (\cN_+ +1) \]
\end{prop}

\begin{proof}
With Lemma \ref{lm:conv-series}, we write 
\[ \begin{split} e^{-B(\eta)} &a_{-p}^* a_q e^{B(\eta)} \\ = \; &a_{-p}^* a_q + \int_0^1 ds \, e^{-sB(\eta)} [a_{-p}^* a_q , B(\eta)] e^{sB(\eta)} \\ = \; &a_{-p}^* a_q + \int_0^1 e^{-sB(\eta)} (\eta_q b_{-p}^* b^*_{-q} + \eta_p b_q b_p) e^{s B(\eta)}  \\ = \; &a_{-p}^* a_q \\ &+ \sum_{n,k \geq 0} \frac{(-1)^{n+k}}{n!k!(n+k+1)} \left[ \eta_q \text{ad}^{(n)}_{B(\eta)} (b_{-p}^*) \text{ad}^{(k)}_{B(\eta)} (b^*_{-q}) + \eta_p \text{ad}^{(n)}_{B(\eta)} (b_q) \text{ad}^{(k)}_{B(\eta)} (b_p) \right] \end{split} \]
Inserting this identity in (\ref{eq:cGN3}), we find
\begin{equation}\label{eq:GN3-in0} 
\cG_{N,\beta}^{(3)}  = \cG_{N,\beta}^{(3,1)} + \cG_{N,\beta}^{(3,2)} + \cG_{N,\beta}^{(3,3)} 
\end{equation}
with 
\begin{equation}\label{eq:GN3-in}
\begin{split} \cG_{N,\beta}^{(3,1)} = \; &\sum_{r \geq 0} \frac{(-1)^r}{r!} \frac{\kappa}{\sqrt{N}} \sum_{p,q \in \Lambda^*_+ : p+q \not = 0} \widehat{V} (p/N^\beta) \text{ad}^{(r)}_{B(\eta)} (b^*_{p+q}) a_{-p}^* a_q + \text{h.c.} \\
\cG_{N,\beta}^{(3,2)} = \; & \sum_{n,k,r \geq 0} \frac{(-1)^{n+k+r}}{n!k!r!(n+k+1)} \\ &\hspace{.3cm} \times \frac{\kappa}{\sqrt{N}} \sum_{p,q \in \Lambda_+^* , p+q \not = 0} \widehat{V} (p/N^\beta) \eta_q \, \text{ad}^{(r)}_{B(\eta)} (b^*_{p+q}) \text{ad}^{(n)}_{B(\eta)} (b_{-p}^*) \text{ad}^{(k)}_{B(\eta)} (b^*_{-q}) + \text{h.c.} \\
\cG_{N,\beta}^{(3,3)} = \; & \sum_{n,k,r \geq 0} \frac{(-1)^{n+k+r}}{n!k!r!(n+k+1)} \\ &\hspace{.3cm} \times \frac{\kappa}{\sqrt{N}} \sum_{p,q \in \Lambda_+^* , p+q \not = 0} \widehat{V} (p/N^\beta) \eta_p \, \text{ad}^{(r)}_{B(\eta)} (b^*_{p+q}) \text{ad}^{(n)}_{B(\eta)} (b_{p}) \text{ad}^{(k)}_{B(\eta)} (b_{q}) +\text{h.c.} 
\end{split} \end{equation}

Let us consider first the term $\cG_{N,\beta}^{(3,3)}$. With Lemma \ref{lm:indu}, the operator
\begin{equation}\label{eq:L-in} \frac{\kappa}{\sqrt{N}} \sum_{p,q \in \Lambda_+^* , p+q \not = 0} \widehat{V} (p/N^\beta) \eta_p \, \text{ad}^{(r)}_{B(\eta)} (b^*_{p+q}) \text{ad}^{(n)}_{B(\eta)} (b_{p}) \text{ad}^{(k)}_{B(\eta)} (b_{q}) \end{equation}
can be expanded in the sum of $2^{n+k+r} n!k!r!$ terms having the form 
\begin{equation}\label{eq:typG}
\begin{split} \text{P} = \; &\frac{\kappa}{\sqrt{N}} \sum_{p,q \in \Lambda_+^* , p+q \not = 0} \widehat{V} (p/N^\beta) \eta_p \, \Pi^{(1)}_{\sharp , \flat} (\eta^{j_1} , \dots , \eta^{j_{k_1}}; \eta^{\ell_1}_{p+q} \ph_{\alpha_{\ell_1} (p+q)})^*  \Lambda^*_{i_1} \dots \Lambda^*_{i_1} \\ &\hspace{3cm} \times  \Lambda'_1 \dots \Lambda'_{i_2} N^{-k_2} \Pi^{(1)}_{\sharp' , \flat'} (\eta^{m_1} , \dots , \eta^{m_{k_2}}; \eta^{\ell_2}_p \ph_{\alpha_{\ell_2} p})  \\ &\hspace{3cm} \times \Lambda''_1 \dots \Lambda''_{i_3} N^{-k_3} \Pi^{(1)}_{\sharp'' , \flat''} (\eta^{s_1} , \dots , \eta^{s_{k_3}}; \eta^{\ell_3} \ph_{\alpha_{\ell_3} q}) \end{split} \end{equation}
for $i_1, i_2, i_3, k_1, k_2, k_3, \ell_1, \ell_2 , \ell_3 \in \bN$, $j_1, \dots , j_{k_1}, m_1, \dots, m_{k_2}, s_1 , \dots , s_{k_3} \in \bN \backslash \{ 0 \}$, $\alpha_{\ell_i} = (-1)^{\ell_i}$ and where each $\Lambda_i, \Lambda'_i, \Lambda''_i$ is either a factor $(N-\cN_+ )/N$, a factor $(N+1 - \cN_+ )/N$ or a $\Pi^{(2)}$-operator of the form
\begin{equation}\label{eq:Pi2-GN3} N^{-h} \Pi^{(2)}_{\underline{\sharp}, \underline{\flat}} (\eta^{z_1} , \dots , \eta^{z_s}) \end{equation}
for some $h, z_1 \dots, z_h \in \bN\backslash \{ 0 \}$. We bound the expectation of (\ref{eq:typG}) by 
\[ \begin{split} 
|\langle \xi , \text{P} \xi \rangle | \leq \; &\frac{\kappa}{\sqrt{N}} \sum_{p,q \in \Lambda^*_+ : p \not =-q} |\widehat{V} (p/N^\beta)| \eta_p \\ &\hspace{2cm} \times \| \Lambda_1 \dots \Lambda_{i_1} N^{-{k_1}} \Pi^{(1)}_{\sharp, \flat} (\eta^{j_1}, \dots , \eta^{j_{k_1}} ; \eta^{\ell_1}_{p+q} \ph_{\alpha_{\ell_1} (p+q)} ) \xi \| \\ &\hspace{2cm} \times \| \Lambda'_1 \dots \Lambda'_{i_2} N^{-k_2} \Pi^{(1)}_{\sharp' , \flat'} (\eta^{m_1} , \dots , \eta^{m_{k_2}}; \eta^{\ell_2}_p \ph_{\alpha_{\ell_2} p})  \\ &\hspace{3cm} \times \Lambda''_1 \dots \Lambda''_{i_3} N^{-k_3} \Pi^{(1)}_{\sharp'' , \flat''} (\eta^{s_1} , \dots , \eta^{s_{k_3}}; \eta^{\ell_3} \ph_{\alpha_{\ell_3} q}) \xi \| \end{split} \]
{F}rom Lemma \ref{lm:aux}, part i) and ii), we conclude that 
\[  \begin{split} 
|\langle \xi , \text{P} &\xi \rangle |\\  \leq \; &C^{n+k+r} \kappa^{n+k+r+2} \\ & \times \frac{1}{\sqrt{N}}\sum_{p,q \in \Lambda^*_+ : p \not = -q} \frac{1}{p^{2}} \Big\{ \frac{1}{(p+q)^{2}} \| (\cN_+ +1)^{1/2} \xi \| + \| a_{p+q} \xi \| \Big\} \\ &\hspace{1.5cm} \times \Big\{ \frac{(1+ r/N)}{p^{2} q^{2}} \| (\cN_+ +1) \xi \|  + \frac{(1+r/N)}{p^{2}} \| a_q (\cN_+ +1)^{1/2} \xi \| \\ & \hspace{6cm} + \frac{1}{q^{2}} \| a_p (\cN_+ +1)^{1/2} \xi \| + \| a_p a_q \xi \| \Big\} \\
\leq \; &\frac{C^{n+k+r} (1+r) \kappa^{n+k+r+2}}{\sqrt{N}}  \| (\cN_+ +1) \xi \| \| (\cN_+ +1)^{1/2} \xi \|  \end{split} \] 
Hence, for $\kappa > 0$ sufficiently small, we obtain 
\begin{equation}\label{eq:fir-3}\begin{split} \big| \langle \xi , \cG_{N,\beta}^{(3,3)} \xi \rangle \big| \leq  \frac{C \kappa^2}{\sqrt{N}} \| (\cN_+ +1) \xi \| \| (\cN_+ +1)^{1/2} \xi \| \end{split} \end{equation}

Next, we consider the term $\cG_{N,\beta}^{(3,2)}$ in  (\ref{eq:GN3-in}) (we take its hermitian conjugate). Since we will use the potential energy operator to control this term, it is convenient to switch to position space. We write
\begin{equation}\label{eq:adbadbadb}  \begin{split} 
\frac{\kappa}{\sqrt{N}} &\sum_{p,q \in \Lambda_+^* , p+q \not = 0} \widehat{V} (p/N^\beta) \eta_q \, \text{ad}^{(r)}_{B(\eta)} (b_{-q}) \text{ad}^{(n)}_{B(\eta)} (b_{-p}) \text{ad}^{(k)}_{B(\eta)} (b_{p+q})  \\ &= \frac{\kappa}{\sqrt{N}} \int_{\Lambda \times \Lambda} dx dy \, N^{3\beta} V(N^\beta(x-y)) \text{ad}^{(r)}_{B(\eta)} (b (\check{\eta}^{1+\ell_1}_x) \text{ad}^{(n)}_{B(\eta)} (\check{b}_y) \text{ad}^{(k)}_{B(\eta)} (\check{b}_x) 
\end{split} \end{equation}
where we used the notation $\check{\eta}^{s}$ to indicate the Fourier transform of the sequence $\Lambda^* \ni p \to \eta^s_p$, and $\check{\eta}^s_x$ denotes the function (or the distribution, if $s =0$) $z \to \check{\eta}^s_x (z) = \check{\eta}^s (z-x)$.  
With Lemma \ref{lm:indu}, the r.h.s. of (\ref{eq:adbadbadb}) can be written as the sum of $2^{n+k+r} n!k!r!$ terms, all having the form 
\begin{equation}\label{eq:Hop} \begin{split} 
\text{Q}  = \; & \frac{\kappa}{\sqrt{N}} \int_{\Lambda \times \Lambda} dx dy \, N^{3\beta} V(N^\beta (x-y)) \, \Lambda_1 \dots \Lambda_{i_1} N^{-k_1} 
\Pi^{(1)}_{\sharp , \flat} (\eta^{j_1} , \dots , \eta^{j_{k_1}}; \check{\eta}^{1+\ell_1}_x )  \\ & \hspace{3cm} \times \Lambda'_1 \dots \Lambda'_{i_2}  N^{-k_2} \Pi^{(1)}_{\sharp' , \flat'} (\eta^{m_1} , \dots , \eta^{m_{k_2}}; \check{\eta}^{\ell_2}_y)  \\ & \hspace{3cm} \times \Lambda^{''}_1 \dots \Lambda^{''}_{i_3} N^{-k_3} \Pi^{(1)}_{\sharp^{''} , \flat^{''}} (\eta^{s_1} , \dots , \eta^{s_{k_3}}; \check{\eta}^{\ell_3}_{x}) \end{split} 
\end{equation}
where $i_1, i_2, i_3, k_1, k_2, k_3, \ell_1, \ell_2 , \ell_3 \in \bN$, $j_1, \dots , j_{k_1}, m_1, \dots, m_{k_2}, s_1 , \dots , s_{k_3} \in \bN \backslash \{ 0 \}$ and where each operator $\Lambda_i, \Lambda'_i, \Lambda''_i$ is either a factor $(N-\cN_+ )/N$, a factor $(N+1 - \cN_+ )/N$ or a $\Pi^{(2)}$-operator of the form (\ref{eq:Pi2-GN3}). To estimate the expectation of (\ref{eq:Hop}), we first assume that $(\ell_2, \ell_3) \not = (0,1)$. Under this assumption, we bound 
\begin{equation}\label{eq:MM} \begin{split} 
|\langle \xi , \text{Q} \xi \rangle | \leq \; &\frac{\kappa}{\sqrt{N}} \int_{\Lambda \times \Lambda} dx dy \, N^{3\beta} V(N^\beta (x-y))\\ &\hspace{1cm} \times \| N^{-k_1} \Pi^{(1)}_{\sharp,\flat} (\eta^{j_1}, \dots, \eta^{j_{k_1}} ; \check{\eta}^{\ell_1+ 1}_x )^* \Lambda_{i_1}^* \dots \Lambda^*_1 \xi \| \\ &\hspace{1cm} \times \Big\| \Lambda'_1 \dots \Lambda'_{i_2}  N^{-k_2} \Pi^{(1)}_{\sharp' , \flat'} (\eta^{m_1} , \dots , \eta^{m_{k_2}}; \check{\eta}^{\ell_2}_y)  \\ & \hspace{2cm} \times \Lambda^{''}_1 \dots \Lambda^{''}_{i_3} N^{-k_3} \Pi^{(1)}_{\sharp^{''} , \flat^{''}} (\eta^{s_1} , \dots , \eta^{s_{k_3}}; \check{\eta}^{\ell_3}_{x}) \xi \Big\| \end{split} \end{equation}
With Lemma \ref{lm:aux2} we estimate
\begin{equation}\label{eq:MM1} \| N^{-k_1} \Pi^{(1)}_{\sharp,\flat} (\eta^{j_1}, \dots, \eta^{j_{k_1}} ; \check{\eta}^{\ell_1+ 1}_x )^* \Lambda_{i_1}^* \dots \Lambda^*_1 \xi \| \leq C^r \kappa^{r+1} 
\| (\cN_+ +1)^{1/2} \xi \| \end{equation}
and, using the condition $(\ell_2, \ell_3) \not = (0,1)$, 
\begin{equation*} \begin{split}  &\Big\| \Lambda'_1 \dots \Lambda'_{i_2}  N^{-k_2} \Pi^{(1)}_{\sharp' , \flat'} (\eta^{m_1} , \dots , \eta^{m_{k_2}}; \check{\eta}^{\ell_2}_y)   \Lambda^{''}_1 \dots \Lambda^{''}_{i_3} N^{-k_3} \Pi^{(1)}_{\sharp^{''} , \flat^{''}} (\eta^{s_1} , \dots , \eta^{s_{k_3}}; \check{\eta}^{\ell_3}_{x}) \xi \Big\| \\ &\hspace{1cm} \leq C^{n+k} \kappa^{n+k} \Big\{ (1+k/N) \| (\cN_+ +1) \xi \| + (1+k/N) \| \check{a}_x (\cN_+ +1)^{1/2} \xi \| \\ &\hspace{5cm} + \|\check{a}_y (\cN_+  +1)^{1/2} \xi \| + \| \check{a}_x \check{a}_y \xi \|  \Big\} \, . \end{split} 
\end{equation*}
Inserting these bounds in (\ref{eq:MM}), we arrive at
\begin{equation}\label{eq:Mbd} \begin{split} |\langle \xi , \text{Q} \xi \rangle | \leq \; &C^{n+k+r} 
\kappa^{n+k+r+2} (1+k) \| (\cN_+ +1)^{1/2} \xi \| \\ &\times \frac{1}{\sqrt{N}} \int_{\Lambda \times \Lambda} dx dy \, N^{3\beta} V(N^\beta (x-y)) \\ &\hspace{1cm} \times \Big\{ \| (\cN_+ +1) \xi \| + \| \check{a}_x (\cN_+ +1)^{1/2} \| + \| \check{a}_y (\cN_+ +1)^{1/2} \xi \| + \| \check{a}_x \check{a}_y \xi \| \Big\}  \\ \leq \; & \frac{C^{n+k+r} \kappa^{n+k+r+2} (1+k)}{\sqrt{N}}  \| (\cN_+ +1) \xi \| \| (\cN_+ +1)^{1/2} \xi \|  \\ &+ C^{n+k+r} \kappa^{n+k+r+1} (1+k) \| \cV_N^{1/2} \xi \| \| (\cN_+ +1)^{1/2} \xi \| \,.
\end{split} \end{equation}
For $(\ell_2,\ell_3) = (0,1)$ we can proceed similarly. The only additional remark is that, in this case, the the commutator \[ [\check{a}_y , a^* (\check{\eta}_x)] = \check{\eta} (x-y) \]
between the annihilation operator associated with the second $\Pi^{(1)}$-factor (the one containing $\check{\eta}_y^{\ell_2}$) and the creation operator $a^* (\check{\eta}_x)$ associated with the third $\Pi^{(1)}$-operator, gives a vanishing contribution to the expectation $\langle \xi , \text{Q} \xi \rangle$, for all $\xi \in \cF_+^{\leq N}$ (because of the assumption that $\xi$ is orthogonal to $\ph_0$).

With (\ref{eq:Mbd}) we conclude that, if $\kappa > 0$ is small enough, 
\begin{equation}\label{eq:sec-3} \begin{split} 
\big| \langle \xi , \cG_{N,\beta}^{(3,2)} \xi \rangle \big| \leq \frac{C \kappa^2}{\sqrt{N}} \| (\cN_+ +1)^{1/2} \xi \| \| (\cN_+ +1) \xi \| + C \kappa \| \cV_N^{1/2} \xi \| \| (\cN_+ +1)^{1/2} \xi \| \end{split} 
\end{equation}

Finally, we consider the term $\cG_{N,\beta}^{(3,1)}$ in (\ref{eq:GN3-in}). {F}rom Lemma \ref{lm:indu}, each operator
\begin{equation}\label{eq:first-3} \frac{\kappa}{\sqrt{N}} \sum_{p,q \in \Lambda^*_+ : p+q \not = 0} \widehat{V} (p/N^\beta) \text{ad}^{(r)}_{B(\eta)} (b^*_{p+q}) a^*_{-p} a_q \end{equation}
can be written as the sum of $2^r r!$ terms having the form
\begin{equation}\label{eq:rsum} \text{R} = \frac{\kappa}{\sqrt{N}} \sum_{p,q \in \Lambda^*_+ : p+q \not = 0} \widehat{V} (p/N^\beta) N^{-k_1} \Pi^{(1)}_{\sharp, \flat} (\eta^{j_1}, \dots , \eta^{j_{k_1}} ; \eta_{p+q}^{\ell_1}  \ph_{\alpha_{\ell_1} (p+q)} )^* \Lambda_{i_1}^* \dots \Lambda_1^* a_{-p}^* a_q \end{equation}
for $i_1, k_1, \ell_1 \in \bN$, $j_1, \dots , j_{k_1} \in \bN \backslash \{ 0 \}$, $\alpha_{\ell_1} = (-1)^{\ell_1}$, and where each $\Lambda_j$ operator is either a factor $(N-\cN_+ )/N$, a factor $(N+1-\cN_+ )/N$ of a $\Pi^{(2)}$-operator of the form (\ref{eq:Pi2-GN3}). If $\ell_1 \geq 2$, we use Lemma \ref{lm:aux}, part iii), to bound 
\begin{equation}\label{eq:G31-1} \begin{split} |\langle \xi , \text{R} \xi \rangle | &\leq \frac{C \kappa}{\sqrt{N}} \sum_{p,q \in \Lambda^*_+ : p \not = -q} |\eta_{p+q}|^{\ell_1} \, \| a_q \xi \| \\ &\hspace{1cm} \times  \| a_{-p} \Lambda_1 \dots \Lambda_{i_1} N^{-k_1} \Pi^{(1)}_{\sharp, \flat} (\eta^{j_1}, \dots , \eta^{j_{k_1}} ;  \ph_{\alpha_{\ell_1} (p+q)}) \xi \|  \\ &\leq \frac{C^r \kappa^{r+1}}{\sqrt{N}} \sum_{p,q \in \Lambda^*_+ : p \not = -q} \frac{1}{(p+q)^{4}} \| a_q \xi \| \left\{ \| a_{-p} (\cN_+ +1)^{1/2} \xi \| + \frac{r}{N p^2} \| (\cN_+ +1) \xi \| \right\} 
\\ &\leq \frac{C^r \kappa^{r+1} (1+r)}{\sqrt{N}}  \| (\cN_+ +1)^{1/2} \xi \| \| (\cN_+ +1) \xi \|  \end{split} \end{equation}

If $\ell_1 = 1$, we commute the operator $a_{-(p+q)}$ (or the $b_{-(p+q)}$ operator) appearing in the $\Pi^{(1)}$-operator in (\ref{eq:rsum}) to the right, and the operator $a_{-p}^*$ to the left (it is important to note that $[a_{-(p+q)}, a_{-p}^*] = 0$ since $q \not = 0$). Lemma \ref{lm:aux}, part iii), implies that 
\begin{equation}\label{eq:G31-2} \begin{split} 
|\langle \xi , \text{R} \xi \rangle | &\leq \frac{C^r \kappa^{r+1}}{\sqrt{N}} \sum_{p,q \in \Lambda^*_+: p \not = -q} |\widehat{V} (p/N^\beta)| \frac{1}{(p+q)^2} \\ & \hspace{1cm} \times \Big\{ \frac{r}{N p^2} \| (\cN_+ +1) \xi \| \| a_q \xi \|  + \frac{1}{N (p+q)^2} \| a_{-p} (\cN_+ +1)^{1/2} \xi \| \| a_q \xi \| \\ &\hspace{7cm}  + \| a_{-p} \xi \| \| a_{-(p+q)} a_q \xi \| \Big\} \\ &\leq \frac{C^r \kappa^{r+1}}{\sqrt{N}} \| (\cN_+ +1) \xi \| \| (\cN_+ +1)^{1/2} \xi \| \end{split} \end{equation}

Finally, if $\ell_1 = 0$, we commute $a^*_{-p}$ to the left. With Lemma \ref{lm:aux}, we find  
\begin{equation}\label{eq:ell10} \begin{split} \langle \xi , \text{R} \xi \rangle = \; &\frac{\kappa}{\sqrt{N}} \sum_{p,q \in \Lambda_+^* : p+q \not = 0} \widehat{V} (p/N^\beta) \langle \text{D}_1 (p,q) , a_q \xi \rangle \\ &+ \frac{\kappa}{\sqrt{N}} \sum_{p,q \in \Lambda_+^* : p+q \not = 0} \widehat{V} (p/N^\beta) \langle \text{D}_2 \, a_{-p} a_{p+q} \xi , a_q \xi \rangle \end{split}  \end{equation}
where \[ \| \text{D}_1 (p,q) \| \leq \frac{C^r \kappa^r r}{N p^2} \| a_{p+q} (\cN_+ +1)^{1/2}  \xi \|  \]
and $\| \text{D}_2 \| \leq C^r \kappa^r$. Switching to position space to control the second term on the r.h.s. of (\ref{eq:ell10}), we conclude therefore that 
\[ \begin{split} |\langle \xi, \text{R} \xi \rangle| \leq \; &\frac{C^r \kappa^{r+1} r}{N^{3/2}} \sum_{p,q \in \Lambda^*_+} \frac{|\widehat{V} (p/N^\beta)|}{p^2} \| a_{q+p} (\cN_+ +1)^{1/2} \xi \| \| a_q \xi \| \\ &+ \frac{C^r \kappa^{r+1}}{\sqrt{N}} \int_{\Lambda \times \Lambda} dx dy \, N^{3\beta} V(N^\beta (x-y)) \| \check{a}_x \check{a}_y \xi \| \| \check{a}_y \xi \|  \\ \leq\;& \frac{C^r \kappa^{r+1} r}{\sqrt{N}} \| (\cN_+ +1) \xi \| \| (\cN_+ +1)^{1/2} \xi \| + C^r \kappa^{r+1/2} \| \cV_N^{1/2} \xi \| \| (\cN_+ +1)^{1/2} \xi \| \end{split} \]
Together with (\ref{eq:G31-1}) and (\ref{eq:G31-2}), the last estimate implies that, if $\kappa > 0$ is small enough, 
\begin{equation}\label{eq:G31-fin} |\langle \xi, \cG_{N,\beta}^{(3,1)} \xi \rangle | \leq \frac{C \kappa}{\sqrt{N}} \| (\cN_+ +1)^{1/2} \xi \| \| (\cN_+ +1) \xi \| + C \kappa^{1/2} \| \cV_N^{1/2} \xi \| \| (\cN_+ +1)^{1/2} \xi \| \end{equation}
Combining the last bound with (\ref{eq:fir-3}) and (\ref{eq:sec-3}) (and using the fact that $\cN_+  \leq N$ on $\cF_+^{\leq N}$), we easily obtain that for every $\delta > 0$ there exists $C > 0$ with 
\[ \pm \cG_{N,\beta}^{(3)}  \leq \delta \cV_N + C \kappa (\cN_+ +1) \]
As usual, we can show the same bound for the commutator of $\cG_{N,\beta}^{(3)}$ with $\cN_+ $ (simply because the commutator of $\cN_+ $ with all terms of the form (\ref{eq:typG}), (\ref{eq:Hop}) and (\ref{eq:rsum}) has again the same form, up to a constant bounded by $C(n+k+r)$), i.e. 
\[ \pm \big[ \cG_{N,\beta}^{(3)} , i \cN_+  \big]   \leq \delta \cV_N + C \kappa (\cN_+ +1) \]
Finally, combining (\ref{eq:fir-3}), (\ref{eq:sec-3}) and (\ref{eq:G31-fin}) with Lemma \ref{lm:cVN}, we also  obtain 
\[ \pm \cG_{N,\beta}^{(3)} \leq C N^{(\beta-1)/2} (\cN_+ +1) (\cK +1) \, . \]

\end{proof}

\subsection{Analysis of $\cG_N^{(4)}$}

{F}rom (\ref{eq:cLNj}) we have
\[ \cG_{N,\beta}^{(4)} = e^{-B(\eta)} \cL_{N,\beta}^{(4)} e^{B(\eta)} = \frac{1}{2N} \sum_{p,q \in \Lambda_+^* , r \in \Lambda^* : r \not = -p, -q} \widehat{V} (r/N^\beta) e^{-B(\eta)} a_{p+r}^* a_q^* a_p a_{q+r} e^{B(\eta)} \] 
We define the operator $\cE_N^{(4)}$ through 
\begin{equation}\label{eq:def-cEN4} \begin{split} \cG_{N,\beta}^{(4)} = \; &\cV_N + \frac{1}{2N} \sum_{p,q \in \Lambda^*_+} \widehat{V} ((p-q)/N^\beta) \eta_p \eta_q \\ &+ \frac{1}{2N} \sum_{p,q \in \Lambda^*_+} \widehat{V} ((p-q)/N^\beta) \eta_q (b_p^* b_{-p}^* + b_p b_{-p} ) + \cE_{N,\beta}^{(4)} \end{split} \end{equation}
In the next proposition, we estimate the error term $\cE_{N,\beta}^{(4)}$.
\begin{prop}\label{prop:G4} 
Under the assumptions of Theorem \ref{thm:gene}, for every $\delta > 0$ there exists $C > 0$ such that, on $\cF_+^{\leq N}$,  
\[ \begin{split} 
\pm \cE_{N,\beta}^{(4)} &\leq \delta \cV_N + C \kappa (\cN_+ +1) 
\\
\pm \left[ \cE_{N,\beta}^{(4)} , i \cN_+  \right] &\leq C (\cH_N^\beta + 1) \end{split} \]
Furthermore, we find 
\[ \pm \cE_{N,\beta}^{(4)} \leq C N^{(\beta-1)/2} (\cH_N^\beta + 1) ( \cN_+ +1) \]
\end{prop}

\begin{proof}
We have
\begin{equation}\label{eq:quartic1} \begin{split} 
\cG_{N,\beta}^{(4)} = \;& \frac{\kappa}{2N} \sum_{p,q \in \Lambda_+^*, r \in \Lambda^* : r \not = -p,q} \widehat{V} (r/N) e^{-B(\eta)} a_p^* a_q^* a_{q-r} a_{p+r}  e^{B(\eta)} \\ =\; & \cV_N + \frac{\kappa}{2N} \sum_{p,q \in \Lambda_+^*, r \in \Lambda^* : r \not = -p,q} \widehat{V} (r/N) \int_0^1 ds \,  e^{-sB(\eta)} \left[ a_p^* a_q^* a_{q-r} a_{p+r} , B(\eta) \right] e^{sB(\eta)}
 \\ = \; &\cV_N + \frac{\kappa}{2N} \sum_{q \in \Lambda_+^*, r \in \Lambda^* : r \not = -q} \widehat{V} (r/N) \eta_{q+r} \int_0^1 ds \, \left( e^{-sB(\eta)} b_q^* b_{-q}^* e^{sB(\eta)} + \text{h.c.} \right)  \\
& +\frac{\kappa}{N} \sum_{p,q \in \Lambda_+^* , r \in \Lambda^* : r \not = p,-q} \widehat{V} (r/N) \eta_{q+r} \int_0^1 ds \left( e^{-s B(\eta)} b_{p+r}^* b_q^* a^*_{-q-r} a_p e^{sB(\eta)} + \text{h.c.} \right) \end{split} \end{equation}
Expanding again 
\[ \begin{split}  e^{-sB(\eta)} a^*_{-q-r} a_p e^{s B(\eta)}  &=
a^*_{-q-r} a_p + \int_0^s d\tau \, e^{-\tau B(\eta)} \left[ a^*_{-q-r} a_p , B(\eta) \right] e^{-\tau B(\eta)} \\ &= a^*_{-q-r} a_p + \int_0^s d\tau \, e^{-\tau B(\eta)} \left( \eta_p b^*_{-p} b^*_{-q-r} + \eta_{q+r} b_p b_{q+r} \right) e^{-\tau B(\eta)} \end{split}  \]
and using Lemma \ref{lm:conv-series}, we obtain 
\begin{equation}\label{eq:cG4-def}  \cG_{N,\beta}^{(4)} -\cV_N = \text{W}_1 + \text{W}_2 + \text{W}_3 + \text{W}_4 \end{equation}
where we defined 
\begin{equation}\label{eq:defW} \begin{split} \text{W}_1 = \; &
\sum_{n,k =0}^\infty \frac{(-1)^{n+k}}{n!k!(n+k+1)} 
 \\ &\hspace{.5cm} \times \frac{\kappa}{2N} \sum_{q \in \Lambda_+^*, r \in \Lambda^* : r \not = -q} \widehat{V} (r/N^\beta) \eta_{q+r} \left( \text{ad}^{(n)}_{B(\eta)} (b_q) \text{ad}^{(k)}_{B(\eta)} ( b_{-q}) + \text{h.c.} \right) \\ \text{W}_2 = \; & \sum_{n,k =0}^\infty \frac{(-1)^{n+k}}{n!k!(n+k+1)}  \\ &\hspace{.5cm} \times  \frac{\kappa}{N} \sum_{p,q \in \Lambda_+^* , r \in \Lambda^* : r \not = p,-q} \widehat{V} (r/N^\beta) \eta_{q+r} \, \left( \text{ad}^{(n)}_{B(\eta)} (b_{p+r}^*) \text{ad}^{(k)}_{B(\eta)} (b_q^*)  a^*_{-q-r} a_p + \text{h.c.} \right)  
\end{split} \end{equation} 
and 
\begin{equation}\label{eq:defW2}
\begin{split}
\text{W}_3 = \; &  \sum_{n,k,i,j =0}^\infty \frac{(-1)^{n+k+i+j}}{n!k!i!j!(i+j+1)(n+k+i+j+2)} \\ &\hspace{1cm}  \frac{\kappa}{N} \sum_{p,q\in \Lambda^*_+, r \in \Lambda^* : r \not = -p -q} \widehat{V} (r/N^\beta) \eta_{q+r} \eta_p  \\ &\hspace{2cm} \times \left( \text{ad}^{(n)}_{B(\eta)} (b^*_{p+r}) \text{ad}^{(k)}_{B(\eta)} (b^*_q) \text{ad}^{(i)}_{B(\eta)} (b_{-p}^*) \text{ad}^{(j)}_{B(\eta)} (b_{-q-r}^*) + \text{h.c.} \right) \\ \text{W}_4 = \; & \sum_{n,k,i,j =0}^\infty \frac{(-1)^{n+k+i+j}}{n!k!i!j!(i+j+1)(n+k+i+j+2)} \\ &\hspace{1cm} \times  \frac{\kappa}{N} \sum_{p,q\in \Lambda^*_+, r \in \Lambda^* : r \not = -p -q} \widehat{V} (r/N^\beta) \eta^2_{q+r}  \\ & \hspace{2cm} \times  \left( \text{ad}^{(n)}_{B(\eta)} (b^*_{p+r}) \text{ad}^{(k)}_{B(\eta)} (b^*_q) \text{ad}^{(i)}_{B(\eta)} (b_{p}) \text{ad}^{(j)}_{B(\eta)} (b_{q+r}) + \text{h.c.} \right) 
  \end{split} \end{equation}
In $\text{W}_1$, we isolate the contributions associated to $(n,k) = (0,0), (0,1)$. To this end, we write
\[ \begin{split} \text{W}_1 = \; &\frac{\kappa}{2N} \sum_{q \in \Lambda_+^*, r \in \Lambda^* : r \not = -q} \widehat{V} (r/N^\beta) \eta_{r+q} (b_q b_{-q} + \text{h.c.}) \\ &- \frac{\kappa}{4N} \sum_{q \in \Lambda_+^*, r \in \Lambda^* : r \not = -q}  \widehat{V} (r/N^\beta) \eta_{q+r} (b_q [B(\eta) , b_{-q}] + \text{h.c.}) + \wt{\text{W}}_1 \\ 
= \; &\frac{\kappa}{2N} \sum_{q \in \Lambda_+^*, r \in \Lambda^* : r \not = -q} \widehat{V} (r/N^\beta) \eta_{r+q} (b_q b_{-q} + \text{h.c.}) \\ &- \frac{\kappa}{4N} \sum_{q \in \Lambda_+^*, r \in \Lambda^* : r \not = -q}  \widehat{V} (r/N^\beta) \eta_{q+r} \eta_q + \text{T} + \wt{\text{W}}_1 
\end{split} \]
where we defined 
\begin{equation}\label{eq:wtW1def} \begin{split} \wt{W}_1 = \; &\sum_{n,k}^* \frac{(-1)^{n+k}}{n!k!(n+k+1)}  \\ &\hspace{.5cm} \times \frac{\kappa}{2N} \sum_{q \in \Lambda_+^*, r \in \Lambda^* : r \not = -q} \widehat{V} (r/N^\beta) \eta_{q+r} \left( \text{ad}^{(n)}_{B(\eta)} (b_q) \text{ad}^{(k)}_{B(\eta)} ( b_{-q}) + \text{h.c.} \right) \end{split} \end{equation}
with the sum $\sum_{n,k}^*$ running over all pairs $(n,k) \not = (0,0), (0,1)$, and 
\begin{equation}\label{eq:Tdef} \begin{split} \text{T} = \; &- \frac{\kappa}{4N} \sum_{q \in \Lambda_+^*, r \in \Lambda^* : r \not = -q}  \widehat{V} (r/N^\beta) \eta_{q+r} (b_q [B(\eta) , b_{-q}] + \text{h.c.}) \\ &+ \frac{\kappa}{2N} \sum_{q \in \Lambda_+^*, r \in \Lambda^* : r \not = -q} \widehat{V} (r/N^\beta) \eta_{q+r} \eta_q \\ =: \; & \text{T}_1 + \text{T}_2 + \text{T}_3 \end{split} \end{equation}
with
\[ \begin{split}  \text{T}_1 &= \frac{\kappa}{N^2} \sum_{q \in \Lambda^*_+, r\in \Lambda^* : r \not = -q} \widehat{V} (r/N^\beta) \eta_{r+q} \eta_q (2\cN_+ +1 + \cN_+ /N + \cN_+ ^2/N )  \\
\text{T}_2 &= \frac{2\kappa}{N^2} \sum_{q\in \Lambda^*_+, r\in \Lambda^* : r \not = -q} \widehat{V} (r/N^\beta) \eta_{r+q} \eta_q a_q^* a_q \left( 1 - \frac{\cN_+ +1}{N} \right) \\
\text{T}_3 &= \frac{\kappa}{N^3} \sum_{q,m \in \Lambda^*_+, r\in \Lambda^* :r \not = -q} \widehat{V} (r/N^\beta) \eta_{r+q} \eta_m  a_m^* a^*_{-m} a_q a_{-q} \end{split} \]
In the computation of $\text{T}$, we used the fact that
\[ [B(\eta), b_{-q} ] = -\eta_q(1-\cN_+ /N) b_q^*   + \frac{1}{N} \sum_{m \in \Lambda^*_+}  \eta_m b^*_m a^*_{-m} a_{-q} \]
Comparing with (\ref{eq:def-cEN4}), we arrive at
\begin{equation}\label{eq:cEN4-in} \cE_{N,\beta}^{(4)} = \text{T} + \wt{\text{W}}_1 + \text{W}_2 + \text{W}_3 + \text{W}_4 \end{equation}

Let us start by analyzing the operator $\text{T}$, defined in (\ref{eq:Tdef}).
Using (\ref{eq:etapN}), we estimate
\[ \frac{1}{N^2} \sum_{q \in \Lambda^*_+ , r \in \Lambda^* : r \not = -q} |\widehat{V} (r/N^\beta )| |\eta_{r+q}| |\eta_q|  \leq C N^{\beta-1} \left[  \frac{1}{N^2} \sum_{q \in \Lambda^*_+ , r \in \Lambda^* : r \not = -q} \frac{|\widehat{V} (r/N^\beta)|^2}{q^2 (q+r)^2} \right]^{1/2} \]
Next, we observe that
\[ \begin{split} 
\frac{1}{N^2} \sum_{q \in \Lambda^*_+ , r \in \Lambda^* : r \not = -q} \frac{|\widehat{V} (r/N^\beta)|^2}{q^2 (q+r)^2} &\leq \frac{C}{N^2} \sum_{r \in \Lambda^*} \frac{|\widehat{V} (r/N^\beta)|^2}{|r|+1} \\ &\leq \frac{C \| \widehat{V} \|^2_{\infty}}{N^2} \sum_{|r| \leq N^\beta} \frac{1}{|r|+1}  + \frac{C\| \widehat{V} (./N^\beta) \|^2_2}{N^{2+\beta}} \leq C N^{2(\beta-1)} \end{split} \]
Hence, we conclude that 
\begin{equation}\label{eq:bd-eta-eta} \frac{1}{N^2} \sum_{q \in \Lambda^*_+ , r \in \Lambda^* : r \not = -q} |\widehat{V} (r/N^\beta)| |\eta_{r+q}| |\eta_q|  \leq C N^{2 (\beta-1)} \end{equation}
With this bound, we easily arrive at  
\begin{equation}\label{eq:T1T2}  |\langle \xi , \text{T}_1 \xi \rangle |, |\langle \xi , \text{T}_2 \xi \rangle | \leq C \kappa^3 N^{2(\beta-1)} \| (\cN_+ +1)^{1/2} \xi \|^2 \end{equation}
To bound $\text{T}_3$, we switch to position space. We obtain 
\[ \begin{split} \text{T}_3 &= \frac{\kappa}{N^3} \sum_{q,m \in \Lambda^*_+, r\in \Lambda^* :r \not = -q} \widehat{V} (r/N^\beta) \eta_{r+q} \eta_m  a_m^* a^*_{-m} a_q a_{-q} \\ &=  \frac{\kappa}{N^3} \int_{\Lambda \times \Lambda} dx dy \, N^{3\beta} V(N^\beta (x-y)) \check{\eta} (x-y) \text{D} \check{a}_x \check{a}_y \end{split} \]
where $\text{D} = \sum_{m\in \Lambda_+^*} \eta_m a_m^* a_{-m}^*$. Since $\| \text{D}^* \xi \| \leq C \kappa \| (\cN_+ +1) \xi \|$, we find 
\[ \begin{split}  | \langle \xi , \text{T}_3 \xi \rangle | &\leq C \kappa^2  N^{-3} \| (\cN_+ +1) \xi \| \int_{\Lambda \times \Lambda} dxdy \, N^{3\beta} V(N^\beta (x-y)) |\check{\eta} (x-y)| \| \check{a}_x \check{a}_y \xi \| \\ 
&\leq C \kappa^3 N^{-3+\beta}  \| (\cN_+ +1) \xi \| \int_{\Lambda \times \Lambda}  dx dy \, N^{3\beta} V(N^\beta (x-y))  \| \check{a}_x \check{a}_y \xi \| \\ &\leq C \kappa^{5/2} N^{-1} \| (\cN_+ +1)^{1/2} \xi \| \| \cV_N^{1/2} \xi \| \end{split} \]
Together with (\ref{eq:T1T2}), we conclude that 
\begin{equation}\label{eq:Tfin} |\langle \xi, \text{T} \xi \rangle | \leq C\kappa^3  N^{2(\beta-1)}\| (\cN_+ +1)^{1/2} \xi \|^2 +   C \kappa^{5/2} N^{-1} \| (\cN_+ +1)^{1/2} \xi \| \| \cV_N^{1/2} \xi \| \, . 
\end{equation}
 
Let us now consider the operator $\wt{W}_1$, defined in 
(\ref{eq:wtW1def}). According to Lemma \ref{lm:indu}, the operator
\[ \frac{\kappa}{N} \sum_{q \in \Lambda^*_+, r\in \Lambda^* : r \not = -q} \widehat{V} (r/N^\beta) \eta_{q+r} \text{ad}^{(n)} (b_q) \text{ad}^{(k)} (b_{-q})  \]
can be written as the sum of $2^{n+k} n!k!$ term s having the form
\[ \begin{split} \text{X} = \frac{\kappa}{N} \sum_{q\in \Lambda^*_+, r\in \Lambda^* : r \not = -q} \widehat{V} (r/N^\beta) \eta_{q+r} \Lambda_1 &\dots \Lambda_{i_1} N^{-k_1} \Pi^{(1)}_{\sharp,\flat} (\eta^{j_1} , \dots , \eta^{j_{k_1}} ; \eta^{\ell_1}_q \ph_{\alpha_{\ell_1} q} ) \\ &\times \Lambda'_1 \dots \Lambda'_{i_2} N^{-k_2} \Pi^{(1)}_{\sharp', \flat'} (\eta^{m_1} , \dots , \eta^{m_{k_2}} ; \eta^{\ell_2}_q \ph_{-\alpha_{\ell_2} q} ) \end{split} \]
where $i_1, i_2, k_1, k_2, \ell_1, \ell_2 \in \bN$, $j_1, \dots , j_{k_1}, m_1, \dots , m_{k_2} \in \bN \backslash \{0 \}$, $\alpha_{\ell_i} = (- 1)^{\ell_i}$  and where each operator $\Lambda_r, \Lambda'_r$ is either a factor $(N-\cN_+ )/N$, a factor $(N+1-\cN_+ )/N$ or a $\Pi^{(2)}$-operator of the form 
\begin{equation}\label{eq:Pi2-E4} N^{-h} \Pi^{(2)}_{\underline{\sharp}, \underline{\flat}} (\eta^{z_1}, \dots , \eta^{z_p}) .\end{equation} 
for $h, z_1, \dots, z_h \in \bN \backslash \{ 0 \}$. To bound the expectation of the operator $\text{X}$, we distinguish two cases. If $\ell_1 + \ell_2 \geq 1$, we use Lemma \ref{lm:aux}, part ii), to estimate
\[ \begin{split} |\langle \xi , \text{X} \xi  \rangle | \leq \; &\frac{C^{n+k} \kappa^{n+k+2}}{N}  \| (\cN_+ +1)^{1/2} \xi \|  \\ & \hspace{1cm} \times \sum_{q,r \in \Lambda^*_+ : r \not = -q} \frac{|\widehat{V} (r/N^\beta)|}{(q+r)^2}  \left\{ \frac{1}{q^4} (1+k/N) \| (\cN_+ +1)^{1/2} \xi \| + \frac{1}{q^2} \| a_q \xi \|
\right\} \\ &+ \frac{C^{n+k} \kappa^{n+k} \| (\cN_+ +1)^{1/2} \xi \|^2}{N^2}  \sum_{q,r \in \Lambda^*_+ : r \not = -q} |\widehat{V} (r/N^\beta)| |\eta_{q+r}| |\eta_q|  \end{split} \]
Here we used the fact that we excluded the pairs $(n,k) = (0,0), (0,1)$ to make sure that, if $\ell_1 =0$ and $\ell_2 = 1$, then either $k_1 > 0$ or $k_2 > 0$ or at least one of the operators $\Lambda$ or $\Lambda'$ has to be a $\Pi^{(2)}$-operator. {F}rom (\ref{eq:bd-eta-eta}) and since, as we already showed in (\ref{eq:eta-bd1}),  
\[ \sup_{q \in \Lambda^*_+} \frac{1}{N} \sum_{r\in \Lambda^* : r \not = -q} |\widehat{V} (r/N^\beta)| \frac{1}{(q+r)^2} \leq C N^{\beta-1}  \] 
we conclude that, for $\ell_1 + \ell_2 \geq 1$, 
\begin{equation}\label{eq:X1} |\langle \xi, \text{X} \xi \rangle | \leq C^{n+k} \kappa^{n+k+2} N^{\beta-1} \| (\cN_+ +1)^{1/2} \xi \|^2 \end{equation}

For $\ell_1 = \ell_2 = 0$, we use Lemma \ref{lm:aux}, part ii), to write 
\[ \text{X} = \frac{\kappa}{N} \sum_{q\in \Lambda^*_+, r\in \Lambda^*} \widehat{V}(r/N^\beta) \eta_{q+r} \left[ \text{D}_{1} (q) +  \, \text{D}_2 \, a_q a_{-q} \right] =: \text{X}_1 + \text{X}_2 \]
where 
\[ |\langle \xi , \text{D}_1 (q) \xi \rangle | \leq  \frac{C^{n+k} \kappa^{n+k} k}{N q^2} \| (\cN_+ +1)^{1/2} \xi \|^2 \]
and (since we excluded the term with $(n,k) = (0,0)$) 
\[ \| \text{D}_2^* \xi \| \leq C^{n+k} N^{-1}  \kappa^{n+k} \| (\cN_+ +1) \xi \| \] 
We immediately obtain, using again (\ref{eq:bd-eta-eta}), that
\[ \begin{split} 
|\langle \xi , \text{X}_1 \xi \rangle | &\leq \frac{C^{n+k} \kappa^{n+k+2}}{N^2} \sum_{q\in \Lambda^*_+, r\in \Lambda^*} \widehat{V} (r/N^\beta) \frac{1}{(q+r)^2 q^2} \| (\cN_+ +1)^{1/2} \xi \|^2 \\ &\leq C^{n+k} \kappa^{n+k+2} N^{2(\beta-1)} \| (\cN_+ +1)^{1/2} \xi \|^2 \end{split} \]
Switching to position space, we also find 
\[ \begin{split} 
| \langle \xi , \text{X}_2 \xi \rangle |  &=  \Big| \frac{\kappa}{N} \int_{\Lambda\times \Lambda} dx dy \, N^{3\beta} V(N^\beta (x-y)) \check{\eta} (x-y) \langle \text{D}_2^* \, \xi ,  \check{a}_x \check{a}_y \xi \rangle  \Big|\\ &\leq \frac{\kappa}{N} \int_{\Lambda \times \Lambda} dx dy N^{3\beta} V(N^\beta (x-y)) |\check{\eta} (x-y)| \| \check{a}_x \check{a}_y \xi \| \| \text{D}_2^* \, \xi \| \\ &\leq \frac{C^{n+k} \kappa^{n+k+2}}{N^{2-\beta}}  \| (\cN_+ +1) \xi \| \int_{\Lambda \times \Lambda}  dx dy N^{3\beta} V(N^\beta (x-y)) \| \check{a}_x \check{a}_y \xi \| 
\\ &\leq C^{n+k}  \kappa^{n+k+3/2}  N^{\beta-1} \| (\cN_+ +1)^{1/2} \xi \| \| \cV_N^{1/2} \xi \| \end{split} \] 
Combining the last two bounds with (\ref{eq:X1}), and then summing over all $n,k$, we find 
\begin{equation}\label{eq:wtW1-fin} |\langle \xi , \wt{W}_1 \xi \rangle | \leq C \kappa^2 N^{\beta-1} \| (\cN_+ +1)^{1/2} \xi \|^2 + C \kappa^{3/2} N^{\beta-1} \| (\cN_+ +1)^{1/2} \xi \| \| \cV_N^{1/2}\xi \| \, . \end{equation}

Next, we consider the expectation of the operator $\text{W}_2$, defined in (\ref{eq:defW}). Since we will need the potential energy operator to bound this term, it is convenient to switch to position space. On $\cF_+$, we find
\begin{equation} \begin{split} \text{W}_2 = \; &\sum_{n,k= 0}^\infty \frac{(-1)^{n+k}}{n!k!(n+k+1)} \\ & \times \frac{\kappa}{N} \int_{\Lambda \times \Lambda} dx dy N^{3\beta} V(N^\beta(x-y)) \left( \text{ad}^{(n)}_{B(\eta)} (\check{b}^*_x) \text{ad}_{B(\eta)}^{(k)} (\check{b}^*_y) a^* (\check{\eta}_x) \check{a}_y  + \text{h.c.} \right) \end{split} 
\end{equation}
with the notation $\check{\eta}_x (z) = \check{\eta} (x-z)$. With  Cauchy-Schwarz, we find
\begin{equation}\label{eq:star1} \begin{split} 
\Big| \frac{\kappa}{N} \int_{\Lambda \times \Lambda} dx dy \, &N^{3\beta} V(N^\beta (x-y)) \langle \xi , \text{ad}^{(n)}_{B(\eta)} (\check{b}_x^*) \text{ad}^{(k)}_{B(\eta)} (\check{b}_y^*) a^* (\check{\eta}_x) \check{a}_y \xi \rangle \Big| \\ 
&\leq \frac{\kappa}{N}  \int_{\Lambda \times \Lambda} dx dy \, N^{3\beta} V(N^\beta (x-y)) \\ &\hspace{1cm} \times  \| (\cN_+ +1)^{1/2} \text{ad}^{(k)}_{B(\eta)} (\check{b}_y) \text{ad}^{(n)}_{B(\eta)} (\check{b}_x) \xi \| \| (\cN_+ +1)^{-1/2} a^* (\check{\eta}_x) \check{a}_y \xi \| \end{split} \end{equation}
We bound
\begin{equation}\label{eq:star2} \| (\cN_+ +1)^{-1/2} a^* (\check{\eta}_x) \check{a}_y \xi \| \leq C \kappa \| \check{a}_y \xi \|  \end{equation}
With Lemma \ref{lm:indu}, we estimate $\| (\cN_+ +1)^{1/2} \text{ad}^{(k)}_{B(\eta)} (\check{b}_y) \text{ad}^{(n)}_{B(\eta)} (\check{b}_x) \xi \|$ by the sum of $2^{n+k} n!k!$ terms of the form 
\begin{equation}\label{eq:norm-T} \begin{split} 
\text{Z} = &\left\| (\cN_+ +1)^{1/2} \Lambda_1 \dots \Lambda_{i_1} N^{-k_1} \Pi^{(1)}_{\sharp, \flat} (\eta^{j_1}, \dots , \eta^{j_{k_1}} ; \check{\eta}_{y}^{\ell_1}) \right. \\ &\hspace{3cm} \left. \times  \Lambda'_1 \dots \Lambda'_{i_2} N^{-k_2} \Pi^{(1)}_{\sharp,\flat} (\eta^{m_1} , \dots, \eta^{m_{k_2}} ; \check{\eta}^{\ell_2}_{x}) \xi \right\| \end{split} \end{equation}
with $i_1, i_2, k_1, k_2, \ell_1, \ell_2 \geq 0$, $j_1, \dots , j_{k_1}, m_1, \dots , m_{k_2} \geq 0$ and where each $\Lambda_i$ and $\Lambda'_i$ operator is either a factor $(N-\cN_+ )/N$, $(N+1-\cN_+ )/N$ or a $\Pi^{(2)}$-operator of the form (\ref{eq:Pi2-E4}) (here $\check{\eta}^{\ell_1}$ indicates the function with Fourier coefficients given by $\eta^{\ell_1}_p$, for all $p \in \Lambda^*_+$). With Lemma \ref{lm:aux2}, we find 
\begin{equation}\label{eq:Tf} \begin{split} \text{Z} &\leq (n+1) C^{k+n} \kappa^{k+n} \Big\{ \| (\cN_+ +1)^{3/2} \xi \| + \| \check{a}_y (\cN_+ +1) \xi \| + \| \check{a}_x (\cN_+ +1) \xi \| \\ &\hspace{6cm} + N^\beta \| (\cN_+ +1)^{1/2} \xi \| + \sqrt{N} \| \check{a}_x \check{a}_y \xi \| \Big\}  \end{split} \end{equation}
Inserting (\ref{eq:star2}) and (\ref{eq:Tf}) into (\ref{eq:star1}) we obtain, for any $\xi \in \cF_+^{\leq N}$, 
\[ \begin{split} 
&\left| \frac{\kappa}{N} \int_{\Lambda \times \Lambda} dx dy N^{3\beta} V(N^\beta (x-y)) \langle \xi,  \text{ad}^{(n)}_{B(\eta)} (\check{b}_x^*) \text{ad}^{(k)}_{B(\eta)} (\check{b}^*_y)  a^* (\check{\eta}_x) \check{a}_y \xi \rangle \right| \\ &\hspace{1cm} 
\leq \frac{(n+1)!k!C^{n+k} \kappa^{n+k+2}}{N} \int dx dy \, N^{3\beta} V(N^\beta (x-y)) \| \check{a}_y \xi \| \\ &\hspace{3cm} \times \Big\{ N \| (\cN_+ +1)^{1/2} \xi \| + N \| \check{a}_y \xi \| + N \| \check{a}_x \xi \| + \sqrt{N} \| \check{a}_x \check{a}_y \xi \| \Big\}  \\
&\hspace{1cm} \leq (n+1)! k! C^{n+k} \kappa^{n+k+2} \| (\cN_+ +1)^{1/2} \xi \|^2 \\ &\hspace{1.5cm} + (n+1)!k! C^{n+k} \kappa^{n+k+3/2} \| (\cN_+ +1)^{1/2} \xi \| \| \cV_N^{1/2} \xi \| \end{split} \]
Therefore, if $\kappa > 0$ is small enough, we find, for every $\delta > 0$, a constant $C > 0$ such that 
\begin{equation}\label{eq:W2-end0} |\langle \xi ,\text{W}_2 \xi \rangle | \leq \delta \| \cV_N^{1/2} \xi \|^2 + C \kappa^2 \| (\cN_+ +1)^{1/2} \xi \|^2  \end{equation}
On the other hand, inserting (\ref{eq:star2}) and (\ref{eq:Tf}) into (\ref{eq:star1}), we also arrive at
\[ \begin{split} 
&\left| \frac{\kappa}{N} \int_{\Lambda \times \Lambda} dx dy N^{3\beta} V(N^\beta (x-y)) \langle \xi,  \text{ad}^{(n)}_{B(\eta)} (\check{b}_x^*) \text{ad}^{(k)}_{B(\eta)} (\check{b}^*_y)  a^* (\check{\eta}_x) \check{a}_y \xi \rangle \right| \\ &\hspace{1cm} \leq  \frac{(n+1)!k! \, C^{n+k} \kappa^{n+k+2}}{N}  \int_{\Lambda \times \Lambda}  dx dy \, N^{3\beta} V(N^\beta (x-y))  \| \check{a}_y \xi \| \\ &\hspace{1.5cm} \times \Big\{ (\sqrt{N} + N^\beta) \| (\cN_+ +1) \xi \| +   \sqrt{N} \| \check{a}_x (\cN_+ +1)^{1/2} \xi \| \\ &\hspace{4cm} + \sqrt{N} \| \check{a}_y (\cN_+ +1)^{1/2} \xi \| + \sqrt{N} \,  \| \check{a}_x \check{a}_y \xi \| \Big\} 
\\ &\hspace{1cm} \leq (n+1)! k! \, C^{n+k} \kappa^{n+k+2}  N^{-\min (1-\beta, 1/2)} \| (\cN_+ +1) \xi \|^2 \\ &\hspace{1.5cm} + (n+1)! k! \, C^{n+k} \kappa^{n+k+3/2}  \| (\cN_+ +1)^{1/2} \xi \| \| \cV_N^{1/2} \xi \|
 \end{split} \] 
Therefore, for $\kappa > 0$ small enough and using  Lemma \ref{lm:cVN} we obtain 
\begin{equation}\label{eq:W2-end}
|\langle \xi , \text{W}_2 \xi \rangle | \leq  C N^{(\beta-1)/2}   \| (\cN_+ +1)^{1/2} (\cH_N^\beta +1)^{1/2}  \xi \|^2 .  \end{equation}

Next, let us consider the term $\text{W}_3$, defined in (\ref{eq:defW2}). As above, we switch to position space. We find
\begin{equation}\label{eq:W3}
\begin{split}  \text{W}_3 = &\; \sum_{n,k,i,j = 0}^\infty \frac{(-1)^{n+k+i+j}}{n!k!i!j! (i+j+1)(n+k+i+j+2)} \\ & \times \frac{\kappa}{N} \int dx dy \, N^{3\beta} V(N^\beta (x-y)) \\ &\hspace{1cm} \times  \left( \text{ad}^{(n)} (\check{b}^*_x) \text{ad}^{(k)}_{B(\eta)} (\check{b}^*_y) \text{ad}_{B(\eta)}^{(i)} (b^* (\check{\eta}_x)) \text{ad}^{(j)}_{B(\eta)} (b^* (\check{\eta}_y)) + \text{h.c.} \right) \end{split} \end{equation}
With Cauchy-Schwarz, we have
\[ \begin{split} & \left| \frac{\kappa}{N} \int dx dy N^{3\beta} V(N^\beta (x-y)) \langle \xi , \text{ad}^{(n)}_{B(\eta)} (\check{b}_x^*) \text{ad}^{(k)}_{B(\eta)} (\check{b}_y^*) \text{ad}^{(i)}_{B(\eta)} (\check{b}^* (\check{\eta}_x)) \text{ad}^{(j)}_{B(\eta)}  (\check{b} (\check{\eta}_y)) \xi \rangle \right| \\
&\hspace{2cm} \leq \frac{\kappa}{N} \int dx dy \, N^{3\beta} V(N^\beta (x-y)) \, \| (\cN_+ +1)^{1/2} \text{ad}^{(k)}_{B(\eta)} (\check{b}_y)  \text{ad}^{(n)}_{B(\eta)} (\check{b}_x) \xi \| \\ &\hspace{5cm} \times  \|  (\cN_+ +1)^{-1/2} 
\text{ad}^{(i)}_{B(\eta)} (b (\check{\eta}_x)) \text{ad}^{(j)} (b (\check{\eta}_y)) \xi \| 
\end{split} \]
Expanding $\text{ad}^{(i)}_{B(\eta_t)} (b (\check{\eta}_x)) \text{ad}^{(j)} (b (\check{\eta}_y))$ as in Lemma \ref{lm:indu} and using Lemma \ref{lm:aux2}, we obtain  
\begin{equation}\label{eq:W3-1} \begin{split} 
\| (\cN_+ +1)^{-1/2} &\text{ad}^{(i)}_{B(\eta)} (b (\check{\eta}_x)) \text{ad}^{(j)} (b (\check{\eta}_y)) \xi \| \leq i!j! \, C^{i+j} \kappa^{i+j+2}   \| (\cN_+ +1)^{1/2} \xi \|  \end{split} \end{equation}
As for the norm $\| (\cN_+ +1)^{1/2} \text{ad}^{(k)}_{B(\eta)} (\check{b}_y)  \text{ad}^{(n)}_{B(\eta)} (\check{b}_x) \xi \|$, we can estimate by the sum of $2^{n+k} n!k!$ contributions of the form (\ref{eq:norm-T}). With the bound (\ref{eq:Tf}), we can argue as in the analysis of the term $\text{W}_2$. Similarly to (\ref{eq:W2-end}) and (\ref{eq:W2-end0}), we conclude that, if $\kappa >0$ is sufficiently small, for every $\delta > 0$, there exists $C > 0$ such that 
\begin{equation}\label{eq:W3end}
\begin{split}  
|\langle \xi , \text{W}_3 \xi \rangle | &\leq \delta \| \cV_N^{1/2} \xi \|^2 + C \kappa^2 \| (\cN_+ +1)^{1/2} \xi \|^2 \\ 
|\langle \xi , \text{W}_3 \xi \rangle | &\leq C N^{(\beta-1)/2} \| (\cN_+ +1)^{1/2} (\cH_N^\beta + 1)^{1/2} \xi \|^2 \end{split} \end{equation}

The term $\text{W}_4$ in (\ref{eq:defW2}) can be bounded similarly. First, we switch to position space: 
\begin{equation}\label{eq:W4} \begin{split} \text{W}_4 = \; &\sum_{n,k,i,j =0}^\infty \frac{(-1)^{n+k+i+j}}{n!k!i!j! (i+j+1) (n+k+i+j+2)} \\ &\times \frac{\kappa}{N} \int dxdy \, N^{3\beta} V(N^\beta (x-y)) \, \left( \text{ad}^{(n)} (\check{b}_x) \text{ad}^{(k)} (\check{b}_y) \text{ad}^{(i)} (b (\check{\eta}^2_x)) \text{ad}^{(j)} (\check{b}_y) + \text{h.c.} \right) \end{split} \end{equation}
The expectation of the operators on the r.h.s. of (\ref{eq:W4}) can be bounded similarly as we did for the operators on the r.h.s. of (\ref{eq:W3}). The only difference is the fact that now we have to replace the estimate (\ref{eq:W3-1}) with 
\[ \| (\cN_+ +1)^{-1/2} \text{ad}^{(i)} (b (\check{\eta}^2_x)) \text{ad}^{(j)} (\check{b}_y) \xi \| \leq i! j! C^{i+j} \kappa^{i+j+2}  \left[ \| (\cN_+ +1)^{1/2} \xi \| + \| \check{a}_y \xi \| \right] \] 
Hence, we obtain that, for every $\delta > 0$, 
\begin{equation}\label{eq:W4end}
\begin{split}  
|\langle \xi , \text{W}_4 \xi \rangle | &\leq \delta \| \cV_N^{1/2} \xi \|^2 + C \kappa^2 \| (\cN_+ +1)^{1/2} \xi \|^2 \\ 
|\langle \xi , \text{W}_4 \xi \rangle | &\leq C N^{(\beta-1)/2} \| (\cN_+ +1)^{1/2} (\cH_N^\beta + 1)^{1/2} \xi \|^2 \, . \end{split} \end{equation}

Combining the bounds (\ref{eq:Tfin}), (\ref{eq:wtW1-fin}), (\ref{eq:W2-end0}), (\ref{eq:W2-end}), (\ref{eq:W3end}) and (\ref{eq:W4end}) we conclude, by (\ref{eq:cEN4-in}),  that, for all $\delta > 0$ there is $C > 0$ such that 
\begin{equation}\label{eq:cE4-fin0} \pm \cE^{(4)}_{N,\beta} \leq \delta \cV_N + C \kappa (\cN_+ +1) \end{equation}
and that, furthermore,
\[ \pm \cE^{(4)}_{N,\beta} \leq C N^{(\beta-1)/2} (\cN_+ +1) (\cH_N^\beta +1) \]
As usual, the bound for the commutator of $\cE^{(4)}_{N,\beta}$ with $\cN_+ $ can be proven exactly as we proved (\ref{eq:cE4-fin0}).
\end{proof}

\subsection{Proof of Theorem \ref{thm:gene}}

Combining the results of Prop. \ref{prop:G0}, Prop. \ref{prop:K}, Prop. \ref{prop:V}, Prop. \ref{prop:G3} and Prop. \ref{prop:G4}, we conclude that 
\[ \begin{split} 
\cG^\beta_{N} = \; &\frac{(N-1)}{2} \kappa \widehat{V} (0) + \sum_{p \in \Lambda^*_+} \left[ p^2 \eta_p^2 + \kappa \widehat{V} (p/N^\beta) \eta_p + \frac{1}{2N} \sum_{q \in \Lambda^*_+} \widehat{V} ((p-q)/N^\beta) \eta_p \eta_q \right] \\ &+ \cH_N^\beta + \sum_{p \in \Lambda^*_+} \left[ p^2 \eta_p + \frac{\kappa \widehat{V} (p/N^\beta)}{2} + \frac{1}{2N} \sum_{q \in\Lambda^*_+} \widehat{V} ((p-q)/N^\beta) \eta_q \right] \big[ b_p^* b_{-p}^* + b_p b_{-p} \big] \\ &+ \wt{\cE}_{N,\beta}
\end{split} \]
where the error $\wt{\cE}_{N,\beta} = \cE^{(0)}_{N,\beta} + \wt{\cE}_{N,\beta}^{(K)} + \wt{\cE}^{(V)}_{N,\beta} + \cG_{N,\beta}^{(3)} + \cE_{N,\beta}^{(4)}$ is such that, for every $\delta > 0$, there exists $C > 0$ with 
\[ \begin{split} \pm \wt{\cE}_{N,\beta} &\leq \delta \cH_N^\beta + C \kappa (\cN_+ +1) \\ \pm [\wt{\cE}_{N,\beta} , i \cN_+  ] &\leq C (\cH_N^\beta + 1) \end{split} \]
Using the relation (\ref{eq:eta-scat}), we can rewrite
\begin{equation}\label{eq:GNb-fin} \begin{split} 
\cG^\beta_{N} = \; &\frac{(N-1)}{2} \kappa \widehat{V} (0) + \frac{\kappa}{2}  \sum_{p\in \Lambda^*_+} \widehat{V} (p/N^\beta) \eta_p  + \cH_N^\beta + \cE'_{N,\beta} 
\end{split} \end{equation}
with 
\[ \begin{split} \cE'_{N,\beta} = \; &
\sum_{p \in \Lambda^*_+} \eta_p \Big[ N\lambda_{N,\ell} \widehat{\chi}_\ell (p) + \lambda_{N,\ell} \sum_{q \in \Lambda^*} \widehat{\chi}_\ell (p-q) \wt{\eta}_q - \frac{\kappa}{2N} \widehat{V} (p/N^\beta) \wt{\eta}_0 \Big] \\ &+ \sum_{p \in \Lambda^*_+} \Big[ N \lambda_{N,\ell} \widehat{\chi}_\ell (p) + \lambda_{N,\ell} \sum_{q \in \Lambda^*} \widehat{\chi}_\ell (p-q) \wt{\eta}_q - \frac{\kappa}{2N} \widehat{V} (p/N^\beta) \wt{\eta}_0  \Big] \big[ b_p^* b_{-p}^* + b_p b_{-p} \big] \\ 
&+ \wt{\cE}_{N,\beta} \end{split} \]
Since, by Lemma \ref{3.0.sceqlemma}, $N\lambda_{N,\ell} \leq C\kappa$ uniformly in $N$, since $|\wt{\eta}_q| \leq C \kappa / (|q|^2 +1)$ from (\ref{eq:etap}) and  (\ref{eq:wteta0}) and since $|\widehat{\chi}_\ell (p)| \leq C |p|^{-2}$ (see (\ref{eq:chip})), we conclude easily that for every $\delta > 0$ there exists $C >0$ such that 
\begin{equation}\label{eq:cE'-fin} \begin{split} \pm \cE'_{N,\beta} &\leq \delta \cH_N^\beta + C \kappa (\cN_+ +1) \\ 
\pm \big[ \cE'_{N,\beta} , i \cN_+  \big] &\leq C (\cH_N^\beta +1) \end{split} \end{equation}
Eq. (\ref{eq:GNb-fin}) implies, in particular, that the ground state energy of the Hamiltonian (\ref{eq:Ham0}) is such that 
\[ E_N^\beta \leq \langle \Omega, \cG_N^\beta \Omega \rangle \leq \frac{(N-1)}{2} \kappa \widehat{V} (0) + \frac{\kappa}{2} \sum_{p \in \Lambda^*_+} \widehat{V} (p/N^\beta) \eta_p  + C \]
for a constant $C > 0$ independent of $N$. Inserting in (\ref{eq:GNb-fin}) and using the first bound in  (\ref{eq:cE'-fin}) (taking for example $\delta = 1/4$)  we conclude that, for $\kappa$ small enough,  
\[ \cG_N^\beta - E_N^\beta \geq  \frac{1}{2} \cH_N^\beta - C \]
Furthermore, (\ref{eq:GNb-fin}) and the second bound in (\ref{eq:cE'-fin}) immediately give 
\[ \pm \big[ \cG_N^\beta , i \cN_+  \big] \leq C (\cH_N^\beta + 1) \]
which concludes the proof of part a) of Theorem \ref{thm:gene}. To show part b), we notice that Prop. \ref{prop:G0}, Prop. \ref{prop:K}, Prop. \ref{prop:V}, Prop. \ref{prop:G3} and Prop. \ref{prop:G4} also imply that
\[ \begin{split} \cG_N^\beta = \; &\frac{(N-1)}{2} \kappa \widehat{V} (0)  \\ &+ \sum_{p \in \Lambda^*_+} \Big[ p^2 \sigma_p^2 + \kappa \widehat{V} (p/N^\beta) (\sigma_p^2 +\sigma_p \gamma_p) + \frac{\kappa}{2N} \sum_{q \in \Lambda^*_+} \widehat{V} ((p-q)/N^\beta) \eta_p \eta_q \Big] \\ &+ \cH_N^\beta +  \sum_{p \in \Lambda^*_+} \big[ 2p^2 \sigma_p^2 + \kappa \widehat{V} (p/N^\beta) (\gamma_p + \sigma_p)^2 \big] b_p^* b_p \\ &+ \sum_{p \in \Lambda^*_+} \Big[ p^2 \sigma_p \gamma_p + \frac{\kappa}{2} \widehat{V} (p/N^\beta) (\gamma_p + \sigma_p)^2 + \frac{\kappa}{2N} \sum_{q \in \Lambda^*_+} \widehat{V} ((p-q)/N^\beta) \eta_q \Big] \\ &\hspace{9cm} \times (b_p^* b_{-p}^* + b_p b_{-p} ) \\ &+ \widehat{\cE}_{N}^\beta \end{split} \]
where the error term $\widehat{\cE}^\beta_{N} = \cE_{N,\beta}^{(0)} + \cE_{N,\beta}^{(K)} + \cE_{N,\beta}^{(V)} + \cG_{N,\beta}^{(3)} + \cE^{(4)}_{N,\beta}$ is such that 
\begin{equation}\label{eq:cE'} \pm \widehat{\cE}_{N}^\beta \leq C N^{(\beta-1)/2} (\cN_+ +1)(\cK + 1)
\end{equation}
Comparing with (\ref{eq:CN}) and (\ref{eq:FpGp}), we obtain that 
\[ \cG_N^\beta = C_N^\beta + \cQ_N^\beta + \cE_N^\beta \]
with
\begin{equation}\label{eq:cEN-fin} \cE^\beta_N = \widehat{\cE}^\beta_N + \cV_N + \frac{\kappa \wt{\eta}_0}{2N} \sum_{p \in \Lambda^*_+}  \widehat{V} (p/N^\beta) \wt{\eta}_0 (b_p^* b_{-p}^* + b_p b_{-p}) \end{equation}
Switching to position space, we have 
\[  \frac{\kappa}{2N} \sum_{p\in \Lambda^*_+} \widehat{V} (p/N^\beta) \wt{\eta}_0 \langle \xi, b_p b_{-p} \xi \rangle = \frac{\kappa \wt{\eta}_0}{2N} \int_{\Lambda \times \Lambda} dx dy \, N^{3\beta} V(N^\beta (x-y)) \langle \xi , \check{b}_x \check{b}_y \xi \rangle \]
Since $|\wt{\eta}_0| \leq C$  from (\ref{eq:wteta0}), we find
\[ \begin{split} \Big| \frac{\kappa}{2N} \sum_{p\in \Lambda^*_+} \widehat{V} (p/N^\beta) \wt{\eta}_0 \langle \xi, b_p b_{-p} \xi \rangle \Big| &\leq \frac{C}{N} \int_{\Lambda \times \Lambda} dx dy \, N^{3\beta} V(N^\beta (x-y)) \| \check{a}_x \check{a}_y \xi \| \| \xi \| \\ &\leq C N^{-1/2} \,  \| \cV_N^{1/2} \xi \| \| \xi \| \\ & \leq C N^{-1/2} \| (\cN_+ +1)^{1/2} (\cK +1)^{1/2} \xi \|^2 
\end{split} \]
where we used Lemma \ref{lm:cVN}. Combining the last estimate with (\ref{eq:cE'}) and again with Lemma \ref{lm:cVN}, Eq. (\ref{eq:cEN-fin}) implies that  
\[ \pm \cE^\beta_N \leq C N^{(\beta-1)/2} (\cN_+ +1) ( \cK + 1) \]
This completes the proof of part b) of Theorem \ref{thm:gene}.

\end{document}